\documentclass[transmag]{IEEEtran}
\usepackage{latexsym}
\usepackage{graphicx}
\usepackage{bm}
\usepackage{amsfonts,amssymb,amsmath,amsthm}
\usepackage{mathtools}
\theoremstyle{remark}
\newtheorem{theorem}{Theorem}
\newtheorem{lemma}{Lemma}

\newtheorem{remark}{Remark}

\usepackage{cite}
\usepackage{hyperref}
\hypersetup{linktocpage} 
\hypersetup{
	colorlinks,
	citecolor=black,
	filecolor=black,
	linkcolor=black,
	urlcolor=black
}
\usepackage[linesnumbered,ruled,vlined]{algorithm2e}

\usepackage[font=footnotesize,labelfont=bf, figurename=Fig.]{caption} 
\usepackage{subcaption}
\newcommand{\norm}[1]{\left\lVert#1\right\rVert}
\DeclareMathOperator*{\argmax}{arg\,max}

\def\BibTeX{{\rm B\kern-.05em{\sc i\kern-.025em b}\kern-.08em T\kern-.1667em\lower.7ex\hbox{E}\kern-.125emX}}
\begin{document}
\title{Full-Duplex Cell-Free Massive MIMO Systems: Analysis and Decentralized Optimization}
\author{Soumyadeep Datta, Dheeraj Naidu Amudala, Ekant Sharma, Rohit Budhiraja \\ and Shivendra S. Panwar,
\IEEEmembership{Fellow, IEEE}
\thanks{Date of submission: 22 November, 2021. Part of this work was published in IEEE ICC 2021. The conference details are as follows:
S. Datta, E. Sharma, D. N. Amudala, R. Budhiraja and S. S. Panwar, ``FD Cell-Free mMIMO: Analysis and Optimization," ICC 2021 - IEEE International Conference on Communications, 2021, pp. 1-6, doi: 10.1109/ICC42927.2021.9500917.}
\thanks{This work was supported by Visvesvaraya Ph.D. Scheme, MeitY, Govt. of India MEITY-PHD-2721.}
\thanks{Soumyadeep Datta is jointly with the New York University Tandon School of Engineering, Brooklyn, NY 11201, USA, and the Indian Institute of Technology Kanpur, Uttar Pradesh 208016, India (e-mail: sdatta@nyu.edu, sdatta@iitk.ac.in).}
\thanks{Dheeraj Naidu Amudala is with the Indian Institute of Technology Kanpur, Uttar Pradesh 208016, India (e-mail: dheeraja@iitk.ac.in).}
\thanks{Ekant Sharma is with the Indian Institute of Technology Roorkee, Uttarakhand 247667, India (e-mail: ekant@ece.iitr.ac.in).}
\thanks{Rohit Budhiraja is with the Indian Institute of Technology Kanpur, Uttar Pradesh 208016, India (e-mail: rohitbr@iitk.ac.in).}
\thanks{Shivendra S. Panwar is with the New York University Tandon School of Engineering, Brooklyn, NY 11201, USA (e-mail: panwar@nyu.edu).}
}
\IEEEtitleabstractindextext{
\begin{abstract}
Cell-free (CF) massive multiple-input-multiple-output (mMIMO) deployments are usually investigated with  half-duplex nodes and  high-capacity fronthaul links. To leverage the possible gains in throughput and energy efficiency (EE) of full-duplex (FD) communications, we consider a  FD CF mMIMO system with \textit{practical limited-capacity fronthaul links}.  We derive closed-form spectral efficiency (SE) lower bounds for this system with maximum-ratio combining/maximum-ratio transmission processing and optimal uniform quantization. We then optimize the weighted sum EE (WSEE) via downlink and uplink power control by using a {two-layered} approach: the {first layer} formulates the optimization  as a generalized convex program, while the {second layer} solves the optimization decentrally using the alternating direction method of multipliers. We analytically show that the proposed two-layered formulation yields a Karush-Kuhn-Tucker point of the original WSEE optimization. We numerically show the  influence of weights on the individual EE of the users, which demonstrates the utility of the WSEE metric to  incorporate heterogeneous EE requirements of users.  We show that low fronthaul capacity reduces the number of users each AP can support,  and the cell-free system, consequently, becomes  user-centric.
\end{abstract}
\vspace{-0.4cm}
}
\maketitle
\section{INTRODUCTION}
Massive multiple-input-multiple-output (mMIMO) wireless systems employ a large number of antennas at the base stations (BSs), and achieve higher spectral efficiency (SE) and energy efficiency (EE) with relatively simple signal processing~\cite{MassiveMIMO1,CellFreeMassiveMIMOBook}.  Two distinct  mMIMO variants are being investigated in the literature: i) co-located, wherein all antennas are located at one place~\cite{MassiveMIMO1}; and ii) distributed, wherein  antennas are spread over a large area~\cite[and the references therein]{CellFreeMassiveMIMOBook},\cite{CFvsSmallCells,NgoCellFree2,9406061}. While co-located mMIMO systems have a low fronthaul requirement, distributed mMIMO systems, at the cost of higher fronthaul infrastructure, have greater spatial diversity to exploit and consequently have greater immunity to shadow fading~\cite{CellFreeMassiveMIMOBook,CFvsSmallCells, NgoCellFree2}.
Cell-free (CF) mMIMO is one of the most promising distributed mMIMO variants in the current literature~\cite{9406061,CellFreeMassiveMIMOBook,CFvsSmallCells, NgoCellFree2}. CF mMIMO envisions a communication region with no cell boundaries, and  promises substantial gains in SE and fairness over small-cell deployments~\cite{9406061,CFvsSmallCells,NgoCellFree2}. 

Full-duplex (FD) wireless systems have now been practically realized with advanced self-interference (SI) cancellation mechanisms~\cite{LIC, LIC2, FullDuplex1, FullDuplex2}. Co-located FD massive MIMO systems have also been extensively investigated~\cite[and the referencestherein]{ChoiFD,LiuFDSmallCell}. FD CF mMIMO is a relatively recent area of interest~\cite{FDCellFree,NAFDCellFree,FDCellFree3}, where access points (APs) simultaneously serve downlink and uplink user equipments (UEs) on the same spectral resource. Vu \textit{et al.} in~\cite{FDCellFree} considered a FD CF mMIMO system with maximum-ratio combining and showed that if SI at the APs is suppressed up to a certain limit, it has higher throughput than its half-duplex (HD) counterpart and FD co-located systems. Wang \textit{et al.} in~\cite{NAFDCellFree} evaluated the SE of a network-assisted FD CF mMIMO system using zero-forcing and regularized zero-forcing  beamforming. Reference~\cite{FDCellFree3} proposed a heap-based algorithm for pilot assignment to overcome pilot contamination in FD CF mMIMO systems.

In CF mMIMO, APs are connected to a central processing unit (CPU) using fronthaul links. The existing FD CF mMIMO literature assumes  high-capacity fronthaul links~\cite{FDCellFree,NAFDCellFree,FDCellFree3}. These links, however, have limited capacity, and the information needs to be consequently quantized and sent over them. The limited-capacity fronthaul has been considered only for HD CF mMIMO systems in \cite{CellFreeMaxMinUQ,CellFreeLimFronthaul,CellFreeLimFronthaul2}.  Femenias \textit{et al.} in \cite{CellFreeLimFronthaul} studied a max-min uplink/downlink power allocation problem for HD CF mMIMO with limited-capacity fronthaul, while Masoumi \textit{et al.} in~\cite{CellFreeLimFronthaul2} optimized the SE of a HD CF mMIMO uplink with limited-capacity  fronthaul and hardware impairments. Bashar \textit{et al.} in~\cite{CellFreeMaxMinUQ} derived the SE of HD CF mMIMO uplink with limited-capacity fronthaul. We  consider quantized  fronthaul for a FD CF mMIMO system to derive achievable SE expressions. To the best of our knowledge, the current work is first one to do so.

With tremendous increase in network traffic, the EE has become an important metric to design a modern wireless system. Global energy efficiency (GEE), defined as the ratio of the network SE and its total energy consumption, is being used to design CF mMIMO communication systems \cite{CellFreeEE,CellFreeEEUQ,CellFreeEEmmW,FDCellFree2}. Ngo \textit{et al.} in~\cite{CellFreeEE} optimized the GEE for the downlink  of a HD CF mMIMO system.  Bashar \textit{et al.} in~\cite{CellFreeEEUQ} optimized the uplink GEE of a HD CF mMIMO system with optimal uniform fronthaul quantization. Alonzo \textit{et al.} in~\cite{CellFreeEEmmW} optimized the GEE of CF and UE-centric HD  mMIMO deployments in the mmWave regime. Nguyen \textit{et al.} in~\cite{FDCellFree2} maximized a novel SE-GEE metric for the FD CF mMIMO system using a Dinkelbach-like algorithm.

A UE with limited energy availability will accord a much higher importance to its EE than an another UE with a sufficient energy supply. GEE is a network-centric metric and cannot accommodate such heterogeneous EE requirements~\cite{fractionalprogrammingbook}. The weighted sum energy efficiency (WSEE) metric, defined as the weighted sum of individual EEs~\cite{fractionalprogrammingbook}, can prioritize EEs of individual UEs,  by allocating them a higher weight~\cite{WSEEEfrem,WSEEEkant}.
The WSEE  is investigated in \cite{WSEEEfrem} for a general wireless network, and  for a two-way FD relay in~\cite{WSEEEkant}. It is yet to be investigated for CF mMIMO HD and FD systems.

Decentralized designs, which  accomplish a complex task by coordination and cooperation of a set of computing units, are being used to design mMIMO systems~\cite{Decentralized1,Decentralized3}. This interest is driven by high computational complexity and high interconnection data rate requirements between radio frequency  chains and baseband units  in centralized mMIMO system designs~\cite{Decentralized1}. Jeon \textit{et al.} in~\cite{Decentralized1} constructed decentralized  equalizers by partitioning the BS antenna array. 
Reference~\cite{Decentralized3} proposed a coordinate-descent-based decentralized  algorithm for mMIMO uplink detection and downlink precoding.  Reference \cite{ADMM2} employed alternating direction method of multipliers (ADMM) to decentrally allocate edge-computing resource for vehicular networks. 
Such decentralized approaches have not yet been employed to optimize FD CF mMIMO systems. We next list our \textbf{main} contributions in this context:\newline
{1) \textbf{Contributions regarding closed form SE lower bound:} We consider FD CF mMIMO communications with maximal ratio combining/maximal ratio transmission (MRC)/(MRT) processing and limited fronthaul with optimal uniform quantization. We note that for the FD CF mMIMO systems, \textit{unlike their HD counterparts}~\cite{CellFreeMassiveMIMOBook,CFvsSmallCells,NgoCellFree2,9406061}, uplink and downlink transmissions interfere to cause \textit{uplink downlink interference (UDI)} and \textit{inter-/intra-AP residual interference (RI)}.  Further, unlike existing FD CF mMIMO literature~\cite{FDCellFree,NAFDCellFree,FDCellFree3,FDCellFree2}, which consider perfect high-capacity fronthaul links, it is critical to model and analyze the UDI and inter-/intra-AP interferences and limited-capacity impairments while deriving lower bounds for both uplink and downlink UEs SE, which are valid for arbitrary number of antennas at each AP. We model the UDI on the downlink and the RI on the uplink, but unlike existing FD CF mMIMO literature~\cite{FDCellFree,NAFDCellFree,FDCellFree3,FDCellFree2}, we also consider the quantization distortion due to limited-capacity fronthaul links, as modelled in the total quantization distortion (TQD) terms. We also show the impact of quantization on the uplink RI terms themselves, where the distortion in the downlink and uplink signals get coupled. We derive achievable SE expressions for both  uplink and downlink UEs, which are valid for arbitrary number of antennas at each AP.} \newline
{2) \textbf{Contributions regarding centralized WSEE optimization:} 
We use the derived SE expression to maximize the non-convex WSEE metric. While energy-efficient design of CF mMIMO systems have been studied in literature~\cite{CellFreeEE,CellFreeEEUQ,CellFreeEEmmW}, most of them focus on the GEE metric, except reference~\cite{FDCellFree2}. The GEE, being a single ratio, can be expressed as a pseudo-concave (PC) function and can thus be maximized using Dinkelbach's algorithm~\cite{fractionalprogrammingbook}. Reference~\cite{FDCellFree2} is the only work so far which optimized the EE of FD CF mMIMO. It considered a novel SE-GEE objective, which also reduces to a PC function and is maximized using a Dinkelbach-like algorithm. \textit{The WSEE, in contrast, is a sum of PC functions, and is not guaranteed to be a PC function}~\cite{fractionalprogrammingbook}. This makes the WSEE an extremely non-trivial objective to maximize~\cite{fractionalprogrammingbook}. Further, the algorithm in~\cite{FDCellFree2} requires  knowledge of instantaneous small-scale channel fading coefficients. The WSEE metric optimized here, in contrast,  requires large-scale channel coefficients, which remains constant for multiple coherence intervals~\cite{AGWireless}.}\newline
{3) \textbf{Contributions regarding decentralized optimization:} We decentrally maximize WSEE using a {two-layered} iterative approach which combines successive convex approximation (SCA) and ADMM.  The {first layer} simplifies the non-convex WSEE maximization problem by using epigraph transformation, slack variables and  series approximations. It then locally approximates the problem as a generalized convex program (GCP) which is solved iteratively using the SCA approach.  The  second layer decentrally optimizes the GCP by using the consensus ADMM approach, which decomposes the centralized version into multiple sub-problems, each of which is solved independently. The local solutions are combined to obtain the global  solution. We note that the GCP {for the FD system} is not in the standard form which is required for applying ADMM, as it involves FD interference terms that couple power control coefficients from different UEs {as well as from the uplink and downlink}. {We therefore create global and local versions of the power control coefficients separately for the downlink and uplink UEs, which decouple the FD interference terms. We consider separate sub-problems for the downlink and uplink UEs with a separate set of constraints for each. These constraints, rewritten using the local variables, define feasible sets for the sub-problems of the downlink and uplink UEs, respectively. We introduce separate Lagrangian parameters for the downlink and uplink UEs, and separate penalty parameters for the downlink and uplink power control variables. This enables us to properly define the augmented Lagrangian and decouple the respective sub-problems at the D-servers which calculate the local solutions, and then eventually coordinate them into the globally optimal solution at the C-server.} \textit{The FD system required that we introduce these modifications to the standard ADMM approach and to the best of our knowledge, has not been attempted so far in mMIMO literature.}}\newline
{4) \textbf{Contributions regarding the AP selection algorithm:} We show that there is a fundamental limit to the number of UEs a FD AP can serve with a limited fronthaul capacity. We propose a \textit{proportionately-fair} rule capping the maximum number of uplink and downlink UEs served by each AP.  We use this rule to propose a fair AP selection algorithm which efficiently chooses the best subset of APs to serve each uplink and downlink UE. The proposed approach ensures user-centric architecture for our system. The proposed algorithm, which has a trivial complexity,  is shown to perform close to the optimal one proposed in~\cite{tvchen4}.}\newline
{5) \textbf{Contributions regarding the convergence of the distributed optimization algorithm:} We not only analytically prove its convergence but also numerically show that it i) achieves the same WSEE as the centralized approach; and ii) is responsive to changing weights which  can be set to prioritize UEs' EE requirements.}\newline
\section{System model}\label{model}
We consider, as shown in Fig. \ref{fig:0}, a FD CF mMIMO system  where $M$ FD APs serve $K = (K_u + K_d)$ single-antenna HD UEs on the same spectral resource, with $K_u$ and $K_d$ being the number of uplink and downlink UEs, respectively. Each AP has $N_t$ transmit and $N_r$ receive antennas, and is connected to the CPU using a limited-capacity fronthaul link 
which carries quantized uplink/downlink information to/from the CPU. We see from Fig. \ref{fig:0} that due to FD~model \newline 
\begin{figure}[t]
   \centering
   \includegraphics[width = \linewidth]{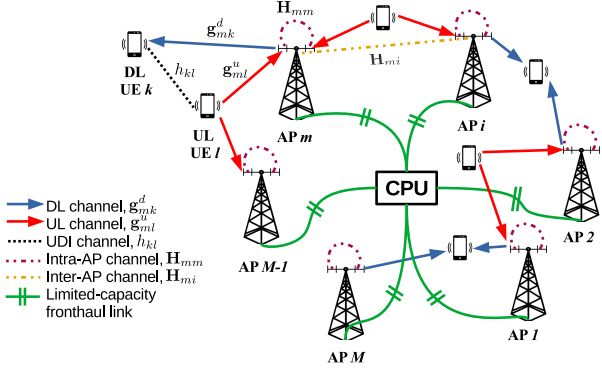}
   \vspace{-0.8cm}
   \caption{System model for FD CF mMIMO communications\vspace{-0.0cm}}
   \label{fig:0}
\end{figure}
\begin{itemize}
    \item  uplink receive signal of each AP is interfered by its own downlink transmit signal and that of other APs. These intra- and inter-AP interferences are shown using purple and brown dashed lines, respectively. 
    \item downlink UEs receive transmit signals from uplink UEs, causing uplink downlink interference (UDI) (shown as black dotted lines between uplink and downlink UEs). Additionally, the UEs experience multi-UE interference (MUI) as the APs serve them on the same spectral resource. 
\end{itemize}
We next explain various channels, their estimation and data transmission. We assume a coherence interval of duration $T_c$ (in s) with $\tau_c$ samples, which  is divided into: a) channel estimation phase of $\tau_t$ samples, and b) downlink and uplink data transmission of ($\tau_c$ - $\tau_t$) samples.
\subsection{{Channel description:}} 
The channel of the $k$th downlink UE to the transmit antennas of the $m$th AP is $\bm{g}^{d}_{mk} \in \mathbb{C}^{N_t \times 1}$, while the channel from the $l$th uplink UE to the receive antennas of the $m$th AP is $\bm{g}^{u}_{ml} \in \mathbb{C}^{N_r \times 1}$.\footnote{We, henceforth, consider $k = 1 \text{ to } K_d, l = 1 \text{ to } K_u$ and $m = 1 \text{ to } M$, to avoid repetition, unless mentioned otherwise.} We model these channels as $\bm{g}^{d}_{mk} = (\beta^{d}_{mk})^{1/2} \Tilde{\bm{g}}^{d}_{mk}$ and  $\bm{g}^{u}_{ml} = (\beta^{u}_{ml})^{1/2} \Tilde{\bm{g}}^{u}_{ml}$. Here $\beta^{d}_{mk}$ and $ \beta^{u}_{ml} \in \mathbb{R}$ are corresponding large scale fading coefficients, which are same for all antennas at the $m$th AP~\cite{CFvsSmallCells, FDCellFree}. The vectors $\Tilde{\bm{g}}^{d}_{mk}$ and $\Tilde{\bm{g}}^{u}_{ml}$ denote small scale fading  with independent and identically distributed (i.i.d.) $\mathcal{CN}(0,1)$ entries. The UDI channel between the $k$th downlink UE and $l$th uplink UE is modeled as $h_{kl} = (\Tilde{\beta}_{kl})^{1/2} \Tilde{h}_{kl}$ \cite{FDCellFree, NAFDCellFree}, where $\Tilde{\beta}_{kl}$ is the large scale fading coefficient and $\Tilde{h}_{kl} \sim \mathcal{CN}(0,1)$ is the small scale fading. The  inter- and intra-AP channels from the transmit antennas of the $i$th AP to the receive antennas of the $m$th AP are denoted as $\bm{H}_{mi} \in \mathbb{C}^{N_r \times N_t}$ for $i=1 \text{ to } M$.
\subsection{Uplink channel estimation:} \label{ul_ch_est}
Recall that the channel  estimation phase consists of $\tau_t$ samples. We divide them as $\tau_t = \tau^d_t + \tau^u_t$, where  $\tau^d_t$ and $\tau^u_t$  are samples used as pilots for the downlink and uplink UEs, respectively.  All the downlink (resp. uplink) UEs simultaneously transmit $\tau^d_t$ (resp. $\tau^u_t$)-length uplink pilots to the APs, which they use to estimate the respective channels. In this phase, both transmit and receive antenna arrays of each AP, similar to~\cite{FDCellFree}, operate in receive mode. The $k$th downlink UE (resp. $l$th uplink UE) transmits pilot signals $\sqrt{\tau^d_t}\bm{\varphi}^{d}_{k} \in \mathbb{C}^{\tau^d_t \times 1}$ (resp. $\sqrt{\tau^u_t}\bm{\varphi}^{u}_{l} \in \mathbb{C}^{\tau^u_t \times 1}$). We assume, similar to~\cite{CellFreeEE,FDCellFree}, that the pilots  i) have unit norm i.e., $\norm{\bm{\varphi}^{u}_{l}} = \norm{\bm{\varphi}^{d}_{k}} = 1$; and ii) are intra-set orthonormal i.e. $(\bm{\varphi}^{u}_{l})^{H} \bm{\varphi}^{u}_{l'} = 0\, \forall l \neq l' \text { and  } (\bm{\varphi}^{d}_{k})^{H} \bm{\varphi}^{d}_{k'} = 0\, \forall  k \neq k'$. 
Therefore, we need $\tau^d_t \geq K_d$ and $\tau^u_t \geq K_u$~\cite{CellFreeEE,FDCellFree}. 

The pilots received by  transmit and receive antennas of the $m$th AP are given respectively as
\begin{align}
    \bm{Y}^{tx}_{m} &=  \sqrt{\tau^d_t \rho_t} \sum\nolimits_{k=1}^{K_d} \bm{g}^{d}_{mk}  (\bm{\varphi}^{d}_{k})^{H}  +  \bm{W}^{tx}_{m}, \nonumber \\
    \bm{Y}^{rx}_{m}  &=  \sqrt{\tau^u_t \rho_t} \sum\nolimits_{l=1}^{K_u} \bm{g}^{u}_{ml}  (\bm{\varphi}^{u}_{l})^{H} +  \bm{W}^{rx}_{m}. \nonumber
\end{align}
Here $\rho_t$ is the normalized pilot transmit signal-to-noise-ratio (SNR). The matrices $\bm{W}^{tx}_{m} \in \mathbb{C}^{N_t \times \tau^d_t}$ and $ \bm{W}^{rx}_{m} \in \mathbb{C}^{N_r \times \tau^u_t}$ denote additive noise  with $\mathcal{CN}(0,1)$ entries. Each AP independently estimates its channels with the uplink and downlink UEs to avoid channel state information (CSI) exchange overhead~\cite{FDCellFree,FDCellFree2}. To estimate the channels $\bm{g}^{d}_{mk}$ and $\bm{g}^{u}_{ml}$, the $m$th AP projects the received signal onto the pilot signals $\bm\varphi_k^d$ and $\bm\varphi_l^u$ respectively, as 
\begin{align}
\hat{\bm{y}}^{tx}_{mk} &= \bm{Y}^{tx}_{m}\bm{\varphi}^{d}_{k} = \sqrt{\tau^d_t \rho_t}  \bm{g}^{d}_{mk} + \bm{W}^{tx}_{m}\bm{\varphi}^{d}_{k}\notag\\ 
\hat{\bm{y}}^{rx}_{ml} &= \bm{Y}^{rx}_{m}\bm{\varphi}^{u}_{l} = \sqrt{\tau^u_t \rho_t}  \bm{g}^{u}_{ml} + \bm{W}^{rx}_{m}\bm{\varphi}^{u}_{l}.\notag
\end{align}
These projections are used to compute the corresponding linear minimum-mean-squared-error (MMSE) channel estimates~\cite{FDCellFree} as 
\begin{align}
\hat{\bm{g}}^{d}_{mk}&\! = \! \mathbb{E}\{\bm{g}^{d}_{mk}(\hat{\bm{y}}^{tx}_{mk})^{H}\}(\mathbb{E}\{\hat{\bm{y}}^{tx}_{mk}(\hat{\bm{y}}^{tx}_{mk})^{H}\})^{-1} \hat{\bm{y}}^{tx}_{mk} \! = \! c^{d}_{mk} \hat{\bm{y}}^{tx}_{mk}, \notag\\
\hat{\bm{g}}^{u}_{ml} &\! = \! \mathbb{E}\{\bm{g}^{u}_{ml}(\hat{\bm{y}}^{rx}_{ml})^{H}\}(\mathbb{E}\{\hat{\bm{y}}^{rx}_{ml}(\hat{\bm{y}}^{rx}_{ml})^{H}\})^{-1} \hat{\bm{y}}^{rx}_{ml} = c^{u}_{ml} \hat{\bm{y}}^{rx}_{ml}, \notag
\end{align}
where $c^{d}_{mk} = \frac{\sqrt{\tau^d_t \rho_t}\beta^{d}_{mk}}{\tau^d_t \rho_t\beta^{d}_{mk} + 1}$ and $c^{u}_{ml} = \frac{\sqrt{\tau^u_t \rho_t}\beta^{u}_{ml}}{\tau^u_t \rho_t\beta^{u}_{ml} + 1}$. The estimation error vectors are defined as $\bm{e}^{u}_{ml} \triangleq \bm{g}^{u}_{ml} - \hat{\bm{g}}^{u}_{ml}$ and $\bm{e}^{d}_{mk} \triangleq \bm{g}^{d}_{mk} - \hat{\bm{g}}^{d}_{mk}$. With MMSE channel estimation, $\hat{\bm{g}}^{d}_{mk}, \bm{e}^{d}_{mk}$ and $\hat{\bm{g}}^{u}_{ml}, \bm{e}^{u}_{ml}$ are mutually independent and their individual terms are i.i.d. with pdf $\mathcal{CN}(0,\gamma^{d}_{mk}),\mathcal{CN}(0,\beta^{d}_{mk}-\gamma^{d}_{mk}), \mathcal{CN}(0,\gamma^{u}_{ml}), \mathcal{CN}(0, \beta^{u}_{ml} - \gamma^{u}_{ml})$ 
respectively, with $\gamma^{d}_{mk} = \frac{\tau^d_t \rho_t (\beta^{d}_{mk})^2}{\tau^d_t \rho_t \beta^{d}_{mk} + 1}$ and $\gamma^{u}_{ml} = \frac{\tau^u_t \rho_t (\beta^{u}_{ml})^2}{\tau^u_t \rho_t \beta^{u}_{ml} + 1}$~\cite{CellFreeEE,FDCellFree}. 

After channel estimation, data transmission starts simultaneously on downlink and uplink.
\subsection{Transmission model:} \label{dl/uldatatransmit}
An objective of this work is to derive a SE lower bound for FD CF mMIMO systems, where the $M$ APs serve $K_u$ uplink UEs and $K_d$ downlink UEs simultaneously on the same spectral resource.  We note that for the FD CF mMIMO systems, unlike the HD CF mMIMO systems~\cite{CFvsSmallCells, CellFreeLimFronthaul, CellFreeMaxMinUQ}, uplink and downlink transmissions interfere to cause UDI and inter-/intra-AP interferences. Further, unlike existing FD CF mMIMO literature \cite{FDCellFree,FDCellFree2,NAFDCellFree}, we consider a limited-capacity fronthaul.  
It is critical to model and analyze the UDI and inter-/intra-AP interferences and limited-capacity impairments while deriving the lower bound. 
\subsubsection{Downlink data transmission}
The CPU chooses  a message symbol $s^d_k$ for the $k$th downlink UE, which is distributed as $\mathcal{CN}(0,1)$. It intends to send this symbol to the $m$th AP via the limited-capacity fronthaul link. Before doing that, it multiplies $s^d_k$ with a power-control coefficient $\eta_{mk}$, and then quantizes the resulting signal.  The $m$th AP, due to its limited fronthaul capacity,  is allowed to serve only a subset $\kappa_{dm} \subset \{1, \dots, K_d\}$ of downlink users, an aspect which is discussed later in Section~\ref{apsellimfh}.
{The CPU consequently sends downlink symbols for UEs in the set $\kappa_{dm}$ to the $m$th AP,  which uses MMSE channel estimates to perform MRT precoding. The transmit signal of the $m$th AP is therefore given as follows
\begin{align} \label{dltransmitsignal}
   \bm{x}^{d}_{m} &= \sqrt{\rho_d} \sum\nolimits_{k \in \kappa_{dm}}  (\hat{\bm{g}}^{d}_{mk})^{*} \mathcal{Q}(\sqrt{\eta_{mk}} s^{d}_{k}) \nonumber \\
   &= \sqrt{\rho_d} \sum\nolimits_{k \in \kappa_{dm}}  (\hat{\bm{g}}^{d}_{mk})^{*} (\Tilde{a}\sqrt{\eta_{mk}} s^{d}_{k} + \varsigma^{d}_{mk}).
\end{align}
Here $\rho_d$ is the normalized maximum transmit SNR at each AP. The function $\mathcal{Q}(\cdot)$ denotes the quantization operation, which is modeled as a multiplicative attenuation $\Tilde{a}$, and an additive distortion $\varsigma^{d}_{mk}$, for the $k$th downlink UE in the fronthaul link between the CPU and the $m$th AP~\cite{CellFreeMaxMinUQ,CellFreeEEUQ}. We have, from Appendix~\ref{UQModel},  $\mathbb{E}\{\left(\varsigma^{d}_{mk}\right)^{2}\} = \left(\Tilde{b}-\Tilde{a}^2\right) \mathbb{E}\left\{|\sqrt{\eta_{mk}} s^d_k|^2\right\} =  \left(\Tilde{b}-\Tilde{a}^2\right)\eta_{mk}$, where the scalar constants $\Tilde{a}$ and $\Tilde{b}$ depend on the number of fronthaul quantization bits.

The $m$th AP must satisfy the average transmit SNR constraint, i.e., $\mathbb{E}\{\|\bm{x}^{d}_{m}\|^{2}\} \leq \rho_d$. Using the expression of $\bm{x}^d_m$ from~\eqref{dltransmitsignal}, and the above expression of quantization error variance, $\mathbb{E}\{\left(\varsigma^{d}_{mk}\right)^{2}\}$, the constraint can be simplified as follows
\begin{equation}
\rho_d \Tilde{b} \!\!\sum_{k \in \kappa_{dm}}\!\! \eta_{mk} \mathbb{E}\{\|\hat{\bm{g}}^d_{mk}\|^{2}\} \leq \rho_d \Rightarrow \Tilde{b} \!\!\sum_{k \in \kappa_{dm}}\!\! \gamma^{d}_{mk} \eta_{mk} \leq \frac{1}{N_t}. \label{cons1}
\end{equation}}

The $k$th downlink UE receives its desired message signal from a subset of all APs, denoted as $\mathcal{M}^{d}_{k} \subset \{1, \dots, M\}$, along with various interference and distortion components, as {in~\eqref{dlsignal} (shown at the top of the next page).} The $m$th AP serves the $k$th downlink UE iff $k \in \kappa_{dm} \Leftrightarrow m \in \mathcal{M}^d_k$. Here $x^u_l$  is the transmit signal of the $l$th uplink UE, which is modelled next. 
\subsubsection{Uplink data transmission} 
The $K_u$ uplink UEs also simultaneously transmit to all $M$ APs on the same spectral resource as that of the $K_d$ downlink UEs. The $l$th uplink UE transmits its signal $x^{u}_{l} = \sqrt{\rho_u \theta_l} s^{u}_{l}$ with  
$s^u_l$ being its  message symbol with pdf $\mathcal{CN}(0,1)$, $\rho_u$  being the  maximum uplink transmit SNR and $\theta_l$ being the power control coefficient. To satisfy the average SNR constraint, $\mathbb{E}\{|x^{u}_{l}|^{2}\} \leq \rho_u$, the $l$th uplink UE satisfies the constraint 
\begin{equation}
0 \leq \theta_l \leq 1. \label{cons2} 
\end{equation}

The FD APs not only receive the uplink UE signals but also their own  downlink transmit signals  and that of the other APs,  referred to as intra-AP and inter-AP interference, respectively. Using \eqref{dltransmitsignal}, the received uplink signal at the $m$th AP is 
\begin{align} 
    \bm{y}^{u}_{m} &= \sum_{l=1}^{K_u} \bm{g}^{u}_{ml} x^{u}_{l}
    + \sum_{i=1}^{M} \bm{H}_{mi} \bm{x}^{d}_{i} + \bm{w}^{u}_{m}  = \sqrt{\rho_u} \sum_{l=1}^{K_u} \bm{g}^{u}_{ml} \sqrt{\theta_l} s^{u}_{l} \nonumber \\
    &+ \sqrt{\rho_d} \sum_{i=1}^{M} \sum_{k \in \kappa_{di}}  \bm{H}_{mi} (\hat{\bm{g}}^{d}_{ik})^{*} (\Tilde{a}\sqrt{\eta_{ik}}s^{d}_{k}+ \varsigma^{d}_{ik}) + \bm{w}^{u}_{m}. \label{rxulsignal}
\end{align}
Here $\bm{w}^{u}_{m} \in \mathbb{C}^{N_r \times 1}$ is the additive receiver noise at the $m$th AP with i.i.d. entries $\sim \mathcal{CN}(0,1)$. 

The intra and inter-AP interference channels vary extremely slowly and  thus can be estimated with very low pilot overhead~\cite{NAFDCellFree}.  The receive antenna array of each AP, with estimated channel, can only partially mitigate the intra- and inter-AP interference~\cite{FDCellFree, NAFDCellFree}. The residual intra-/inter-AP interference (RI) channel $\bm{H}_{mi} \in \mathbb{C}^{N_r \times N_t}$ is modeled as  Rayleigh-faded with i.i.d. entries and pdf $ \mathcal{CN}(0, \gamma_{\text{RI},mi})$ \cite{FDCellFree, NAFDCellFree, LIC,WSEEEkant}. Here $\gamma_{\text{RI},mi} \triangleq \beta_{\text{RI},mi} \gamma_{\text{RI}}$, with $\beta_{\text{RI},mi}$ being the large scale fading coefficient from the $i$th AP to the $m$th AP, and $\gamma_{\text{RI}}$ being the RI power after its suppression.

The $m$th AP receives the  signals from all the uplink UEs,  and performs MRC for the $l$th uplink UE with $(\hat{\bm{g}}^{u}_{ml})^{H}$. Due to its limited fronthaul: i) AP quantizes the combined signal before sending it to CPU; ii)  as discussed in detail later in Section \ref{apsellimfh}, the CPU receives contributions for the $l$th uplink UE only from the subset of APs serving it, denoted as $\mathcal{M}^{u}_{l} \subset \{1, \dots, M\}$. Using \eqref{rxulsignal}, the signal received by the CPU for the $l$th uplink UE is expressed as {in~\eqref{ulsignal}} (shown at the top of the next page).
\begin{figure*}
\begin{align}
    \notag r^{d}_{k} &= \sum_{m=1}^{M} (\bm{g}^{d}_{mk})^{T} \bm{x}^{d}_{m} + \sum_{l=1}^{K_u} h_{kl} x^{u}_{l} + w^{d}_{k} = \underbrace{\Tilde{a}\sqrt{\rho_d} \sum_{m \in \mathcal{M}^{d}_{k}} \sqrt{\eta_{mk}}(\bm{g}^{d}_{mk})^{T}  (\hat{\bm{g}}^{d}_{mk})^{*} s^{d}_{k}}_{\text{message signal}} + \underbrace{\Tilde{a}\sqrt{\rho_d} \sum_{m=1}^{M} \sum_{q \in \kappa_{dm} \setminus k} \!\!\!\! \sqrt{\eta_{mq}}(\bm{g}^{d}_{mk})^{T}  (\hat{\bm{g}}^{d}_{mq})^{*} s^{d}_{q}}_{\text{multi-UE interference, MUI}^{d}_{k}} \\
    &+ \underbrace{\sum\nolimits_{l=1}^{K_u} h_{kl} x^u_l}_{\text{uplink downlink interference, UDI}^{d}_{k}} + \underbrace{\sqrt{\rho_d} \sum\nolimits_{m=1}^{M} \sum\nolimits_{q \in \kappa_{dm}} (\bm{g}^{d}_{mk})^{T}  (\hat{\bm{g}}^{d}_{mq})^{*} \varsigma^{d}_{mq}}_{\text{total quantization distortion, TQD}^{d}_{k}} + \underbrace{w_{k}^{d}}_{\text{AWGN at receiver}}. \label{dlsignal} \\
    \notag r^{u}_{l} &= \sum\nolimits_{m \in \mathcal{M}^{u}_{l}}  \mathcal{Q}((\hat{\bm{g}}^{u}_{ml})^{H} \bm{y}^u_m) 
   = \underbrace{\Tilde{a}\sum\nolimits_{m \in \mathcal{M}^{u}_{l}} \sqrt{\rho_u} \sqrt{\theta_l} (\hat{\bm{g}}^{u}_{ml})^{H} \bm{g}^{u}_{ml} s_{l}^{u}}_{\text{message signal}} + \underbrace{ \Tilde{a}\sum\nolimits_{m \in \mathcal{M}^{u}_{l}} \sum\nolimits_{q=1, q \neq l}^{K_u} \sqrt{\rho_u} \sqrt{\theta_q} (\hat{\bm{g}}^{u}_{ml})^{H} \bm{g}^{u}_{mq} s_{q}^{u}}_{\text{multi-UE interference, MUI}^{u}_{l}} \\
    &+ \underbrace{\Tilde{a} \!\! \sum\nolimits_{m \in \mathcal{M}^{u}_{l}} \sum\nolimits_{i=1}^{M} \sqrt{\rho_d} \sum\nolimits_{k \in \kappa_{di}}  (\hat{\bm{g}}^{u}_{ml})^{H} \bm{H}_{mi} (\hat{\bm{g}}^{d}_{ik})^{*}(\Tilde{a} \sqrt{\eta_{ik}} s^{d}_{k} + \varsigma^{d}_{ik})}_{\text{residual interference (intra-AP and inter-AP), RI}^{u}_{l}} +  \underbrace{\Tilde{a}  \sum\nolimits_{m \in \mathcal{M}^{u}_{l}} (\hat{\bm{g}}^{u}_{ml})^{H} \bm{w}^{u}_{m}}_{\text{AWGN at APs, N}^{u}_{l}} + \!\!\!\!\!\! \underbrace{\sum\nolimits_{m \in \mathcal{M}^{u}_{l}} \varsigma^{u}_{ml}}_{\text{total quantization distortion, TQD}^u_l}\!\!\!\!\!\!.\!\!\!\!\!\!  \label{ulsignal}
\end{align}
\hrule
\vspace{-0.6cm}
\end{figure*}

We denote the subset of uplink UEs served by the $m$th AP as $\kappa_{um} \subset \{1, \dots, K_u\}$. The $m$th AP serves the $l$th uplink UE iff $l \in \kappa_{um} \Leftrightarrow m \in \mathcal{M}^u_l$. The quantization operation $\mathcal{Q}(\cdot)$ is mathematically modeled using constant attenuation $\Tilde{a}$, and additive distortion $\varsigma^{u}_{ml}$ which, as shown in Appendix \ref{UQModel}, has power $\mathbb{E}\{\left(\varsigma^{u}_{ml}\right)^{2}\} = (\Tilde{b} - \Tilde{a}^{2}) \mathbb{E}\left\{\big|(\bm{g}^{u}_{ml})^{H}\bm{y}_m\big|^{2}\right\}$. 
\subsection{{User-centric behavior through limited fronthaul}:}\label{apsellimfh}
Initial CF mMIMO literature considered system models where all APs can serve all UEs~\cite{CFvsSmallCells,NgoCellFree2,9406061}. However, for geographically large areas, each UE can only have practically feasible channels with a subset of APs in its vicinity. Therefore, recent CF mMIMO literature has increasingly focused on user-centric CF mMIMO system design~\cite[and the references therein]{CellFreeMassiveMIMOBook}. In the subsequent discussion, we show that a user-centric CF deployment, as desired by us, is a natural outcome of the design choice to impose fronthaul capacity constraints on the CF mMIMO systen model, as shown in Fig.~\ref{fig:0}.

The fronthaul between the $m$th AP and the CPU uses $\nu_{m}$ bits to quantize the real and imaginary parts of transmit signal of the $m$th downlink UE and   the uplink receive signal after MRC i.e.,  $\sqrt{\eta_{mk}} s^d_k$, and  $(\hat{\bm{g}}^{u}_{ml})^{H} \bm{y}^u_m$, respectively. Due to the limited-capacity fronthaul, the $m$th AP serves only $K_{um} (\triangleq |\kappa_{um}|)$ and $K_{dm} (\triangleq |\kappa_{dm}|)$ UEs on the uplink and downlink, respectively~\cite{CellFreeMaxMinUQ,CellFreeEEUQ}. For each UE, we recall that there are $(\tau_c - \tau_t)$ data samples in each coherence interval of duration $T_c$.  The fronthaul data rate between the $m$th AP and the CPU is 
\begin{equation} \label{fhrate}
    R_{\text{fh},m} = \frac{2\nu_m(K_{dm} + K_{um}) (\tau_c - \tau_t)}{T_c}. 
\end{equation}
The fronthaul link between the $m$th AP and the CPU has capacity $C_{\text{fh},m}$ which implies that
\begin{equation} \label{limfronthaul}
R_{\text{fh},m} \leq C_{\text{fh},m} \Rightarrow   \nu_m \cdot (K_{um} + K_{dm}) \leq \frac{C_{\text{fh},m} T_c}{2(\tau_c - \tau_t)}. 
\end{equation}

We propose the following lemma where we consider a \textit{proportionally fair} approach to calculate $K_{dm}$ and $K_{um}$ \textit{in proportion to} the total number of downlink and uplink UEs, respectively. We use $\varepsilon \triangleq \{d,u\}$ to denote downlink and uplink, respectively, and define the total number of UEs, $K \triangleq K_u + K_d$. 
\begin{lemma}\label{fairAPlemma}
The maximum number of uplink and downlink UEs served by the $m$th AP when connected via a limited optical fronthaul to the CPU with capacity $C_{\text{fh},m}$ are given as 
\begin{equation}
\Bar{K}_{\epsilon m} = \left \lfloor{\frac{K_\epsilon}{K} \frac{C_{\text{fh},m} T_c}{4 (\tau_c - \tau_t) \nu_m}}\right \rfloor. \label{maxUEs}
\end{equation}
\end{lemma}
\vspace{-0.4cm}
\begin{proof}
Let $\Bar{K}_{um}$ and $\Bar{K}_{dm}$ denote the maximum number of uplink and downlink UEs served by the $m$th AP. We consider $\Bar{K}_{um} \propto K_u$ and $\Bar{K}_{dm} \propto K_d$ for proportional fairness on the uplink and downlink. Using~\eqref{limfronthaul}, we get, 
\begin{equation}
\notag \Bar{K}_{\epsilon m} \leq \frac{K_\epsilon}{K} \frac{C_{\text{fh},m} T_c}{2(\tau_c - \tau_t) \nu_m}. 
\end{equation}
The lemma follows from definition of floor function $\lfloor{ \cdot }\rfloor$. 
\end{proof}
Using the maximum limits obtained in~\eqref{maxUEs}, we assign $K_{um} = \min \{K_u, \Bar{K}_{um}\}$ and $K_{dm} = \min \{K_d, \Bar{K}_{dm}\}$. We see that the constraint imposed in~\eqref{limfronthaul} is similar to a UE-centric (UC) CF mMIMO system, wherein each UE is served by a subset of the APs~\cite{CellFreeMassiveMIMOBook}. We now define the procedure for AP selection to obtain the best subset of APs to serve each uplink and downlink UE, while satisfying~\eqref{limfronthaul}. For this, we extend the procedure in~\cite{CellFreeMaxMinUQ} for a FD system as follows: 
\begin{itemize}
    \item  The $m$th AP sorts the uplink and downlink UEs connected to it in descending order based on their channel gains ($\beta^{u}_{ml}$ and $\beta^{d}_{mk}$, respectively) and chooses $K_{um}$ uplink UEs and $K_{dm}$ downlink UEs, with the largest channel gains, to populate the sets $\kappa_{um}$ and $\kappa_{dm}$, respectively.
    \item  For the $l$th uplink UE and the $k$th downlink UE, we populate the sets $\mathcal{M}^u_l$ and $\mathcal{M}^{d}_{k}$, respectively, using the axioms $l \in \kappa_{um} \Leftrightarrow m \in \mathcal{M}^u_l$ and $k \in \kappa_{dm} \Leftrightarrow m \in \mathcal{M}^d_k$.
    \item  If an uplink or downlink UE is found with no serving AP, we use the procedure in Algorithm~\ref{algo0} to assign it the AP with the best channel conditions, while satisfying~\eqref{limfronthaul}.
    \end{itemize}
    \begin{algorithm} [htbp]
	\footnotesize
        \DontPrintSemicolon 
        \For{$k \gets 1$ \textbf{to} $K_d$}{
        \lIf{$\mathcal{M}^d_k = \phi$}{\linebreak
            Sort the APs in descending order of channel gains, $\beta^d_{mk}$, and find the AP $n$ with the largest channel gain. \linebreak
            For this $n$th AP, sort downlink UEs in $\kappa_{dn}$ in descending order of channel gains and find the $q$th downlink UE with minimum channel gain and \textit{at least one more connected AP}. \linebreak
            Remove the $q$th downlink UE from the set $\kappa_{dn}$ and add the $k$th downlink UE to it.
        }
        }
        Repeat the same procedure for all the uplink UEs $l = 1$ to $K_u$. 
        \caption{Fair AP selection for disconnected uplink and downlink UEs}\label{algo0}
    \end{algorithm}
Clearly, Lemma~\ref{fairAPlemma} ensures that each AP can only serve  a limited number of UEs which do not violate the fronthaul capacity constraints. This makes the system effectively a user-centric system. Algorithm~\ref{algo0} ensures that, under limited fronthaul constraints, the strongest AP-UE connections are retained and the UE-centric cell-free system delivers good performance.
\subsection{Self-interference mitigation methods}
{To ensure that our proposed FD CF mMIMO system has substantial performance improvement over an equivalent HD CF mMIMO system, we need effective techniques to cancel the self-interference (SI) caused due to inter-AP transmissions. We show in Eq.~\eqref{dlsignal}-\eqref{ulsignal} that this SI cancellation results in a residual interference (RI) due to the multiplication of a suppression factor, $\gamma_{\text{RI}}$. We now discuss SI cancellation techniques from the existing literature, which makes the SI suppression easier, by not requiring its instantaneous channel knowledge.}
\begin{itemize}
\item {\textit{Passive cancellation}:  Reference~\cite{ref25,ref30} suggests that a careful utilization of the passive self-interference suppression mechanisms (directional isolation, absorptive shielding, and cross polarization) can significantly suppress the SI. Reference~\cite{ref30} also showed that by additionally assuming statistical SI channel knowledge and by using antennas arrays of sources/destinations, the passive cancellation techniques can further suppress the SI.} 

\item {\textit{Large antenna array}: Reference~\cite{ref3} argued that with large N, channel vectors of the desired signal and the SI become nearly orthogonal. The beamforming techniques e.g, MRC/MRT inherently project the desired signal to the orthogonal complement space of the SI, which significantly reduces the SI.}

\item {\textit{Lower transmit power}}: Reference~\cite{ref3} also demonstrated that an alternative way to reduce interference could be to reduce transmit power, since the SI depends strongly on the transmit power. A cell free massive MIMO system, due to large number of transmit antennas, uses radically less transmit power/antenna than conventional MIMO systems, which significantly reduces the SI.

\item {We therefore, similar to existing massive MIMO FD literature~\cite{ref2,ref3,ref30}, assume that the SI can be significantly mitigated by utilizing the above mentioned SI cancellation techniques, and without requiring the knowledge of SI channel. However, if required, the residual SI can be further reduced by employing active (time-domain and spatial suppression) techniques developed in~\cite{ref31}, which require SI channel knowledge.} 

\item {\textit{Active cancellation}: The authors in~\cite{ref31} present an algorithm for SI channel estimation at the relay, which is equipped with large number of antennas. It also noted that the APs, which are infrastructure devices, are in a stationary environment.  The SI channel changes much more slowly than the channel from users to the APs. It is therefore reasonable to assume that i) the SI channel remains constant for multiple consecutive blocks; and ii) inter-AP pilot overhead is affordable because of the sufficiently longer coherence time of the residual SI channels. Similar to~\cite{ref31}, one can estimate the SI channel by utilizing its slowly-varying nature using a cost-efficient expectation-maximization algorithm with reduced complexity.}
\end{itemize}
\section{Achievable spectral efficiency} \label{rateexpresns}
We now derive the ergodic SE for the $k$th downlink UE and the $l$th uplink UE, denoted respectively as $\Bar{S}^{d}_{k}$ and $\Bar{S}^{u}_{l}$. The AP employs MRC/MRT in the uplink/downlink  and optimal uniform fronthaul quantization. {We use $\varepsilon \triangleq \{d,u\}$ to denote downlink and uplink, respectively; $\phi \triangleq \{k,l\}$ to denote $k$th downlink UE and $l$th uplink UE, respectively; and $\upsilon^{\varepsilon}_{m\phi} \triangleq \{\eta_{mk} \text{ for } \phi=k, \theta_{l} \text{ for } \phi=l\}$.} The ergodic SE expressions are calculated using \eqref{dlsignal} and \eqref{ulsignal}, as 
\begin{align} 
    &\Bar{S}^{\varepsilon}_{\phi} = \left(\frac{\tau_c - \tau_t}{\tau_c}\right) \mathbb{E}\{\log_2\Big(1 + \frac{P^{\varepsilon}_{\phi}}{I^{\varepsilon}_{\phi} + (\sigma^{\varepsilon}_{\phi,0})^2}\Big)\} ,\text{ where} \label{ergodrate} 
\end{align}
\begin{align} 
P^{\varepsilon}_{\phi} &=  \Big|\Tilde{a}\sum\nolimits_{m \in \mathcal{M}^{\varepsilon}_{\phi}} \sqrt{\rho_{\varepsilon}} \sqrt{\upsilon^{\varepsilon}_{m\phi}} (\hat{\bm{g}}^{\varepsilon}_{m\phi})^{H} \bm{g}^{\varepsilon}_{m\phi} s_{\phi}^{\varepsilon}\Big|^{2}, \nonumber \\
(\sigma^d_{k,0})^{2} &=  |w_{k}^{d}|^{2}, (\sigma^u_{l,0})^2 = \Big|\Tilde{a} \sum\nolimits_{m \in \mathcal{M}^{u}_{l}} (\hat{\bm{g}}^{u}_{ml})^{H} \bm{w}^{u}_{m}\Big|^{2},  \nonumber \end{align}
\begin{align}
I^d_k &= |\sum\nolimits_{l=1}^{K_u} h_{kl} \sqrt{\rho_u \theta_l} s^{u}_{l}|^{2} \nonumber \\
&+ \Big|\Tilde{a}\sqrt{\rho_d} \sum\nolimits_{m=1}^{M} \sum\nolimits_{q \in \kappa_{dm} \setminus k}\!\!\!\!\!\!  \sqrt{\eta_{mq}}(\bm{g}^{d}_{mk})^{T}  (\hat{\bm{g}}^{d}_{mq})^{*} s^{d}_{q}\Big|^{2} \nonumber \\
&+ \Big|\sqrt{\rho_d} \sum\nolimits_{m = 1}^{M} \sum\nolimits_{q \in \kappa_{dm}} (\bm{g}^{d}_{mk})^{T}  (\hat{\bm{g}}^{d}_{mq})^{*} \varsigma^{d}_{mq}\Big|^{2}, \nonumber \\
&I^u_l = \Big|\Tilde{a} \!\! \sum_{m \in \mathcal{M}^{u}_{l}} \sum_{q = 1, q \neq l}^{K_u} \!\!\!\! \sqrt{\rho_u} \sqrt{\theta_q} (\hat{\bm{g}}^{u}_{ml})^{H} \bm{g}^{u}_{mq} s_{q}^{u}\Big|^{2} + \Big|\sum_{m \in \mathcal{M}^{u}_{l}} \!\!\!\! \varsigma^{u}_{ml}\Big|^{2} \nonumber\\
&+ \Big|\Tilde{a} \!\!\! \sum_{m \in \mathcal{M}^{u}_{l}} \! \sum_{i=1}^{M}\!\! \sqrt{\rho_d}\! \sum_{k \in \kappa_{di}}\!\! (\hat{\bm{g}}^{u}_{ml})^{H} \bm{H}_{mi} (\hat{\bm{g}}^{d}_{ik})^{*} (\Tilde{a}\sqrt{\eta_{ik}}s^{d}_{k} + \varsigma^d_{ik})\Big|^{2}, \nonumber 
\end{align}
are signal, noise and interference powers respectively, for the $k$th downlink and $l$th uplink UEs. The expectation outside logarithm in the SE expressions in \eqref{ergodrate}  is mathematically intractable,  and it is difficult to simplify them further~\cite{CFvsSmallCells,FDCellFree,CellFreeMaxMinUQ}. We, similar to~\cite{CFvsSmallCells}, employ use-and-then-forget (UatF) technique to derive SE lower bounds. To use UatF, we rewrite the received signal at the $k$th downlink UE in~\eqref{dlsignal}, and at the CPU for the $l$th uplink UE in~\eqref{ulsignal} as
\begin{align}
    r^{\varepsilon}_{\phi} = \underbrace{\Tilde{a} \sum\nolimits_{m \in \mathcal{M}^{\varepsilon}_{\phi}} \sqrt{\rho_{\varepsilon}} \sqrt{\upsilon^{\varepsilon}_{m\phi}} \mathbb{E}\Big\{(\hat{\bm{g}}^{\varepsilon}_{m\phi})^{H} \bm{g}^{\varepsilon}_{m\phi}\Big\}s^{\varepsilon}_{\phi}}_{\text{desired signal, DS}^{\varepsilon}_{\phi}} + n^{\varepsilon}_{\phi}, \label{sigeff1}
\end{align}
where the effective additive noise terms $n^{\varepsilon}_{\phi}$ are expressed in~\eqref{dlnoiseeff1}-\eqref{ulnoiseeff1} (shown at the top of next page).
\begin{figure*}
\begin{align}
   \notag n^{d}_{k} &= \underbrace{\Tilde{a} \sqrt{\rho_d} \sum\nolimits_{m \in \mathcal{M}^{d}_{k}} \sqrt{\eta_{mk}}\left((\bm{g}^{d}_{mk})^{T}\! (\hat{\bm{g}}^{d}_{mk})^{*} -\mathbb{E}\{(\bm{g}^{d}_{mk})^{T} (\hat{\bm{g}}^{d}_{mk})^{*}\}\right)s^{d}_{k} }_{\text{beamforming uncertainty, BU}^d_k} + \underbrace{\sqrt{\rho_u} \sum\nolimits_{l=1}^{K_u} h_{kl} \sqrt{\theta_l} s^{u}_{l}}_{\text{UDI}^{d}_{k}} \\
    & \quad + \underbrace{ \Tilde{a} \sqrt{\rho_d} \sum\nolimits_{m=1}^{M} \sum\nolimits_{q \in \kappa_{dm} \setminus k} \sqrt{\eta_{mq}} (\bm{g}^{d}_{mk})^{T} (\hat{\bm{g}}^{d}_{mq})^{*} s_{q}^{d}}_{\text{MUI}^{d}_{k}} + \underbrace{\sqrt{\rho_d} \sum\nolimits_{m=1}^{M} \sum\nolimits_{q \in \kappa_{dm}} (\bm{g}^{d}_{mk})^{T} (\hat{\bm{g}}^{d}_{mq})^{*} \varsigma^{d}_{mq}}_{\text{TQD}^{d}_{k}} + w^{d}_{k}. \label{dlnoiseeff1}\\
    \notag n^{u}_{l}& = \underbrace{\Tilde{a} \sqrt{\rho_u} \sqrt{\theta_l} \sum\nolimits_{m \in \mathcal{M}^{u}_{l}} \left((\hat{\bm{g}}^{u}_{ml})^{H} \bm{g}^{u}_{ml} - \mathbb{E}\{(\hat{\bm{g}}^{u}_{ml})^{H} \bm{g}^{u}_{ml}\}\right)s^{u}_{l}}_{\text{beamforming uncertainty, BU}^u_l} + \underbrace{\Tilde{a} \sum\nolimits_{m \in \mathcal{M}^{u}_{l}} \sum\nolimits_{q = 1, q \neq l}^{K_u} \sqrt{\rho_u} \sqrt{\theta_q} (\hat{\bm{g}}^{u}_{ml})^{H} \bm{g}^{u}_{mq} s_{q}^{u}}_{\text{MUI}^{u}_{l}} \\
    & \quad + \underbrace{\Tilde{a}\sqrt{\rho_d} \!\! \sum\nolimits_{m \in \mathcal{M}^{u}_{l}} \sum\nolimits_{i=1}^{M} \! \sum\nolimits_{k \in \kappa_{di}} \!\!\!\! (\hat{\bm{g}}^{u}_{ml})^{H} \bm{H}_{mi} (\hat{\bm{g}}^{d}_{ik})^{*}(\Tilde{a} \sqrt{\eta_{ik}} s^{d}_{k} + \varsigma^d_{ik})}_{\text{RI}^{u}_{l}} +  \underbrace{\Tilde{a} \!\!\sum\nolimits_{m \in \mathcal{M}^{u}_{l}} (\hat{\bm{g}}^{u}_{ml})^{H} \bm{w}^{u}_{m}}_{N^{u}_{l}} + \underbrace{\sum\nolimits_{m \in \mathcal{M}^{u}_{l}} \!\! \varsigma^{u}_{ml}}_{\text{TQD}^u_l}, \label{ulnoiseeff1}.
\end{align} 
\hrule
\vspace{-.6cm}
\end{figure*}
The term DS$^{\varepsilon}_{\phi}$ in \eqref{sigeff1} denotes the  desired signal received over the channel mean, and the term BU$^{\varepsilon}_{\phi}$ in \eqref{dlnoiseeff1}-\eqref{ulnoiseeff1} denotes beamforming uncertainty i.e., the signal received over deviation of channel from mean. It is easy to see that $n^{\varepsilon}_{\phi}$ are uncorrelated with their respective $\text{DS}^{\varepsilon}_{\phi}$ terms. We, similar to \cite{FDCellFree}, treat them as worst-case additive Gaussian noise, an approximation  which is tight for mMIMO systems \cite{FDCellFree}. Using \eqref{sigeff1}-\eqref{dlnoiseeff1}, we next derive an achievable SE lower bound.
\begin{theorem}
An achievable lower bound to the SE for the $k$th downlink UE with MRT and the $l$th uplink UE with MRC can be expressed respectively as 
\begin{align} 
    \!\!S^{d}_{k} \!&=\! \tau_f \log_2 \Bigg(1 \!+\! \frac{(\sum_{m \in \mathcal{M}^{d}_{k}} \!\! A^d_{mk} \sqrt{\eta_{mk}})^{2}}{\sum\limits_{m=1}^{M} \sum\limits_{q \in \kappa_{dm}}\!\!\!\!\! B^d_{kmq} \eta_{mq} \!+\! \sum\limits_{l=1}^{K_u} \!\! D^d_{kl} \theta_l \!+\! 1}\Bigg), \label{dlrate} \\
    \!\!S^{u}_{l} \!&=\! \tau_f \! \log_2 \! \Bigg(\!1 \!+\! \frac{A^u_l \theta_l}{\sum\limits_{q=1}^{K_u} \!\! B^u_{lq} \theta_q \!\!+\!\! \sum\limits_{i=1}^{M} \! \sum\limits_{k \in \kappa_{di}}\!\!\!\!\! D^u_{lik} \eta_{ik} \!\!+\!\! E^u_l \theta_l \!\!+\!\! F^u_l}\!\Bigg), \label{ulrate} 
\end{align}
where 
$\tau_f \!\!=\!\! \left(\!\frac{\tau_c - \tau_t}{\tau_c}\!\right)$, $A^d_{mk} \!\!=\!\! \Tilde{a} N_t \sqrt{\rho_d} \gamma^d_{mk}$, $B^d_{kmq} \!\!=\!\! \Tilde{b} N_t \rho_d \beta^{d}_{mk} \gamma^{d}_{mq}$, $D^d_{kl} = \rho_u \Tilde{\beta}_{kl}$, $A^u_l \!\!=\!\! \Tilde{a}^{2} N^2_r \rho_u (\!\!\!\!\sum\limits_{m \in \mathcal{M}^{u}_{l}}\!\! \gamma^{u}_{ml})^{2}\!\!$, $B^u_{lq} \!\!\!=\!\!\! \Tilde{b} N_r \rho_u \!\!\!\! \sum\limits_{m \in \mathcal{M}^{u}_{l}} \!\!\!\! \gamma^{u}_{ml}\beta^{u}_{mq}$, $D^u_{lik} \!\!\!=\!\!\! \Tilde{b}^2 N_r N_t \rho_d \gamma^{d}_{ik}\!\!\!\! \sum\limits_{m \in \mathcal{M}^{u}_{l}}\!\!\!\! \gamma^{u}_{ml} \beta_{\text{RI},mi} \gamma_{\text{RI}}$,$E^u_l\!\!=\!\!(\Tilde{b}\!\!-\!\!\Tilde{a}^2)N^2_r \rho_u \!\!\!\!\!\! \sum\limits_{m \in \mathcal{M}^u_l}\!\!\!\! (\gamma^{u}_{ml})^{2}$, and $F^u_l \!\!=\!\! \Tilde{b} N_r \!\!\sum_{m \in \mathcal{M}^{u}_{l}}\!\! \gamma^{u}_{ml}$. Here $\bm{\eta} \triangleq \{\eta_{mk}\} \in \mathbb{C}^{M \times K_d}$, $\bm{\Theta} \triangleq \{ \theta_l\}  \in \mathbb{C}^{K_u \times 1}$ and $\bm{\nu} \triangleq \{\nu_m\}  \in \mathbb{C}^{M \times 1}$ are the variables on which the SE is dependent. We recall from Section \ref{model} that $\Tilde{a}$ and $\Tilde{b}$ in \eqref{dlrate}-\eqref{ulrate} depend on the number of quantization bits, $\bm{\nu}$.
\end{theorem}
\begin{proof}
Refer to Appendix~\ref{SINRterms}. The SE expressions are functions of large scale fading coefficients, $\gamma^d_{mk}$ and $\gamma^u_{ml}$, which we will use to  optimize WSEE. \textit{This is unlike~\cite{FDCellFree2} which requires instantaneous channel while optimizing SE-GEE metric.} 
\end{proof}
\begin{remark}
MRC/MRT has tractable SE expression that depend solely on large-scale channel statistics, which remain constant over \textit{hundreds} of coherence intervals~\cite{AGWireless}. 
This is in contrast to zero-forcing designs which yield better SE but not tractable SE expressions~\cite{CellFreeMassiveMIMOBook}. 
Further, MRC/MRT can be implemented in a distributed fashion with low complexity.
\end{remark}
\section{Two-layer decentralized WSEE optimization for FD CF mMIMO} \label{ADMM}
We now devise a decentralized algorithm which maximizes WSEE by calculating the optimal downlink and uplink power control coefficients $\bm{\eta}^{*}$ and $\bm{\Theta}^{*}$, respectively. We use ``two-layered'' approach to decompose WSEE maximization into a sequential process with two distinct individual steps, each of which is called a ``layer''. The first layer simplifies the non-convex WSEE maximization into a successive convex approximation (SCA) setting. Its output is a generalized convex program (GCP) which needs to be solved iteratively for the optimal solution. The second layer optimally solves above GCP, either centrally through standard interior-point approaches or decentrally using ADMM method. The proposed procedure is outlined in Algorithm~\ref{algo11}.
\begin{algorithm}[h]
	\footnotesize
    \DontPrintSemicolon 
    \emph{AP selection}: Select APs that serve each UE while satisfying limited fronthaul constraints.\\
    \emph{SCA framework (first layer)}: Apply a series of transformations and approximations to recast the non-convex WSEE maximization using successive convex approximation (SCA) framework. The output of first layer is a GCP.\\
    \emph{Decentralized ADMM approach (second layer)}: Introduce global and local variables to decouple the problem into multiple sub-problems. Each sub-problem is solved at a distributed (or ``D") server, whose solutions are coordinated to obtain the global solution at the central (or ``C") server. This procedure is implemented using ADMM.
    \caption{{Two-layer decentralized 
    WSEE maximization algorithm}}
    \label{algo11}
\end{algorithm}

We use $\varepsilon \triangleq \{d,u\}$ for the downlink and uplink, respectively; $\phi \triangleq \{k,l\}$ for the $k$th downlink UE and $l$th uplink UE, respectively; and first define the individual EE for each UE as $\text{EE}^{\varepsilon}_{\phi} = \frac{B \cdot S^{\varepsilon}_{\phi}}{p^{\varepsilon}_{\phi}}$\cite{WSEEEfrem}, where $B$ is the system bandwidth, and $p^{\varepsilon}_{\phi}$ denotes the power consumed by each UE. 
The fronthaul links consume power for both downlink and uplink transmission.  The APs consume power while transmitting data to the downlink UEs, and  the uplink UEs consume power while transmitting their data. The power consumed by the system to transmit data to the $k$th downlink UE and the power consumed by the $l$th uplink UE are given respectively as~\cite{FDCellFree2,CellFreeEEUQ}
\begin{align}
    p^d_k &= P_{\text{fix}} + N_t \rho_d N_0 \sum\nolimits_{m \in \mathcal{M}^{d}_{k}} \frac{1}{\alpha_m} \gamma^{d}_{mk} \eta_{mk} + P^{d}_{\text{tc},k}, \label{dlpower} \\
    p^u_l &= P_{\text{fix}} + \rho_u N_0 \frac{1}{\alpha'_l} \theta_l + P^{u}_{\text{tc},l}. \label{ulpower}
\end{align}

Here $\alpha_m, \alpha'_l$ are  power amplifier efficiencies at the $m$th AP and the $l$th uplink UE respectively~\cite{FDCellFree}, $N_0$ is the noise power and $P^{d}_{\text{tc},k}, P^{u}_{\text{tc},l}$ are the powers required to run the transceiver chains at each antenna of the $k$th downlink UE and the $l$th uplink UE, respectively. The power consumed by the AP transceiver chains and the fronthaul between APs and CPU:
\begin{equation}
    P_{\text{fix}} \!=\! \frac{1}{{K}} \sum\nolimits_{m=1}^{M} \Big( P_{0,m} + (N_t + N_r) P_{\text{tc},m} + P_{\text{ft}}\frac{R_{\text{fh},m}}{C_{\text{fh},m}} \Big). \label{fixpower} 
\end{equation}
Here $P_{\text{tc},m}$ is the power required to run the transceiver chains at each antenna of the $m$th AP. The fronthaul power consumption for the $m$th AP has a fixed component, $P_{0,m}$, and a traffic-dependent component, which attains a maximum value of $P_{\text{ft}}$ at full capacity $C_{\text{fh},m}$. The term $R_{\text{fh},m}$, given in~\eqref{fhrate}, is the  fronthaul data rate of the $m$th AP.
The WSEE is now defined as the {weighted sum of EEs of individual UEs~\cite{fractionalprogrammingbook}.} 
\begin{align}
   \text{WSEE} \!\! &= \!\! \sum_{k=1}^{K_d} \! w^d_k \text{EE}^d_k \!+\! \sum_{l=1}^{K_u} \! w^u_l \text{EE}^u_l \!\stackrel{\bigtriangleup}{=}\!B\! \left(\sum_{k=1}^{K_d}\! w^d_k \frac{S^{d}_{k}}{p^d_k} \!+\! \sum_{l=1}^{K_u}\! w^u_l \frac{S^{u}_{l}}{p^u_l}\right)\!\!, \!\!\!\! \nonumber
\end{align}
where $w^{\varepsilon}_{\phi}$
are weights assigned to the UEs to account for their heterogeneous EE requirements. The WSEE metric can prioritize the EE requirements of individual UEs by assigning them different weights~\cite{WSEEEfrem,WSEEEkant}. For example, it could assign a higher weight to a UE that is more energy-scarce. {After omitting the constant $B$ from the objective, t}he WSEE maximization problem can now be formulated as follows
\begin{subequations} 
\begin{alignat}{2}
\notag \textbf{P1}: &\!\! \underset{\substack{\bm{\eta \text{, } \Theta \text{, } \nu}}}{\mbox{max}} \!\!\!\! && \sum\nolimits_{k=1}^{K_d} w^d_k \frac{S^{d}_{k}(\bm{\eta, \Theta,\nu})}{p^d_k(\bm{\eta,\nu})} + \sum\nolimits_{l=1}^{K_u} w^u_l \frac{S^{u}_{l}(\bm{\eta, \Theta, \nu})}{p^u_l(\bm{\Theta,\nu})}\\ 
& \;\; \text{ s.t. } \; && S^d_k (\bm{\eta, \Theta,\nu}) \geq S^d_{ok} \text{, } S^u_l (\bm{\eta, \Theta,\nu}) \geq S^u_{ol}, \label{P1const1} \\ 
&&&  R_{\text{fh},m} \leq C_{\text{fh},m}, \eqref{cons1}, {\eqref{cons2}}. \label{P1const3}  
\end{alignat}
\end{subequations}
The quality-of-service (QoS) constraints in \eqref{P1const1} guarantee a minimum SE, denoted by the constants $S^d_{ok}$ and $S^u_{ol}$, for each downlink and uplink UE respectively. The first constraint in $\eqref{P1const3}$ ensures that the fronthaul transmission rate for all APs is within the capacity limit. 
We observe that the number of quantization bits $\bm{\nu}$, if included in problem \textbf{P1}, will make it a difficult-to-solve integer optimization problem  \cite{CellFreeMaxMinUQ,CellFreeEEUQ,boyd2004convex}. We therefore solve it to optimize the power control coefficients $\{\bm{\eta}, \bm{\Theta}\}$, by fixing $\bm{\nu}$ such that it satisfies the first constraint in  \eqref{P1const3}~\cite{CellFreeMaxMinUQ,CellFreeEEUQ}, and numerically investigate {$\bm{\nu}$} in Section \ref{simresults}. We reformulate \textbf{P1} as follows
\begin{alignat}{2}
\notag \!\!\!\! \textbf{P2}: & \underset{\substack{\bm{\eta\text{, }\Theta}}}{\text{ }\mbox{max}} && \sum\nolimits_{k=1}^{K_d} w^d_k \frac{S^{d}_{k}(\bm{\eta, \Theta})}{p^d_k(\bm{\eta})} + \sum\nolimits_{l=1}^{K_u} w^u_l \frac{S^{u}_{l}(\bm{\eta, \Theta})}{p^u_l(\bm{\Theta})} \\
& \;\; \text{ s.t. } \; && S^d_k(\bm{\eta, \Theta}) \geq S^d_{ok} \text{, } S^u_l(\bm{\eta, \Theta}) \geq S^u_{ol}, \label{P2const1}\\
&&& \eqref{cons1}, \eqref{cons2}. \nonumber 
\end{alignat}
The objective in \textbf{P2} is a sum of ratios, each of which is a PC function (concave-over-linear) of power control coefficients $\{\bm{\eta}, \bm{\Theta}\}$. 
It is, therefore, not guaranteed to be a PC function and Dinkelbach's algorithm cannot be applied to maximize it \cite{fractionalprogrammingbook}. This makes it a much harder objective to optimize as opposed to the more commonly studied GEE metric, which is a PC function~\cite{fractionalprogrammingbook} and has been investigated for CF mMIMO systems~\cite{FDCellFree2, CellFreeEE, CellFreeEEUQ, CellFreeEEmmW}. 

We now maximize WSEE centrally and decentrally using a two-layered approach. The first layer comprises an SCA framework, which formulates a GCP by approximating the non-convex objective and constraints in \textbf{P2} as convex. In the second layer, the approximate GCP formed in the $n$th SCA iteration is either solved centrally or decentrally using ADMM. 

Since the approximate GCP obtained in the first layer, due to coupled optimization variables, is not in the standard ADMM form, we introduce their local and global versions. The sub-problems to update local variables are solved independently, and the local variables are coordinated to calculate the global solution \cite{ADMM,ADMM2}. The updation of variables and coordination continues till ADMM converges. The obtained solution is then used to formulate  GCP for the $(n+1)$th SCA iteration.
\vspace{-0.4cm}
\subsection{SCA Framework}
We now first linearize the non-convex objective in \textbf{P2} using epigraph transformation as~\cite{boyd2004convex}
\begin{alignat}{2}
\notag \!\!\!\! \textbf{P3}:& \underset{\substack{\bm{\eta\text{, }\Theta\text{, }f^{d}\text{, }f^{u}}}}{\mbox{max}} && \sum\nolimits_{k=1}^{K_d} w^d_k f^d_k + \sum\nolimits_{l=1}^{K_u} w^u_l f^u_l  \\
& \;\; \text{ s.t. } \; &&  f^d_k \leq \frac{S^{d}_{k}(\bm{\eta, \Theta})}{p^d_k(\bm{\eta})} \text{, } f^u_l \leq \frac{S^{u}_{l}(\bm{\eta, \Theta})}{p^u_l(\bm{\Theta})}, \label{P3const1} \\ 
&&& \eqref{cons1}, \eqref{cons2}, \eqref{P2const1}.  \nonumber 
\end{alignat}
Here $\bm{f}^{\varepsilon} \triangleq [f^{\varepsilon}_1 \dots f^{\varepsilon}_{K_{\varepsilon}}] \in \mathbb{C}^{K_{\varepsilon} \times 1}$ are slack variables~\cite{boyd2004convex}. To approximate the non-convex constraints in \eqref{P2const1} and \eqref{P3const1} as convex, we substitute $S^d_k$ and $S^u_l$ from \eqref{dlrate}-\eqref{ulrate} and cross-multiply the terms $p^d_k, p^u_l$ and $f^d_k, f^u_l$ in \eqref{P3const1}. We also introduce slack variables $\bm{\Psi}^{\varepsilon} \triangleq [\Psi^{\varepsilon}_1, \dots, \Psi^{\varepsilon}_{K_{\varepsilon}}] \in \mathbb{C}^{K_{\varepsilon} \times 1}$,  $\bm{\zeta}^{\varepsilon} \triangleq [\zeta^{\varepsilon}_1, \dots, \zeta^{\varepsilon}_{K_{\varepsilon}}] \in \mathbb{C}^{K_{\varepsilon} \times 1}$ and equivalently cast \textbf{P3} as~\cite{fractionalprogrammingbook} 
\begin{subequations}
\begin{align}{2}
\notag \textbf{P4}:& \!\!\!\!\!\!\!\!\!\! \underset{\substack{\bm{ \eta\text{, }\Theta\text{, }f^d\text{, }f^u\text{, }}\\\bm{\Psi^d\text{, }\Psi^u\text{, }\zeta^d\text{, }\zeta^u}}}{\mbox{max}} \!\!\!\!\!\! && \sum\nolimits_{k=1}^{K_d} w^d_k f^d_k + \sum\nolimits_{l=1}^{K_u} w^u_l f^u_l  \\
& \;\; \text{ s.t. } \; && p^d_k \leq \frac{(\Psi^d_{k})^{2}}{f^d_k} 
\text{, }  p^u_l \leq \frac{(\Psi^u_{l})^{2}}{f^u_l},  \label{P4const1} 
\end{align}
\begin{align}
&&& (\Psi^d_{k})^{2}  \leq \tau_f  \log_2 (1 + \zeta^d_k)\text{, }(\Psi^u_{l})^{2} \leq \tau_f \log_2 (1 + \zeta^u_l),\!\!\!\!\!\! \label{P4const2} \\
&&& \zeta^d_k \!\leq\! \frac{(\sum_{m \in \mathcal{M}^{d}_{k}} A^d_{mk} \sqrt{\eta_{mk}})^{2}}{\sum_{m=1}^{M} \!\sum_{q \in \kappa_{dm}} \!\!\!\! B^d_{kmq} \eta_{mq} \!\!+\!\! \sum_{l=1}^{K_u} \!\! D^d_{kl} \theta_l \!\!+\!\! 1}, \!\!\!\!\!\! \label{P4const3}\\ 
&&& \zeta^u_l \!\leq\!  \frac{A^u_l \theta_l}{\sum_{q=1}^{K_u}\!\! B^u_{lq} \theta_q \!+\! \sum_{i=1}^{M}\! \sum_{k \in \kappa_{di}} \!\!\!\! D^u_{lik} \eta_{ik} \!+\! E^u_l \theta_l \!+\! F^u_l}, \!\!\!\!\label{P4const4} \\
&&& \log_2(1 \!+\! \zeta^d_k) \!\geq\! S^d_{ok}/\tau_f \text{,} \log_2(1 \!+\! \zeta^u_l) \!\geq\! S^u_{ol}/\tau_f, \!\!\!\!\!\! \label{P4const5} \\
&&& \eqref{cons1}, \eqref{cons2}. \nonumber 
\end{align}
\end{subequations}
We introduce the variable $c_{mk} \triangleq \sqrt{\eta_{mk}}$ and denote $\bm{C} \triangleq \{c_{mk}\} \in \mathbb{C}^{M \times K_d}$ to remove concave terms in \eqref{P4const3} arising due to $\sqrt{\eta_{mk}}$ and facilitate its conversion into a convex constraint. We introduce additional slack variables $\bm{\lambda}^{\varepsilon} \triangleq [\lambda^{\varepsilon}_1, \dots, \lambda^{\varepsilon}_{K_{\varepsilon}}] \in \mathbb{C}^{K_{\varepsilon} \times 1}$ to further simplify the non-convex constraints \eqref{P4const3}-\eqref{P4const4}. We now cast \textbf{P4} equivalently as 
\begin{subequations}
\begin{alignat}{2}
\notag \textbf{P5}:& \!\!\!\!\!\!\!\! \underset{\substack{\bm{C\text{, }\Theta\text{, }f^d\text{, }f^u}\\ \bm{\Psi^d\text{, }\Psi^u\text{, }\zeta^d\text{, }}\\\bm{\zeta^u\text{, }\lambda^d\text{, }\lambda^u}}}{\mbox{max}} && \sum\nolimits_{k=1}^{K_d} w^d_k f^d_k  + \sum\nolimits_{l=1}^{K_u} w^u_l f^u_l \\
\vspace{-3.0cm}
& \;\; \text{ s.t. } \; && \sum_{m=1}^{M} \! \sum_{q \in \kappa_{dm}}\!\!\!\! B^d_{kmq} c^2_{mq} \!+\! \sum_{l=1}^{K_u} \! D^d_{kl} \theta_l \!+\! 1 \!\leq\! \frac{(\lambda^d_k)^2}{\zeta^d_k}, \!\!\!\! \label{P5const1}  \\ 
&&& \sum_{q=1}^{K_u} \!\! B^u_{lq} \theta_q \!+\! \sum_{i=1}^{M}\! \sum_{k \in \kappa_{di}} \!\!\!\! D^u_{lik} c^2_{ik} \!+\! E^u_l \theta_l \!+\! F^u_l \!\leq\!  \frac{(\lambda^{u}_l)^2}{\zeta^u_l}, \!\!\!\! \label{P5const2}\\
&&& \lambda^d_k \leq \sum\nolimits_{m \in \mathcal{M}^{d}_{k}} A^d_{mk} c_{mk} \text{, } (\lambda^{u}_l)^2 \leq A^u_l \theta_l, \label{P5const3} \\
&&& \lambda^d_k \geq 0, \Tilde{b} \! \sum\nolimits_{k \in \kappa_{dm}} \!\!\!\!\!\! \gamma^{d}_{mk} c^2_{mk} \leq \frac{1}{N_t} \text{, } c_{mk} \geq 0, \!\!\!\! \label{P5const4}\\
&&& \eqref{P4const1}, \eqref{P4const2}, \eqref{P4const5}, \eqref{cons2}. \notag
\end{alignat}
\end{subequations}

We note that \textbf{P5} has all convex constraints except \eqref{P4const1} and \eqref{P5const1}-\eqref{P5const2}. Since a  first-order Taylor approximation is a global under-estimator of a convex function~\cite{boyd2004convex}, we now linearize the right-hand side of these constraints. At the $n$th iteration, we substitute first-order Taylor approximation $\frac{f^{2}_{1}}{f_2} \geq 2 \frac{f^{(n)}_{1}}{f^{(n)}_2} f_1 - \frac{(f^{(n)}_1)^{2}}{(f^{(n)}_2)^{2}}f_2 \triangleq \Lambda^{(n)} \left(\frac{f^{2}_{1}}{f_2}\right)$ and use {~\eqref{dlpower}-\eqref{ulpower}} to recast \textbf{P5} into a GCP:
\begin{subequations}
\begin{align}
\notag \textbf{P6}:& \!\!\!\!\!\!\!\!\underset{\substack{\bm{C\text{, }\Theta\text{, }f^d\text{, }f^u}\\ \bm{\Psi^d\text{, }\Psi^u\text{, }\zeta^d\text{, }}\\\bm{\zeta^u\text{, }\lambda^d\text{, }\lambda^u}}}{\mbox{max}} \!\!\!\!\!\!\!\!\!\!\! && \sum\nolimits_{k=1}^{K_d} w^d_k f^d_k  + \sum\nolimits_{l=1}^{K_u} w^u_l f^u_l \\
& \;\; \text{ s.t. } \; && \sum_{m=1}^{M}\! \sum_{q \in \kappa_{dm}} \!\!\!\! B^d_{kmq} c^2_{mq} \!\!+\!\! \sum_{l=1}^{K_u} \!\! D^d_{kl} \theta_l \!+\! 1 \!\leq\!\! \Lambda^{(n)}\! \left(\frac{(\lambda^{d}_{k})^{2}}{\zeta^d_k}\right)\!\!, \label{P6const1} \\
&&& \sum_{q=1}^{K_u} \!\! B^u_{lq} \theta_q \!\!+\!\! \sum_{i=1}^{M} \!\! \sum_{k \in \kappa_{di}} \!\!\!\! D^u_{lik} c^2_{ik} \!\!+\!\! E^u_l \theta_l \!\!+\!\! F^u_l \!\!\leq\!\! \Lambda^{(n)}\!\! \left(\!\frac{(\lambda^{u}_{l})^{2}}{\zeta^u_l}\!\!\right)\!\!, \label{P6const2} \\
&&& P_{\text{fix}} \!\!\!+\! N_t\rho_d N_0 \!\!\!\! \sum_{m \in \mathcal{M}^{d}_{k}} \!\!\!\! \frac{\gamma^{d}_{mk} c^{2}_{mk}}{\alpha_m}\!+\! P^{d}_{\text{tc},k} \!\leq\! \Lambda^{(n)}\! \left(\frac{(\Psi^{d}_{k})^{2}}{f^d_k}\right)\!\!, \label{P6const3} \\
&&& P_{\text{fix}} \!+\! \rho_u N_0 \frac{\theta_l}{\alpha'_l} \!+\! P^{u}_{\text{tc},l}  \leq \Lambda^{(n)}  \left(\frac{(\Psi^{u}_{l})^{2}}{f^u_l}\right)\!\!, \!\!\!\! \label{P6const4} \\
&&& \eqref{cons2}, \eqref{P4const2}, \eqref{P4const5}, \eqref{P5const3}, \eqref{P5const4}. \nonumber 
\end{align}
\end{subequations}

We next provide a centralized SCA to solve \textbf{P6} in the second layer in Algorithm~\ref{algo1}.
\begin{algorithm}[h]
	\footnotesize
\DontPrintSemicolon 
\KwIn{i) Initialize power control coefficients $\{\bm{C,\Theta}\}^{(1)}$ by allocating equal power to all downlink UEs being served and full power to all uplink UEs.  Set $n = 1$.\\
ii) Initialize $\{\bm{f^d,f^u,\Psi^d,\Psi^u,\zeta^d,\zeta^u,\lambda^d,\lambda^u}\}^{(1)}$ by replacing \eqref{P5const3}, \eqref{P6const1}-\eqref{P6const2}, \eqref{P4const2} and \eqref{P6const3}-\eqref{P6const4} by equality.}
\KwOut{Globally optimal power control coefficients $\{\bm{C,\Theta}\}^{*}$}
\While{$\|\bm{r}_{\text{SCA}}^{(n)}\| \leq \epsilon_{\text{SCA}}$}{
Solve \textbf{P6} for the $n$th SCA iteration to obtain optimal variables, $\{\bm{f^d, f^u, \Psi^d, \Psi^u, \zeta^d, \zeta^u, \lambda^d, \lambda^u, C, \Theta}\}^{*,(n)}$. \\
Assign the SCA iterates for the $(n+1)$th iteration, $\{\bm{f^d, f^u, \Psi^d, \Psi^u, \zeta^d, \zeta^u, \lambda^d, \lambda^u, C, \Theta}\}^{(n+1)} = \{\bm{f^d, f^u, \Psi^d, \Psi^u, \zeta^d, \zeta^u, \lambda^d, \lambda^u, C, \Theta}\}^{*,(n)}$. 
}
\caption{Centralized WSEE maximization algorithm}\label{algo1}
\end{algorithm}

The SCA procedure converges when $\|\bm{r}_{\text{SCA}}^{(n)}\| = \sqrt{\|\bm{C}^{(n+1)} - \bm{C}^{(n)}\|^{2}_{F} + \|\bm{\Theta}^{(n+1)} - \bm{\Theta}^{(n)}\|^{2}}$ has magnitude
$\|\bm{r}_{\text{SCA}}^{(n)}\| \leq \epsilon_{\text{SCA}}$, where $\epsilon_{\text{SCA}}$ is the convergence threshold. 

\begin{remark}{\textit{Convergence of centralized algorithm:}} \label{remark1}
At the $n$th SCA iteration, \textbf{P6} is obtained from \textbf{P5} by applying first-order Taylor approximations to the constraints \eqref{P4const1} and~\eqref{P5const1}-\eqref{P5const2}. These approximations are of the form $\Lambda(\bm{x}) \triangleq \frac{x_1^2}{x_2} \geq 2 \frac{x_1^{(n)}}{x_2^{(n)}} x_1 - \left(\frac{x_1^{(n)}}{x_2^{(n)}}\right)^{2} x_2 \triangleq \Bar{\Lambda}(\bm{x},\bm{x}^{(n)})$. It is easy to show that \textbf{P6} is the \textit{inner-approximation problem} for \textbf{P5}, where we replace each of the constraints~\eqref{P4const1} and~\eqref{P5const1}-\eqref{P5const2}, denoted here as $g_i(\bm{x}) \leq 0, i = 1, 2, 3$, with a convex approximation of the form $\Bar{g}_i(\bm{x}, \bm{x}^{(n)}) \leq 0, i = 1, 2, 3$. For each of the approximations, it can be easily shown that the following properties hold~\cite{convproof}:
i) $g_i(\bm{x}) \leq \Bar{g}_i(\bm{x}, \bm{x}^{(n)})$ for all feasible $\bm{x}$;  
ii) $g_i(\bm{x}^{(n)}) = \Bar{g}_i(\bm{x}^{n}, \bm{x}^{(n)})$; and
$\frac{\partial g_i(\bm{x}^{(n)})}{\partial x_j} = \frac{\partial \Bar{g}_i(\bm{x}^{n}, \bm{x}^{(n)})}{\partial x_j}$, $j = 1, 2$. The constraints in \textbf{P6} also satisfy Slater's conditions~\cite{boyd2004convex}.

This implies that Algorithm~\ref{algo1}, by solving the inner-approximation problem, always converges to a KKT point of \textbf{P2} due to~\cite{convproof}. {It must be noted here that even though Algorithm~\ref{algo1} solves the approximate problem \textbf{P6} in each SCA iteration, it is provably optimal after sufficient number of iterations. This is due to the fact that it provably converges to a KKT point of \textbf{P2} which is an optimal solution~\cite{boyd2004convex}.}
\end{remark}
\subsection{Decentralized ADMM approach}
We now use ADMM to decentrally solve \textbf{P6}  in the second layer, an approach well-suited for CPUs with multiple distributed D-servers, connected via a central C-server~\cite{Decentralized1,Decentralized3}. The ADMM method decomposes a central problem into multiple sub-problems, each of which is solved by a D-server locally and independently. The C-server combines the local solutions to obtain a global solution. We observe that the constraints in \eqref{P6const1}-\eqref{P6const2} couple the power control coefficients of different uplink and downlink UEs. We next introduce global variables for the power control coefficients at the C-server, with local copies at the D-servers to decouple \textbf{P6} into sub-problems for each UE. We observe that the constraints in \textbf{P6} for the downlink and uplink UEs can be divided between downlink and uplink D-servers, respectively. The D-servers solve sub-problems defined for each  downlink and uplink UE. We first define local feasible sets at the $n$th SCA iteration for them, which are denoted as $\mathcal{S}_{k}^{d,(n)}$ and $\mathcal{S}_{l}^{u,(n)}$, respectively. These sets are given as follows
\begin{subequations} \label{S1}
\begin{align}
&\mathcal{S}_k^{d,(n)}\!\!\!\!\!\!\!\! =\!\! \left.\left\{f^d_k, \Psi^d_k, \zeta^d_k, \lambda^d_k,\bm{\widetilde{C}}^{d}_{k},\bm{\widetilde{\Theta}}^{d}_{k} \right. \right\vert \Tilde{b} \!\!\! \sum_{q \in \kappa_{dm}} \!\!\!\! \gamma^{d}_{mk} (\Tilde{c}^{d}_{mq,k})^{2} \!\!\leq\!\!  \frac{1}{N_t}, \!\!\!\!\!\! \label{S1const1} \\
\notag & \lambda^d_k \leq \sum\nolimits_{m \in \mathcal{M}^{d}_{k}}   A^d_{mk} \Tilde{c}^{d}_{mk,k},  \sum\nolimits_{m=1}^{M}  \sum\nolimits_{q \in \kappa_{dm}}  B^d_{kmq} (\Tilde{c}^d_{mq,k})^{2} \\
& + \sum\nolimits_{l=1}^{K_u}  D^d_{kl} \Tilde{\theta}^{d}_{l,k} + 1 \leq \Lambda^{(n)} \Big(\frac{(\lambda^{d}_{k})^2}{\zeta^d_k}\Big), \label{S1const2} \\
& (\Psi_{k}^{d})^{2} \! \leq \! \tau_f \log_2 (1 + \zeta^d_k), P_{\text{fix}} \! + \! N_t \rho_d N_0 \!\!\!\! \sum_{m \in \mathcal{M}^{d}_{k}} \!\!\! \frac{1}{\alpha_m} \gamma^{d}_{mk} (\Tilde{c}^{d}_{mk,k})^2 \! \nonumber \\
&  + \! P^{d}_{\text{tc},k} \leq \Lambda^{(n)} \left(\frac{(\Psi^{d}_{k})^{2}}{f^d_k}\right), \label{S1const3} \\
& \Tilde{c}^{d}_{mq,k} \geq 0 \, \forall q = 1 \text{ to } K_d, 0 \leq \Tilde{\theta}^{d}_{l,k} \leq  1, \lambda^d_k \geq 0, \nonumber \\
&\log_2(1 + \zeta^d_k) \geq S^d_{ok}/\tau_f\Big\}, \label{S1const4}\\
& \mathcal{S}_l^{u,(n)} \!\!\!\!\!\!\!\!= \!\!\left.\left\{f^u_l, \Psi^u_l, \zeta^u_l, \lambda^u_l,\bm{\widetilde{C}^u_{l},\widetilde{\Theta}^u_{l}} \right. \right\vert \Tilde{b} \!\!\!\! \sum_{k \in \kappa_{dm}} \!\!\!\!\!\! \gamma^{d}_{mk} (\Tilde{c}^{u}_{mk,l})^2 \!\!\leq\!\! \frac{1}{N_t},\!\!\!\! \label{S2const1} \\
& (\lambda^{u}_l)^2 \leq A^u_l \Tilde{\theta}^u_{l,l}, \sum\nolimits_{q=1}^{K_u} B^u_{lq} \Tilde{\theta}^u_{q,l} + \sum\nolimits_{i=1}^{M} \sum\nolimits_{k \in \kappa_{di}}   D^u_{lik} (\Tilde{c}^{u}_{ik,l})^2 \nonumber \\
& + E^u_l \Tilde{\theta}^u_{l,l} + F^u_l \leq \Lambda^{(n)} \left(\frac{(\lambda^{u}_{l})^2}{\zeta^u_l}\right), \label{S2const2} \\
& (\Psi_{l}^{u})^2 \leq \tau_f \log_2 (1 + \zeta^u_l), \nonumber \\
&P_{\text{fix}} + \rho_u N_0 \frac{1}{\alpha'_l} \Tilde{\theta}^u_{l,l} + P^{u}_{\text{tc},l} \leq \Lambda^{(n)} \left(\frac{(\Psi^{u}_{l})^2}{f^u_l}\right),  \label{S2const3} \\
& \Tilde{c}^u_{mk,l} \geq 0, 0 \leq \Tilde{\theta}^u_{q,l} \leq 1 \, \forall q = 1 \text{ to }  K_u, \nonumber \\
&\log_2(1 + \zeta^u_l)  \geq  S^u_{ol}/\tau_f \label{S2const4} \Big\}.
\end{align}
\end{subequations}

Here $\widetilde{\bm{C}}^{d}_{k}, \widetilde{\bm{C}}^{u}_{l} \in \mathbb{C}^{M \times K_{d}}$ and $\widetilde{\bm{\Theta}}^{d}_{k}, \widetilde{\bm{\Theta}}^{u}_{l} \in \mathbb{C}^{K_{u} \times 1}$ are  local copies at the D-server of the corresponding global variables at the C-server, which are denoted as $\widetilde{\bm{C}} \in \mathbb{C}^{M \times K_d}$ and  $\widetilde{\bm{\Theta}} \in \mathbb{C}^{K_u \times 1}$ respectively, and represent the downlink and uplink power control coefficients, $\bm{C}$ and $\bm{\Theta}$, in \textbf{P6}. We note that each D-server has its local power control variables and hence the constraints in \eqref{S1}, which are all convex, are independent for each D-server. This ensures that the sets $\mathcal{S}_k^{d,(n)}$ and $\mathcal{S}_l^{u,(n)}$ are convex. We define the sets of local variables for the D-servers corresponding to the downlink and uplink UEs as $\bm{\Omega}^{d}_{k} \triangleq [\widetilde{\bm{C}}^{d}_{k}, \widetilde{\bm{\Theta}}^{d}_{k}, f^{d}_{k}, \Psi^{d}_{k}, \lambda^{d}_{k}, \zeta^{d}_{k}]$ and $\bm{\Omega}^{u}_{l} \triangleq [\widetilde{\bm{C}}^{u}_{l},\widetilde{\bm{\Theta}}^{u}_{l}, f^{u}_{l}, \Psi^{u}_{l}, \lambda^{u}_{l}, \zeta^{u}_{l}]$ respectively. 

We now reformulate \textbf{P6} as follows
\begin{subequations}
\begin{alignat}{2}
\notag \!\!\!\! \textbf{P7}:& \underset{\substack{\bm{\widetilde{C}\text{, }\widetilde{\Theta}\text{, }\bm{\Omega}^{d}_{k}\text{, }\bm{\Omega}^{u}_{l}}}}{\mbox{max}} && \sum\nolimits_{k=1}^{K_d} w^d_k f^d_k  + \sum\nolimits_{l=1}^{K_u} w^u_l f^u_l  \\
& \;\; \text{ s.t. } \; && \bm{\Omega}^{d}_{k} \in \mathcal{S}_{k}^{d,(n)}, \bm{\Omega}^{u}_{l} \in \mathcal{S}_{l}^{u,(n)}, \label{P7const1} 
\end{alignat}
\begin{alignat}{2}
&&& \widetilde{\bm{C}}^d_k = \widetilde{\bm{C}} \text{, } \widetilde{\bm{C}}^u_l = \widetilde{\bm{C}}, \label{P7const2} \\
&&& \widetilde{\bm{\Theta}}^d_k = \widetilde{\bm{\Theta}}  \text{, } \widetilde{\bm{\Theta}}^u_l = \widetilde{\bm{\Theta}}.  \label{P7const3}
\end{alignat}
\end{subequations}
To ensure that the global variables at the C-server have identical local copies maintained at the D-servers, we introduce the consensus constraints \eqref{P7const2}-\eqref{P7const3}. The ADMM algorithm can now be readily applied to \textbf{P7} as it is in the global consensus form~\cite{ADMM}.

We use $\varepsilon \triangleq \{d,u\}$ to denote the downlink and uplink respectively, and  $\phi \triangleq \{k,l\}$ to denote $k$th the downlink UE and $l$th uplink UE, respectively. The sub-problems of individual D-servers can now be written as follows
\begin{subequations}
\begin{alignat}{2}
\notag \!\!\!\! \textbf{P7b}:& \underset{\substack{\bm{\widetilde{C}}\text{, } \bm{\widetilde{\Theta}}\text{, } \bm{\Omega}^{\varepsilon}_{\phi}}}{\mbox{max}} &&  w^{\varepsilon}_{\phi} f^{\varepsilon}_{\phi} \\
& \text{ s.t. } && \bm{\Omega}^{\varepsilon}_{\phi} \in \mathcal{S}_{\phi}^{\varepsilon,(n)}, \widetilde{\bm{C}}^{\varepsilon}_{\phi} = \widetilde{\bm{C}}, \widetilde{\bm{\Theta}}^{\varepsilon}_{\phi} = \widetilde{\bm{\Theta}}.\notag 
\end{alignat}
\end{subequations}

We now define auxiliary functions for the objective in \textbf{P7b} as follows
\begin{align} 
    q^{\varepsilon}_{\phi}(\bm{\Omega}^{\varepsilon}_{\phi})
    \triangleq 
    \begin{cases}
    w^{\varepsilon}_{\phi} f^{\varepsilon}_{\phi}, \, \bm{\Omega}^{\varepsilon}_{\phi} \in S_{\phi}^{\varepsilon,(n)}, \\
    -\infty, \, \text{otherwise}.
    \end{cases} \label{auxfn} 
\end{align}
We write, using \eqref{auxfn}, the augmented Lagrangian function for \textbf{P7} as
\begin{align}
\notag &\mathcal{L}^{(n)} \Big(\bm{\widetilde{C}, \widetilde{\Theta}}, \{\bm{\Omega^d_k, \chi^d_k, \xi^d_k}\}, \{\bm{\Omega^u_l, \chi^u_l, \xi^u_l}\} \Big) \\
\notag &= \sum\nolimits_{k=1}^{K_d}\Big(q^d_k(\bm{\Omega}^d_k) - \langle \bm{\chi}^d_k, \widetilde{\bm{C}}^d_k - \widetilde{\bm{C}}\rangle - \frac{\rho_C}{2} \|\widetilde{\bm{C}}^d_k - \widetilde{\bm{C}}\|^{2}_{F} \\
\notag &- \langle \bm{\xi}^d_k, \widetilde{\bm{\Theta}}^d_k - \widetilde{\bm{\Theta}}\rangle - \frac{\rho_\theta}{2} \|\widetilde{\bm{\Theta}}^d_k - \widetilde{\bm{\Theta}}\|^{2}
\Big) \\
\notag & + \sum\nolimits_{l=1}^{K_u}\Big(q^u_l(\bm{\Omega}^u_l) - \langle \bm{\chi}^u_l, \widetilde{\bm{C}}^u_l - \widetilde{\bm{C}} \rangle - \frac{\rho_C}{2} \|\widetilde{\bm{C}}^u_l - \widetilde{\bm{C}}\|^{2}_{F} \\
&- \langle \bm{\xi}^u_l, \widetilde{\bm{\Theta}}^u_l - \widetilde{\bm{\Theta}} \rangle - \frac{\rho_\theta}{2} \|\widetilde{\bm{\Theta}}^u_l - \widetilde{\bm{\Theta}}\|^{2} \Big), \label{lagrangian}
\end{align}
where $\rho_C, \rho_\theta > 0$ are the penalty parameters corresponding to the global variables $\widetilde{\bm{C}}$ and $\widetilde{\bm{\Theta}}$ respectively, and $\bm{\chi}^{\varepsilon}_{\phi} \in \mathbb{C}^{M \times K_d}, \bm{\xi}^{\varepsilon}_{\phi} \in \mathbb{C}^{K_u\times 1}$ are the Lagrangian variables associated with the equality constraints $\eqref{P7const2}$  and \eqref{P7const3}, respectively. The quadratic penalty terms are added to the objective to penalise equality constraints violations, and to enable the ADMM to converge by relaxing constraints of finiteness and strict convexity~\cite{ADMM}. 

We note that the augmented Lagrangian in \eqref{lagrangian} is not decomposable in general for the problem formulation in \textbf{P7b}~\cite{boyd2004convex}. The auxiliary functions defined in \eqref{auxfn} enable us to decompose it and formulate sub-problems for the D-servers. In ADMM method, the D-servers independently solve the sub-problems and update the local variables, which are collected by the C-server to update the global variables~\cite{ADMM}. In the $(p+1)$th iteration, following steps are executed in succession. \newline
1) \textit{Local computation}: The D-servers for each UE solve \textbf{P8} to update the local variables as 
\begin{alignat}{2}
    \notag  \textbf{P8}:&   \bm{\Omega}_{\phi}^{\varepsilon,(p+1)} =  \underset{\substack{\bm{\Omega}^{\varepsilon}_{\phi}}}{\argmax} \quad q^{\varepsilon}_{\phi}(\bm{\Omega}^{\varepsilon}_{\phi}) - \langle \bm{\chi}^{\varepsilon,(p)}_{\phi}, \widetilde{\bm{C}}^{\varepsilon}_{\phi} - \widetilde{\bm{C}}^{(p)} \rangle \\
    \notag & - \langle \bm{\xi}^{\varepsilon,(p)}_{\phi}, \widetilde{\bm{\Theta}}^{\varepsilon}_{\phi} - \widetilde{\bm{\Theta}}^{(p)} \rangle \\
     &- \frac{\rho^{(p)}_C}{2} \|\widetilde{\bm{C}}^{\varepsilon}_{\phi} - \widetilde{\bm{C}}^{(p)}\|^{2}_{F} - \frac{\rho^{(p)}_\theta}{2} \|\widetilde{\bm{\Theta}}^{\varepsilon}_{\phi} - \widetilde{\bm{\Theta}}^{(p)}\|^{2}. \label{locvarupdate}
\end{alignat}
2)  \textit{Lagrangian multipliers update}: The D-servers now update the Lagrangian multipliers as
\begin{align}
    \bm{\chi}_{\phi}^{\varepsilon,(p+1)} &= \bm{\chi}_{\phi}^{\varepsilon,(p)} + \rho^{(p)}_C (\widetilde{\bm{C}}^{\varepsilon,(p+1)}_{\phi} - \widetilde{\bm{C}}^{(p)}) \label{lagupdates1} \\ 
    \bm{\xi}_{\phi}^{\varepsilon,(p+1)} &= \bm{\xi}_{\phi}^{\varepsilon,(p)} + \rho^{(p)}_\theta (\widetilde{\bm{\Theta}}^{\varepsilon,(p+1)}_{\phi} - \widetilde{\bm{\Theta}}^{(p)}). \label{lagupdates2} 
\end{align}
3) \textit{Global aggregation and computation}: The C-server now collects the updated local variables and Lagrangian multipliers from the D-servers and updates the global variables $\{\widetilde{\bm{C}}, \widetilde{\bm{\Theta}}\}$. 
    \begin{alignat}{2}
        \notag \textbf{P9}: & \{\widetilde{\bm{C}}, \widetilde{\bm{\Theta}}\}^{(p+1)} =   \underset{\substack{\bm{\widetilde{C}, \widetilde{\Theta}}}}{\argmax} \quad \mathcal{L}^{(n)}  \Big(\bm{\widetilde{C}, \widetilde{\Theta}}, \\
        \notag &\{\bm{\Omega^{d}_k, \chi^{d}_k, \xi^{d}_k}\}^{(p+1)}, \{\bm{\Omega^{u}_l, \chi^{u}_l, \xi^{u}_l}\}^{(p+1)} \Big). 
    \end{alignat}
    Using \eqref{lagrangian} and maximizing w.r.t. each global variable, we obtain a closed form solution
    \begin{align}
        \widetilde{\bm{C}}^{(p+1)} &= \frac{1}{{K}} \Big(\sum\nolimits_{k=1}^{K_d} \Big[\widetilde{\bm{C}}^{d,(p+1)}_k + \frac{1}{\rho^{(p)}_C} \bm{\chi}^{d,(p+1)}_k\Big] \nonumber \\
        &+ \sum\nolimits_{l=1}^{K_u} \Big[\widetilde{\bm{C}}^{u,(p+1)}_l + \frac{1}{\rho^{(p)}_C} \bm{\chi}^{u,(p+1)}_l \Big] \Big ), \label{globalupdate2}\\
        \widetilde{\bm{\Theta}}^{(p+1)} &= \frac{1}{{K}} \Big( \sum\nolimits_{k=1}^{K_d} \Big[\widetilde{\bm{\Theta}}^{d,(p+1)}_k + \frac{1}{\rho^{(p)}_\theta} \bm{\xi}^{d,(p+1)}_k \Big] \nonumber \\
        &+ \sum\nolimits_{l=1}^{K_u} \Big[\widetilde{\bm{\Theta}}^{u,(p+1)}_l + \frac{1}{\rho^{(p)}_\theta} \bm{\xi}^{u,(p+1)}_l \Big] \Big). \label{globalupdate3}
    \end{align}
    {The updated global variables in \eqref{globalupdate2}-\eqref{globalupdate3} are broadcasted by the C-server to all the D-servers.}\newline
    4) \textit{Residue calculation and penalty parameter updates}: The C-server calculates the squared magnitude of the primal and dual residuals, denoted as $\bm{r}_{\text{ADMM}}$ and $\bm{s}_{\text{ADMM}}$ respectively, as~\cite{ADMM} 
    \begin{align}
       \!\!\!\! \|\bm{r}_{\text{ADMM}}^{(p+1)}\|_2^{2} \! &= \! \sum\nolimits_{k=1}^{K_d} \! \Big(\|\widetilde{\bm{C}}_{k}^{d} - \widetilde{\bm{C}}\|^{2}_{F} + \|\widetilde{\bm{\Theta}}_{k}^{d} - \widetilde{\bm{\Theta}}\|_2^{2}\Big)^{(p+1)} \!\!\!\! \nonumber \\
       &+\! \sum\nolimits_{l=1}^{K_u} \! \left(\|\widetilde{\bm{C}}_{l}^{u} - \widetilde{\bm{C}}\|^{2}_{F} + \|\widetilde{\bm{\Theta}}_{l}^{u} - \widetilde{\bm{\Theta}}\|_2^{2}\right)^{(p+1)}\!\!\!\!, \!\!\!\! \label{presADMM}  \\
       \!\! \|\bm{s}_{\text{ADMM}}^{(p+1)}\|_2^{2} \! &=  \! K \Big(\|\widetilde{\bm{C}}^{(p+1)} \!\! - \!\! \widetilde{\bm{C}}^{(p)}\|_{F}^{2} \!+\! \|\widetilde{\bm{\Theta}}^{(p+1)} \!\!-\!\! \widetilde{\bm{\Theta}}^{(p)}\|_2^{2}\Big). \!\!\!\! \label{dresADMM} 
    \end{align}
    The C-server now compares the primal and dual residual norms obtained in~\eqref{presADMM}-\eqref{dresADMM}. To accelerate convergence, it updates the penalty parameters for the $(p+1)$th ADMM iteration, $\rho^{(p+1)}_{\{C\}}$ and $\rho^{(p+1)}_{\{\theta\}}$, appropriately as follows~\cite{ADMMVar}:
    \begin{equation}
        \rho^{(p+1)}_{\{C,\theta\}} = \begin{cases}
        \rho^{(p)}_{\{C,\theta\}}\vartheta^{\text{incr}} , \, \|r^{(p+1)}\|_2 > \mu \|s^{(p+1)}\|_2,  \\
        \rho^{(p)}_{\{C,\theta\}}/\vartheta^{\text{decr}} , \, \|s^{(p+1)}\|_2 > \mu \|r^{(p+1)}\|_2, \\
        \rho^{(p)}_{\{C,\theta\}}, \, \text{otherwise}. 
        \end{cases}
        \label{penaltyupdate}
    \end{equation}
    The parameters $\mu > 1, \vartheta^{\text{incr}} > 1, \vartheta^{\text{decr}} > 1$ are tuned to obtain good convergence~\cite{ADMMVar}. 

\textbf{Initialization for ADMM}: At the $(n+1)$th SCA iteration, we initialize the global variables at the C-server and their local copies at the D-servers with the SCA iteration variables as 
\begin{align}
    \Tilde{c}^{(1)}_{mk} &= c^{(n+1)}_{mk}, 
    \Tilde{\theta}^{(1)}_{l} = \theta^{(n+1)}_{l}, 
    \widetilde{\bm{C}}^{d,(1)}_k = \widetilde{\bm{C}}^{u,(1)}_l = \widetilde{\bm{C}}^{(1)},  \notag \\
    \widetilde{\bm{\Theta}}^{d,(1)}_k &= \widetilde{\bm{\Theta}}^{u,(1)}_l = \widetilde{\bm{\Theta}}^{(1)}. \label{initgloballocal}
\end{align}

\begin{table*}[t]
	\footnotesize
	\centering
	\begin{tabular}{|c|c|} 
		\hline
		Parameter & Value \\ [0.5ex] 
		\hline 
		Coverage area side length, $D$ & 1 km \\
		\hline
		Shadowing parameters $\sigma_{\text{sd}}$, $\delta$ & 2 dB, 0.5  \\
		\hline 
		Bandwidth $B$ & 20 Mhz \\
		\hline
		Length of coherence period, $\tau_c$, coherence time  $T_c$ & 200 symbols, 1 ms\\ 
		\hline
		Fronthaul parameters $\nu$, $C_{\text{fh}}$ & $2$, {$10$ Mbps} \\
		\hline
		RI parameters $\gamma_{\text{RI}}$, $\text{PL}_{\text{RI}}$ (in dB) & $-20$, $-81.1846$ \\
		\hline
		AP power parameters, $P_{\text{ft}}, P_{0,m}, P_{\text{tc},m}$ (in W) & 10, 0.825, 0.2\\
		\hline
		UE power parameters, $P^{d}_{\text{tc},k} = P^{u}_{\text{tc},l}$ & 0.2 W\\
		\hline 
		Pilot power $p_t$, Noise power $N_0$ & 0.2 W, -121.4 dB \\
		\hline 
		Power amplifier efficiencies, $\alpha_m, \alpha'_l$ & 0.39, 0.3 \\[1ex] 
		\hline
	\end{tabular}
	\caption{{Full-Duplex Cell-Free mMIMO system model and power consumption model parameters.}}
	\label{sysparams}
\end{table*}

\textbf{ADMM Convergence Criterion}: The ADMM can be said to have converged at iteration $P$ if the primal residue is within a pre-determined tolerance limit $\epsilon_{\text{ADMM}}$ i.e.,  $\|r^{(P)}\|_2 \leq \epsilon_{\text{ADMM}}$. 
The steps \eqref{locvarupdate}, \eqref{lagupdates1}-\eqref{lagupdates2}, \eqref{globalupdate2}-\eqref{globalupdate3} and \eqref{penaltyupdate} are iterated until convergence, after which we obtain the locally optimal power control coefficients $\{\widetilde{\bm{C}}^{*},\widetilde{\bm{\Theta}}^{*}\}$. We assign them to the iterates for the $(n+1)$th SCA iteration, i.e., $\bm{C}^{(n+1)} = \widetilde{\bm{C}}^{*}, \bm{\Theta}^{(n+1)} = \widetilde{\bm{\Theta}}^{*}$. This concludes the $n$th SCA iteration. 
The SCA is iterated till convergence. The steps for the decentralized WSEE maximization using SCA and ADMM are summarized in Algorithm~\ref{algo2}.

\begin{algorithm} [h]
	\footnotesize
\DontPrintSemicolon 
\KwIn{i) Initialize power control coefficients for SCA, $\{\bm{C,\Theta}\}^{(1)}$ by allocating equal power to downlink UEs and maximum power to uplink UEs. Set $n = 1$. Initialize $\{\bm{f^d,f^u,\Psi^d,\Psi^u,\zeta^d,\zeta^u,\lambda^d,\lambda^u}\}^{(1)}$ by replacing inequalities \eqref{P5const3}, \eqref{P6const1}-\eqref{P6const2}, \eqref{P4const2} and \eqref{P6const3}-\eqref{P6const4} by equality, in turn.}
\KwOut{Globally optimal power control coefficients $\{\bm{C,\Theta}\}^{*}$}
\While{$\|\bm{r}_{\text{SCA}}\| \leq \epsilon_{\text{SCA}}$}{
Set $p = 1$. Initialize global variables at C-server, $\{\bm{\widetilde{C},\widetilde{\Theta}}\}^{(1)}$, and local variables at D-servers, $\bm{\Omega}^{\varepsilon,(1)}_{\phi}$, using \eqref{initgloballocal} and replacing inequalities \eqref{S1const2}-\eqref{S1const3} and \eqref{S2const2}-\eqref{S2const3} by equality. \\
\While{$\|\bm{r}_{\text{ADMM}}\| \leq \epsilon_{\text{ADMM}}$}{
Substitute $\{\bm{C, \Theta, f^d,f^u,\Psi^d,\Psi^u,\zeta^d,\zeta^u,\lambda^d,\lambda^u}\}^{(n)}$ in \eqref{S1} to obtain feasible sets $\mathcal{S}^{\varepsilon,(n)}_{\phi}$. \linebreak 
Solve \textbf{P8} at respective D-servers to update local variables $\bm{\Omega}^{\varepsilon,(p+1)}_{\phi}$. \linebreak
Solve \eqref{lagupdates1}-\eqref{lagupdates2} at respective D-servers to update Lagrangian multipliers $\{\bm{\chi,\xi}\}^{\varepsilon,(p+1)}_{\phi}$. \linebreak 
At the C-server, collect the local variables $ \{\widetilde{\bm{C}},\widetilde{\bm{\Theta}}\}^{\varepsilon,(p+1)}_{\phi}$, and the Lagrangian multipliers, $\{\bm{\chi,\xi}\}^{\varepsilon,(p+1)}_{\phi}$, from the D-servers and solve \eqref{globalupdate2}-\eqref{globalupdate3} to update the global variables $\widetilde{\bm{C}}^{(p+1)}, \widetilde{\bm{\Theta}}^{(p+1)}$. \linebreak 
At the C-server,  update penalty parameters $\rho^{(p+1)}_{C,\theta}$ according to \eqref{penaltyupdate} and broadcast them to all  D-servers.}

Update $\bm{C}^{(n+1)} = \widetilde{\bm{C}}^{*}, \bm{\Theta}^{(n+1)} = \widetilde{\bm{\Theta}}^{*}$ and obtain $\{\bm{f^d,f^u,\lambda^d,\lambda^u,\Psi^d, \Psi^u,\zeta^d, \zeta^u}\}^{(n+1)}$ by replacing the inequalities \eqref{P5const3}, \eqref{P6const1}-\eqref{P6const2}, \eqref{P4const2} and \eqref{P6const3}-\eqref{P6const4} by equality.}

\Return{$\{\bm{C,\Theta}\}^{*}$.}
\caption{Decentralized WSEE maximization algorithm using SCA and ADMM}\label{algo2}
\end{algorithm}
\begin{remark}\textit{Convergence of proposed decentralized algorithm:}
	Algorithm~\ref{algo2} uses the iterative SCA technique with each SCA iteration involving ADMM. The algorithm is thus guaranteed to converge if both SCA and ADMM converge. {It must be noted here that Algorithm~\ref{algo2}, despite solving an approximate problem \textbf{P7} in each ADMM iteration, indeed converges to an optimal solution of the original problem \textbf{P2}. This is explained as follows. For a given SCA iteration, the convergence of ADMM is guaranteed and investigated in detail in~\cite{ADMM}. Hence, every SCA iteration converges to an optimal solution of the approximate problem \textbf{P6}. As discussed in Remark~\ref{remark1}, the SCA iterative procedure provably converges to a KKT point of \textbf{P2} which is an optimal solution~\cite{boyd2004convex}.}
\end{remark}
\begin{remark}{\textit{Implementability:}} The maximal ratio combiner/beamformer considered herein is  the {simplest} receiver/transmitter for a distributed cell-free mMIMO system~\cite{CellFreeMassiveMIMOBook}.  Further, the power optimization  algorithms require only long-term fading channel coefficients, which remain constant for hundreds of coherence intervals~\cite{AGWireless}. This is in contrast to the existing work in SE-GEE maximization of FD cell-free massive MIMO systems in~\cite{FDCellFree2}, which requires instantaneous channel.  The current optimization problem whose reduced complexity is discussed below, therefore, needs to be solved over a relaxed time frame, which makes it easily implementable.
\end{remark}
\subsection{Computational complexity of centralized and decentralized algorithms} \label{complexitycompare}
Before beginning this study, it is worth noting that both centralized Algorithm~\ref{algo1} and  decentralized Algorithm~\ref{algo2} comprise of multiple steps that involve solving simple closed form expressions. These steps consume much lesser time than the ones which  solve a GCP, typically using interior points methods~\cite{boyd2004convex}. We therefore compare the per-iteration complexity of centralized and decentralized algorithms by calculating the  complexity of solving the respective GCPs. 
\begin{itemize}
    \item Algorithm~\ref{algo1} solves \textbf{P6} in step-1 of each SCA iteration, which has $4(K_u+K_d) + K_u + MK_d$ real variables and $6(K_u + K_d) + M + MK_d$ linear constraints. It has a worst-case computational complexity $\mathcal{O}((10(K_u\!+\!K_d)\! +\! K_u \!+\! M \!+\! 2MK_d)^{3/2}(4(K_u\!+\!K_d)\!+\! K_u \!+\! MK_d)^{2} \!)$~\cite{cvxproofbook}.
    \item Algorithm~\ref{algo2}, in step-2 of each ADMM iteration, solves \textbf{P8} at the D-servers \textit{in parallel} to update the local variables. We, therefore, need to analyse the computational complexity at \textit{any one of the} D-servers. Since the downlink has an additional constraint (second one in \eqref{S1const4}), we consider a downlink D-server for worst-case complexity analysis, which in \textbf{P8} has $MK_d + K_u + 4$ real variables and $MK_d + M + K_u + 6$ linear constraints. It will have a worst-case computational complexity ~\cite{cvxproofbook}: $\mathcal{O}(\left(2MK_d + M + 2K_u + 10 \right)^{3/2} \left(MK_d + K_u + 4)^{2}\right)$.
\end{itemize}

We consider  $K_d = K_u = K/2$ uplink and downlink UEs  for this analysis. We observe that for a  large $K$, Algorithm~\ref{algo2} has  a much lower computational complexity than Algorithm~\ref{algo1}. 
\section{Simulation results} \label{simresults}
We now numerically investigate the SE and WSEE of a FD CF mMIMO system with limited-capacity fronthaul links. We assume a realistic system model wherein the $M$ APs, $K_d$ downlink UEs and $K_u$ uplink UEs are all scattered randomly in a square of size $D$ km $\times$ $D$ km. To avoid the boundary effects~\cite{CFvsSmallCells}, we wrap the APs and UEs around the edges~\cite{FDCellFree}.  We use $\varepsilon \triangleq \{d,u\}$ to denote downlink and uplink respectively, and  $\phi \triangleq \{k,l\}$ to denote $k$th downlink UE and $l$th uplink UE, respectively. The large-scale fading coefficients, $\beta^{\varepsilon}_{m\phi}$, are modeled as~\cite{CellFreeEE}
\begin{align}
    \beta^{\varepsilon}_{m\phi} = 10^{\frac{\text{PL}^{\varepsilon}_{m\phi}}{10}} 10^{\frac{\sigma_{\text{sd}} z^{\varepsilon}_{m\phi}}{10}} \label{pathloss}.
\end{align}
Here $10^{\frac{\sigma_{\text{sd}} z^{\varepsilon}_{m\phi}}{10}}$ is the log-normal shadowing factor with a standard deviation $\sigma_{\text{sd}}$ (in dB) and $z^{\varepsilon}_{m\phi}$ follows a two-components correlated model~\cite{CFvsSmallCells}. The path loss $\text{PL}^{\varepsilon}_{m\phi}$ (in dB) follows a three-slope model~\cite{CFvsSmallCells, FDCellFree}.

We, similar to~\cite{FDCellFree}, model the large-scale fading coefficients for the inter-AP RI channels, i.e., $\beta_{\text{RI}, mi}, \, \forall i \neq m$, as in~\eqref{pathloss}, and assume that the large-scale fading for the intra-AP RI channels, which do not experience shadowing, are modeled as $\beta_{\text{RI},mm} = 10^{\frac{\text{PL}_{\text{RI}} \text{(dB)}}{10}}$. The inter-UE large scale fading coefficients, $\Tilde{\beta}_{kl}$, are also modeled similar to~\eqref{pathloss}. We consider, for brevity, the same number of quantization bits $\nu$, and the same fronthaul capacity $C_{\text{fh}}$ for all links. We, henceforth, denote the transmit powers on the downlink and uplink as $p_{d}$ $(= \rho_{d} N_0)$ and $p_{u}$ $(= \rho_{u} N_0)$, respectively, and the pilot transmit power as $p_t (= \rho_t N_0)$.  We fix the  system model values and power consumption model parameters, unless mentioned otherwise, as given in Table~\ref{sysparams}. These values are commonly used in the literature e.g.,\cite{CFvsSmallCells,FDCellFree,CellFreeMaxMinUQ}.

\textbf{Validation of SE expressions:} We consider an FD CF mMIMO system with i) {$M = \{16, 32\}$} APs, each having $N_t = N_r = 8$ transmit and receive antennas, $K_d = 12$ downlink UEs and $K_u = 8$ uplink UEs; and ii) unequal uplink and downlink transmit power i.e., $p_d = 2p_u = p$. We verify in Fig.~\ref{fig:2a} the tightness of the SE lower bound derived in \eqref{dlrate}-\eqref{ulrate}, labeled as LB, by comparing it with the numerically-obtained ergodic SE in \eqref{ergodrate}, labeled as upper-bound (UB) as it requires instantaneous CSI. The large-scale fading coefficients are set according to a practical FD CF channel model with parameters  specified in Table~\ref{sysparams}. We, similar to~\cite{CFvsSmallCells,CellFreeEE}, allocate equal power to all downlink UEs  and full power to all uplink UEs, i.e., $\eta_{mk} = \left({b}N_t\left(\sum_{k \in \kappa_{dm}} \gamma^{d}_{mk}\right)\right)^{-1}, \forall  k \in \kappa_{dm}$ and $\theta_l = 1$. \textit{We see that the derived lower bound is tight for both values of $M$.}
\begin{figure}[ht]
    \centering
    \includegraphics[height = 0.3\textwidth, width = 0.3\textwidth]{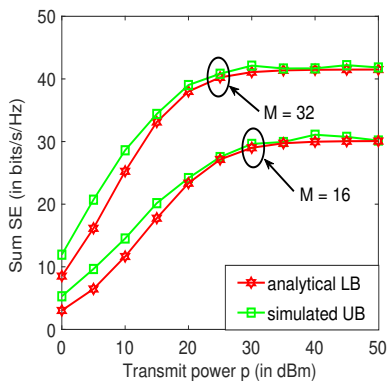}
    \caption{Sum SE vs transmit power, with $N_t = N_r = 8, K_d = 12, K_u = 8$}
    \label{fig:2a}
\end{figure}
\begin{figure}[ht]
	\centering
	\begin{subfigure}{.3\textwidth}
		\centering
		\includegraphics[height = \linewidth, width=\linewidth]{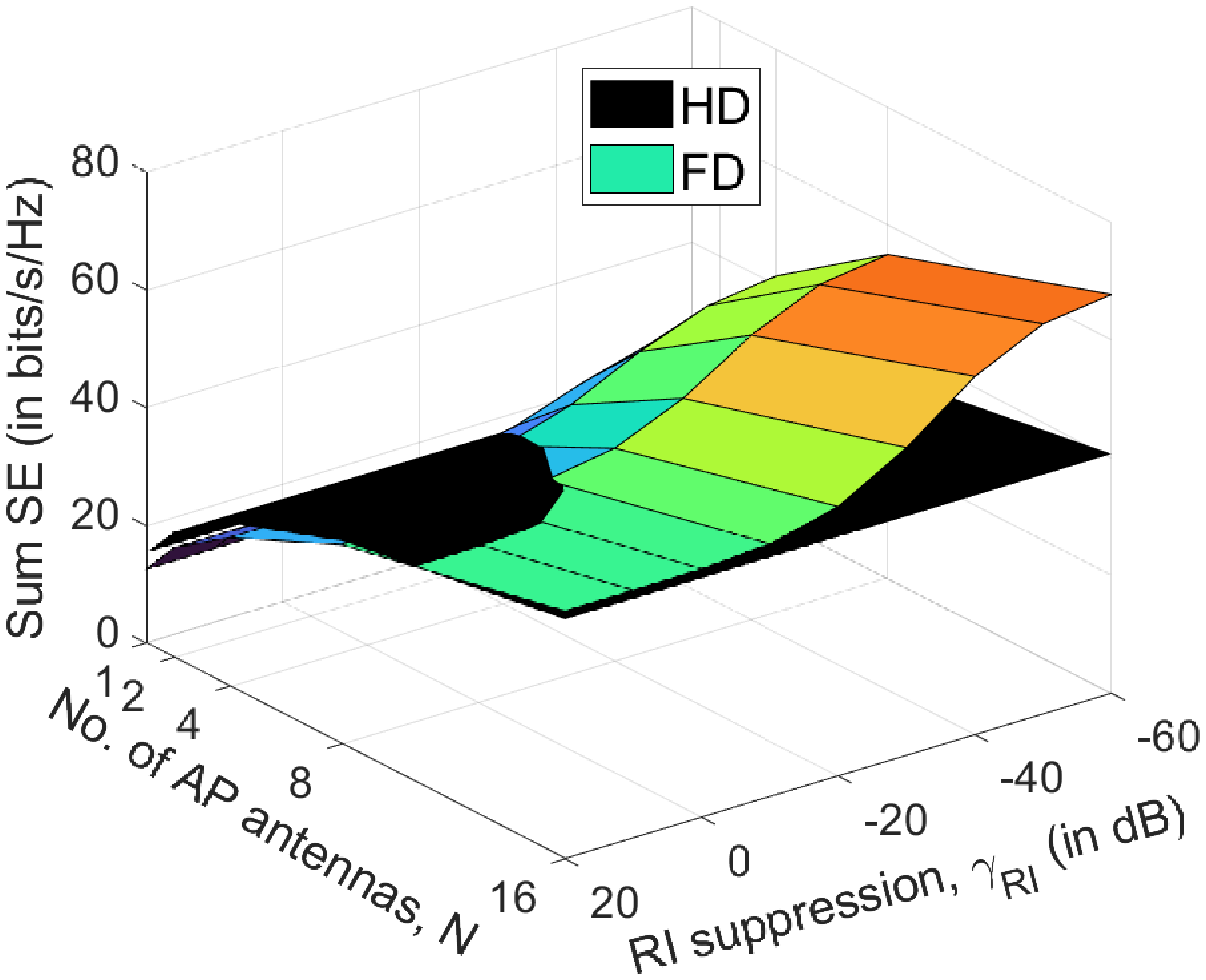}
		\caption{}
		\label{fig:2b}
	\end{subfigure}
	\begin{subfigure}{.3\textwidth}
		\centering
		\includegraphics[height = \linewidth, width=\linewidth]{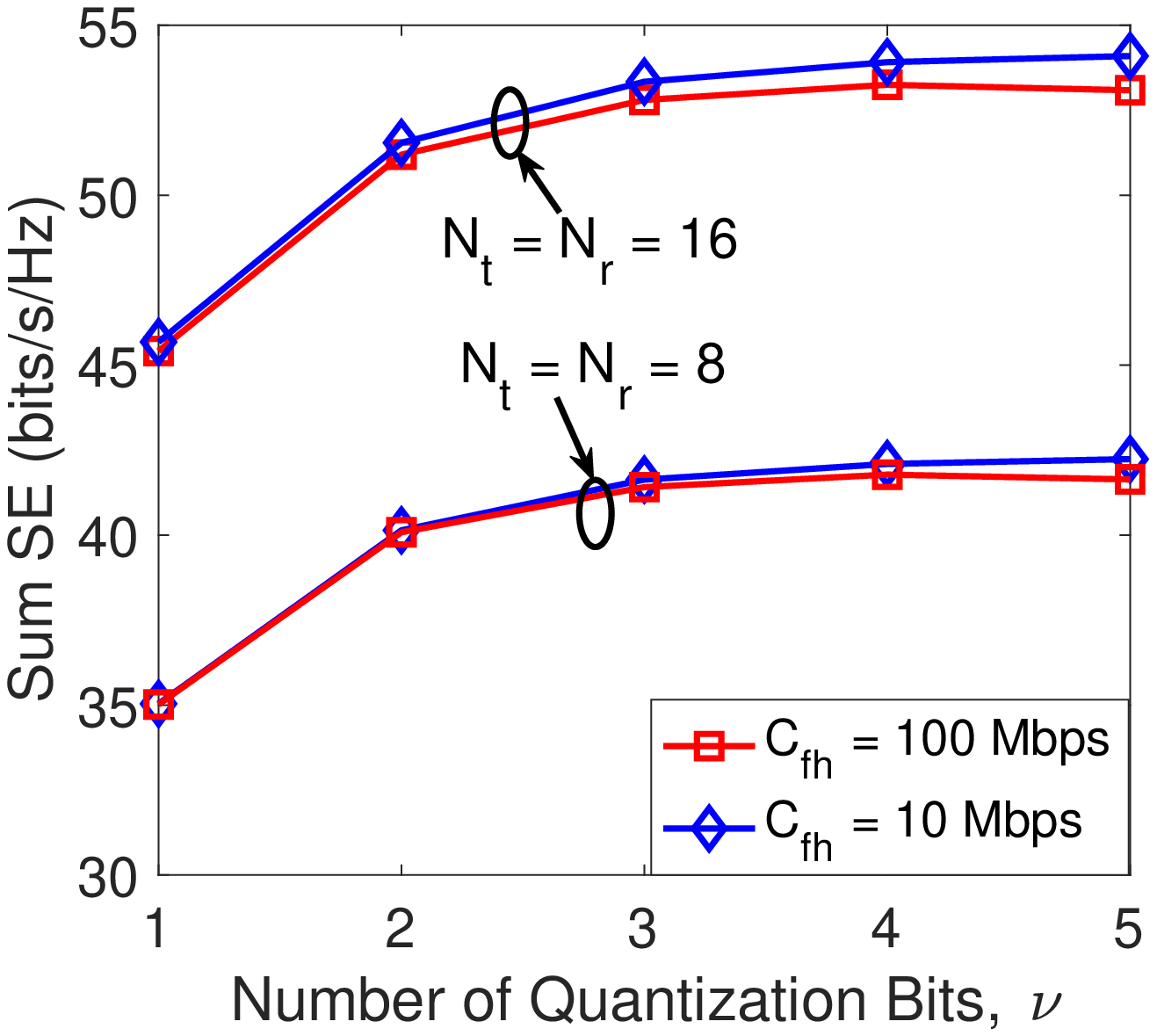}
		\caption{}
		\label{fig:2c}
	\end{subfigure}
	\begin{subfigure}{0.3\textwidth}
        \centering
        \includegraphics[height = \linewidth, width = \linewidth]{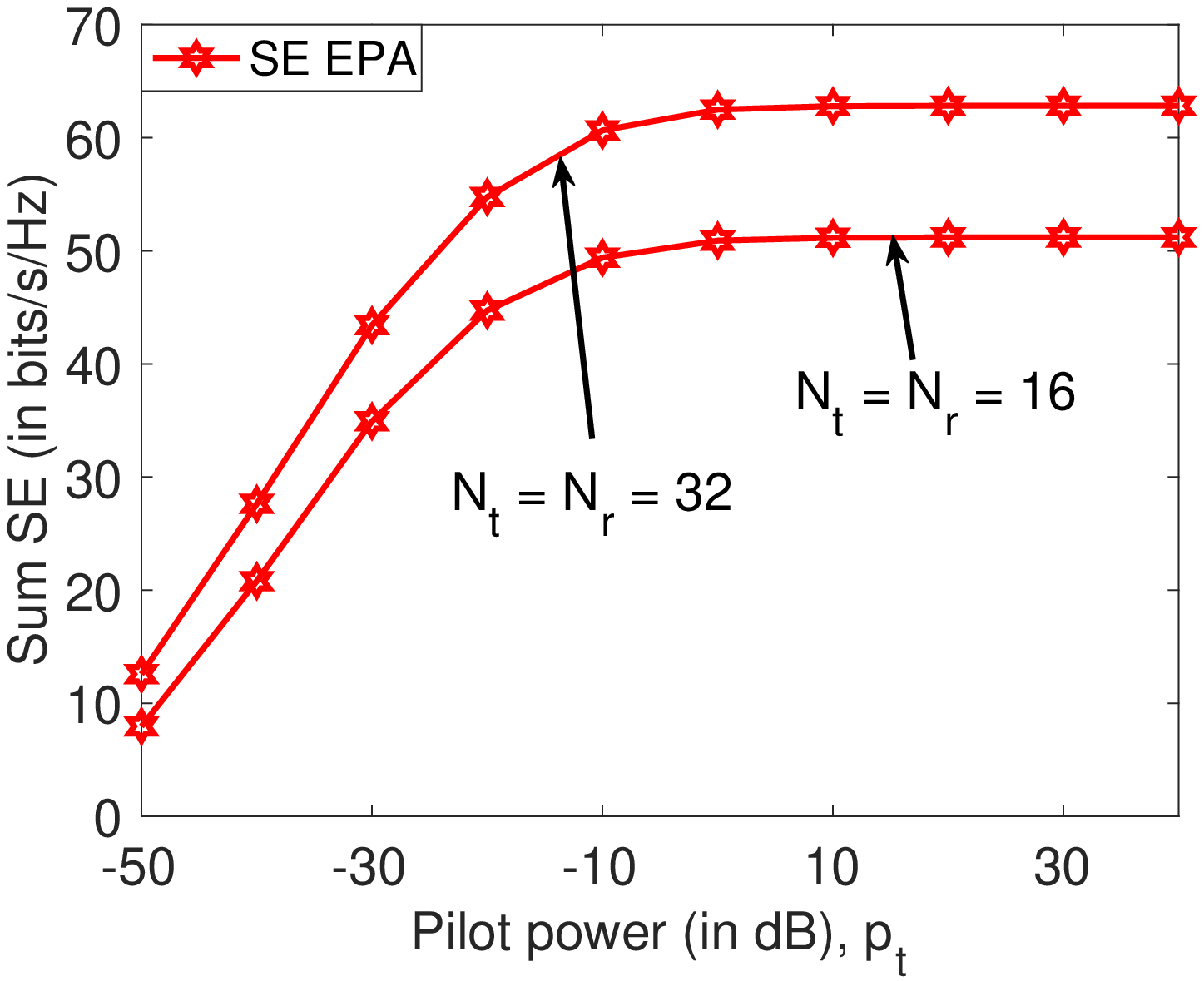}
        \caption{}
        \label{fig:pilot}
    \end{subfigure}
	\caption{
	Sum SE vs a) RI suppression levels, b) Number of quantization bits, and c) pilot power with $M = 32$, $K_d = 12$, $K_u = 8, p_d = 2p_u = 30$ dBm.}
	\label{fig:2}
\end{figure}
\textbf{Sum SE - FD and HD comparison:} We consider an FD CF mMIMO system with $M = 32$ APs, $K_d = 12$ downlink UEs, $K_u = 8$ uplink UEs and with transmit powers $p_d = 30$~dBm, $p_u = 27$~dBm on the downlink and uplink. We compare in Fig.~\ref{fig:2b} the FD CF mMIMO system with varying levels of RI suppression factor $\gamma_{\text{RI}}$ and an equivalent HD system  which serves uplink and downlink UEs in time-division duplex {mode.} 
For the HD system, we  i) set $\gamma_{\text{RI}} = 0$ and inter-UE channel gains $\Tilde{\beta}_{kl} = 0$; ii) use all AP antennas, i.e., $N = (N_t + N_r)$, during uplink and downlink transmission; and iii) multiply sum SE with a factor of $1/2$. 
We see that the FD system has a significantly higher sum SE than an equivalent HD system, provided the RI suppression is good i.e., $\gamma_{\text{RI}} \leq -10$ dB. {It is important to reemphasize here that the gains in sum SE achieved by the FD transmissions completely vanish with poor RI suppression i.e., $\gamma_{\text{RI}} > -10$ dB.  Moreover we note that, contrary to intuitive expectations, the} sum SE does not double, even with significant RI suppression $\gamma_{\text{RI}} \leq -40$ dB. This is due to the UDI experienced by the downlink UEs in a FD CF mMIMO system as shown in Fig.~\ref{fig:0}, which cannot be mitigated by RI suppression at APs. 

\textbf{Sum SE - variation with quantization bits:} We plot in Fig.~\ref{fig:2c} the  sum SE by varying the number of fronthaul quantization bits $\nu$. We consider $M = 32$ APs, $K_d = 12$ downlink UEs , $K_u = 8$ uplink UEs, and $p_d = 2p_u = 30$ dBm power for downlink and uplink, $N_t = N_r = \{8,16\}$ transmit and receive antennas on each AP, and fronthaul capacities $C_{\text{fh}} = \{10, 100\}$ Mbps. We observe that for both antenna configurations, sum SE increases with increase in $\nu$ initially and then saturates. Increasing $\nu$ reduces the quantization distortion and attenuation, which improves the sum SE. This effect, however, saturates as after a limit most of the information is retrieved. We observe that reducing the fronthaul capacity from $C_{\text{fh}} = 100$ Mbps to $C_{\text{fh}} = 10$~Mbps reduces the sum SE slightly, as the procedure outlined in Section~\ref{apsellimfh} \textit{fairly} retains the AP-UE links with the highest channel gains and helps maintain the sum SE.

\begin{figure}[ht]
	\centering
	\begin{subfigure}{.3\textwidth}
		\centering
		\includegraphics[height = \linewidth, width=\linewidth]{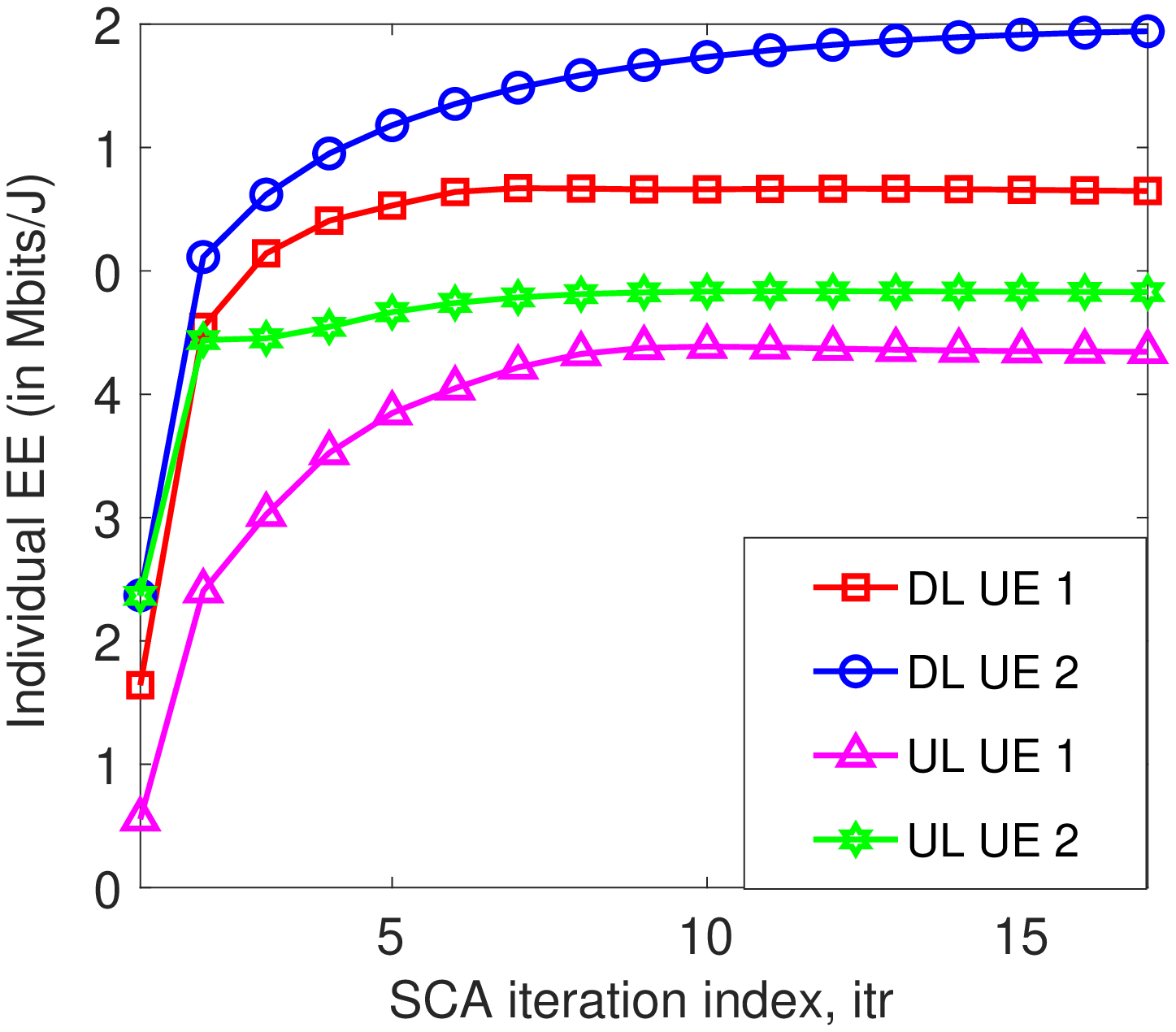}  
		\caption{} 
		\label{fig:4a}
	\end{subfigure}
	\begin{subfigure}{.3\textwidth}
		\centering
		\includegraphics[height = \linewidth, width=\linewidth]{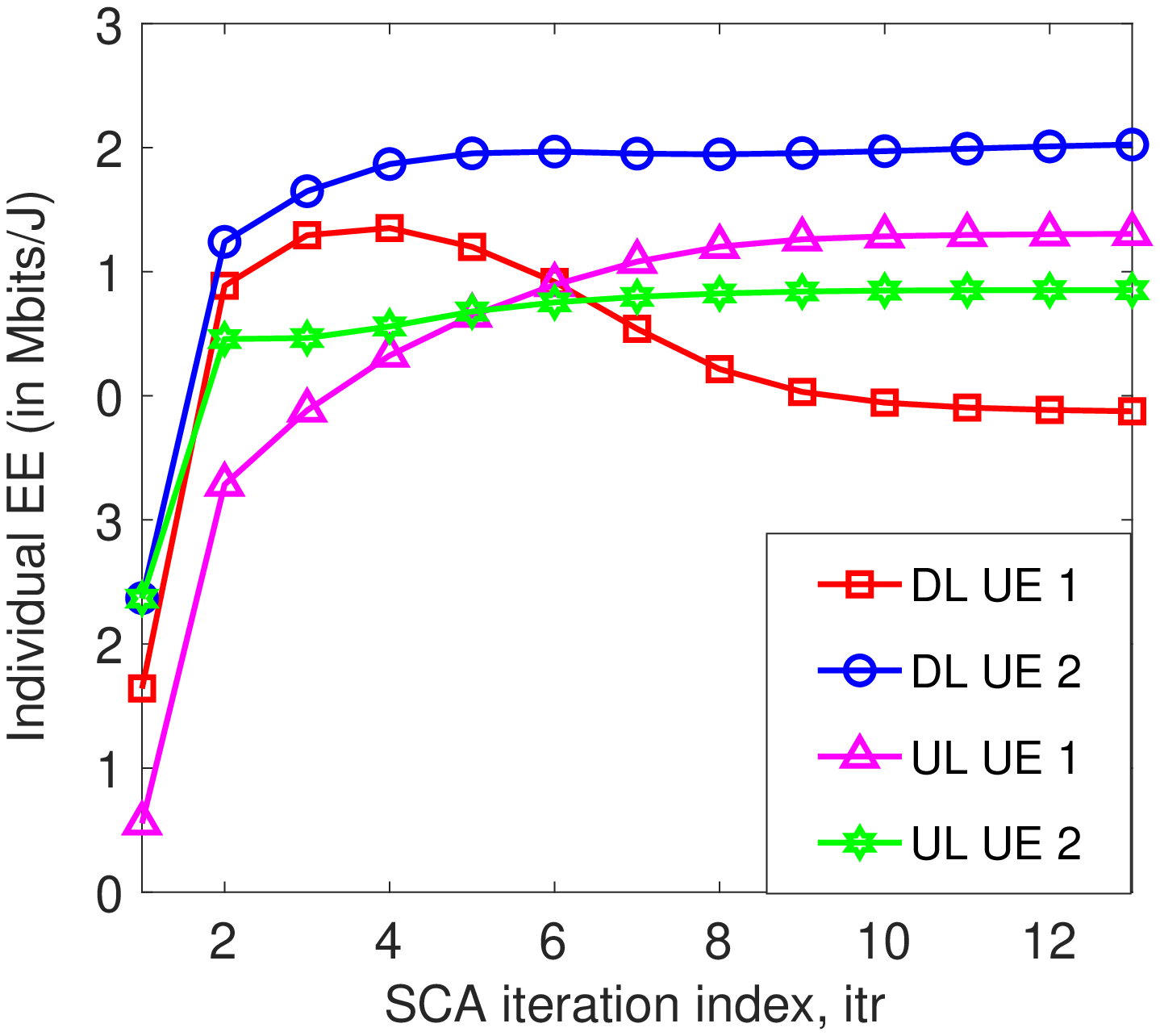}
		\caption{}
		\label{fig:4b}
	\end{subfigure}
	\begin{subfigure}{.3\textwidth}
		\centering
		\includegraphics[height = \linewidth, width =\linewidth]{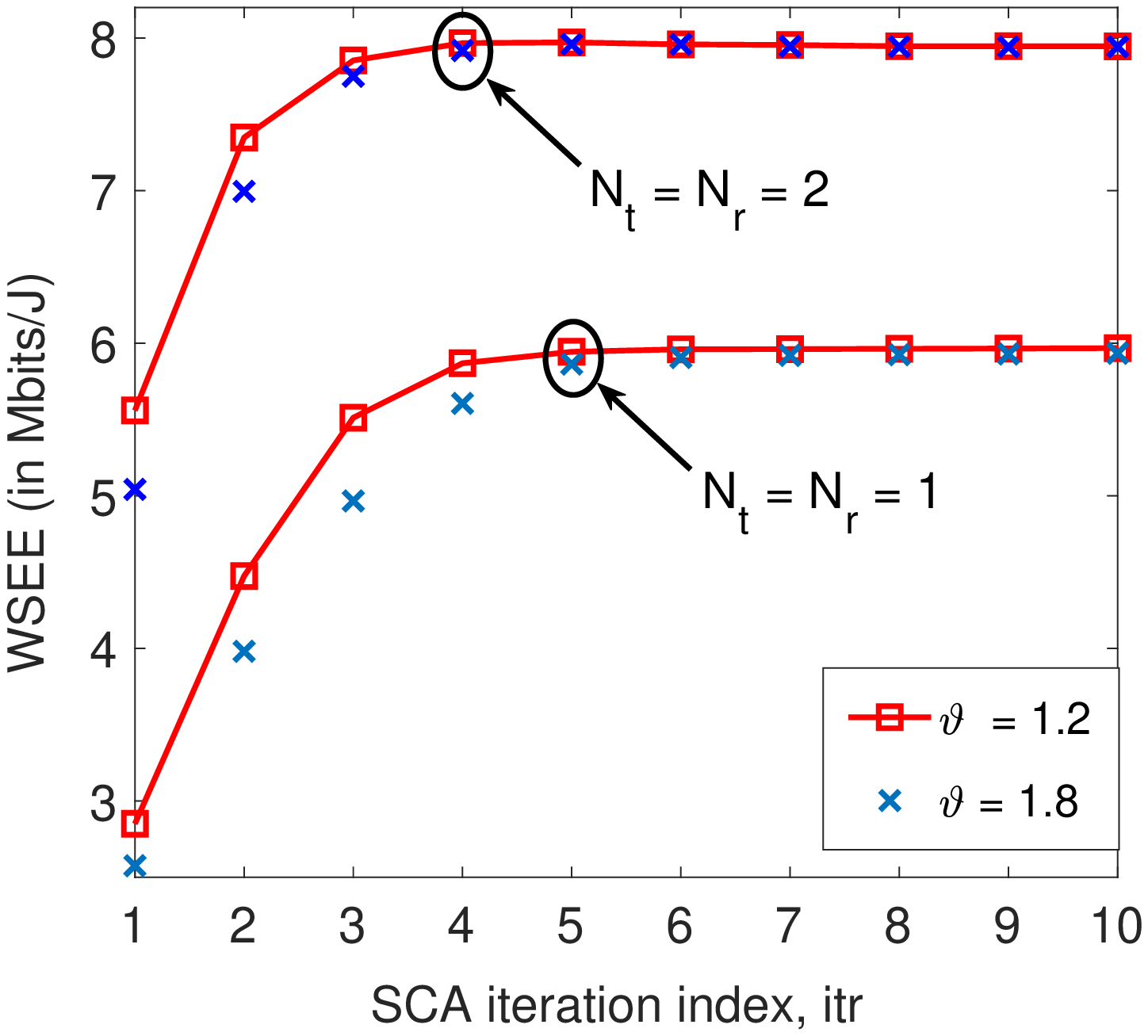}
		\caption{}
		\label{fig:4c}
	\end{subfigure}
	\caption{Effect of UE priorities on individual EEs with $M = 32$, $K_d = K_u  = 2$, $N_t = N_r = 2$ and $S_{ok} = S_{ol} = 0.1$ bits/s/Hz: (a) $w_1 = w_2 = w_3 = w_4 = 0.25$, (b) $w_1 = 0.08, w_2 = 0.02, w_3 = 0.5, w_4 = 0.4$; c) Convergence of decentralized algorithm.} 
	\label{fig:4}
\end{figure}

{\textbf{Sum SE - impact of channel estimation error:} We know that the channel estimation error is a function of pilot transmit power $p_t$. We now vary $p_t$ and evaluate its impact on the sum SE for a full-duplex cell-free massive MIMO system in Fig.~\ref{fig:pilot}. For this study, we considered $M = 32$ APs, $K_d = 12$ downlink UEs, $K_u = 8$ uplink UEs and transmit power $p_d = 2p_u = 30$ dBm. We see that the sum SE increases for $p_t \le -10$ dB but saturates beyond that. This is because the channel estimation error reduces with increase in pilot power till $p_t = -10$ dB.  Any further increase in $p_t$, only marginally reduces the channel estimation error, which does not affect the sum SE. Our choice of $p_t = 0.2$ W in the numerical studies is, therefore, practical.}

\textbf{WSEE metric - influence of weights:} We now demonstrate that the WSEE metric can accommodate the heterogeneous EE requirements of both uplink and downlink UEs. For this study, we consider a particular realization of a FD CF mMIMO system with a transmit power {$p_d = 2p_u = 30$ dBm}, $M = 32$ APs, $K_d = K_u = K/2 = 2$ uplink and downlink UEs and $N_t = N_r = N = 2$ transmit and receive antennas on each AP, with QoS constraints $S_{ok} = S_{ol} = 0.1$ bits/s/Hz. We plot 
the individual EEs of the uplink (UL) and downlink (DL) UEs versus the SCA iteration index for centralized WSEE maximization, using Algorithm~\ref{algo1}, for two different combinations of UE weights. Weights $w_1$ and $w_2$ are associated with DL UE $1$ and DL UE $2$, while weights $w_3$ and $w_4$ are associated with UL UE $1$ and UL UE $2$, respectively.

We plot in Fig. \ref{fig:4a} and Fig. \ref{fig:4b} the individual EEs of UL and DL UEs, with: i) equal weights ($w_1 = w_2 = w_3 = w_4 = 0.25$), and ii) $w_1 = 0.08$, $w_2 = 0.02$, $w_3 = 0.5$, $w_4 = 0.4$, respectively. In Fig.~\ref{fig:4a}, with equal weights, UEs attain an EE depending on their relative channel conditions, which clearly indicates that in terms of channel conditions, DL UE $2$ $\gg$ DL UE $1$ $>$ UL UE $2$ $>$ UL UE $1$. In Fig.~\ref{fig:4b}, the weights are chosen in an order which is opposite to the channel conditions. The EEs of the UL UEs now dominate the EE of DL UE $1$, while reversing their relative order. The DL UE $2$, with excellent channel, still attains a high EE, although lower than in Fig.~\ref{fig:4a}.

\textbf{Convergence of decentralized ADMM algorithm:}
We plot in Fig.~\ref{fig:4c} the WSEE obtained using decentralized Algorithm~\ref{algo2} with SCA iteration index.  We consider $M = 10$ APs, $K_u = K_d = K/2 = 2$ uplink and downlink UEs and $N_t = N_r = \{1, 2\}$ transmit and receive antennas on each AP at transmit power {$p_d = 2p_u = p = 30$ dBm}. We assume the following:  i) penalty parameters $\rho_C = \rho_{\theta} = 0.1$; ii) penalty parameter update threshold factor $\mu = 10$; iii) ADMM convergence threshold $\epsilon_{\text{ADMM}} = 0.01$; and iv) SCA convergence threshold $\epsilon_{\text{SCA}} = 0.001$. We consider two values of the penalty update parameter: $\vartheta = \{1.2, 1.8\}$. We note that the algorithm in both cases  converges  marginally quicker with $\vartheta = 1.2$. A smaller penalty update parameter is therefore beneficial as then changes in the penalty parameters are not too abrupt, and a bad ADMM iteration which causes the primal and dual residues to diverge is, consequently, not overly responded to~\cite{ADMMVar}. We therefore fix $\vartheta = 1.2$ for the rest of the simulations. 

{\textbf{Comparison with existing schemes:} We now compare our proposed FD CF mMIMO WSEE optimization strategy with some existing approaches. In particular, we compare the
\begin{itemize}
	\item proposed fair AP selection algorithm, Algorithm~\ref{algo0}, with the optimal AP selection scheme proposed in~\cite{tvchen4}.
	\item maximum-ratio combining (MRC)/maximal ratio transmission (MRT) considered herein with  zero-forcing reception (ZFR)/ zero-forcing transmission (ZFT)~\cite{tvchen1}.
\end{itemize}}
\begin{figure}[ht]
	\centering
	\begin{subfigure}{.3\textwidth}
	    \centering
	    \includegraphics[height = \linewidth, width = \textwidth]{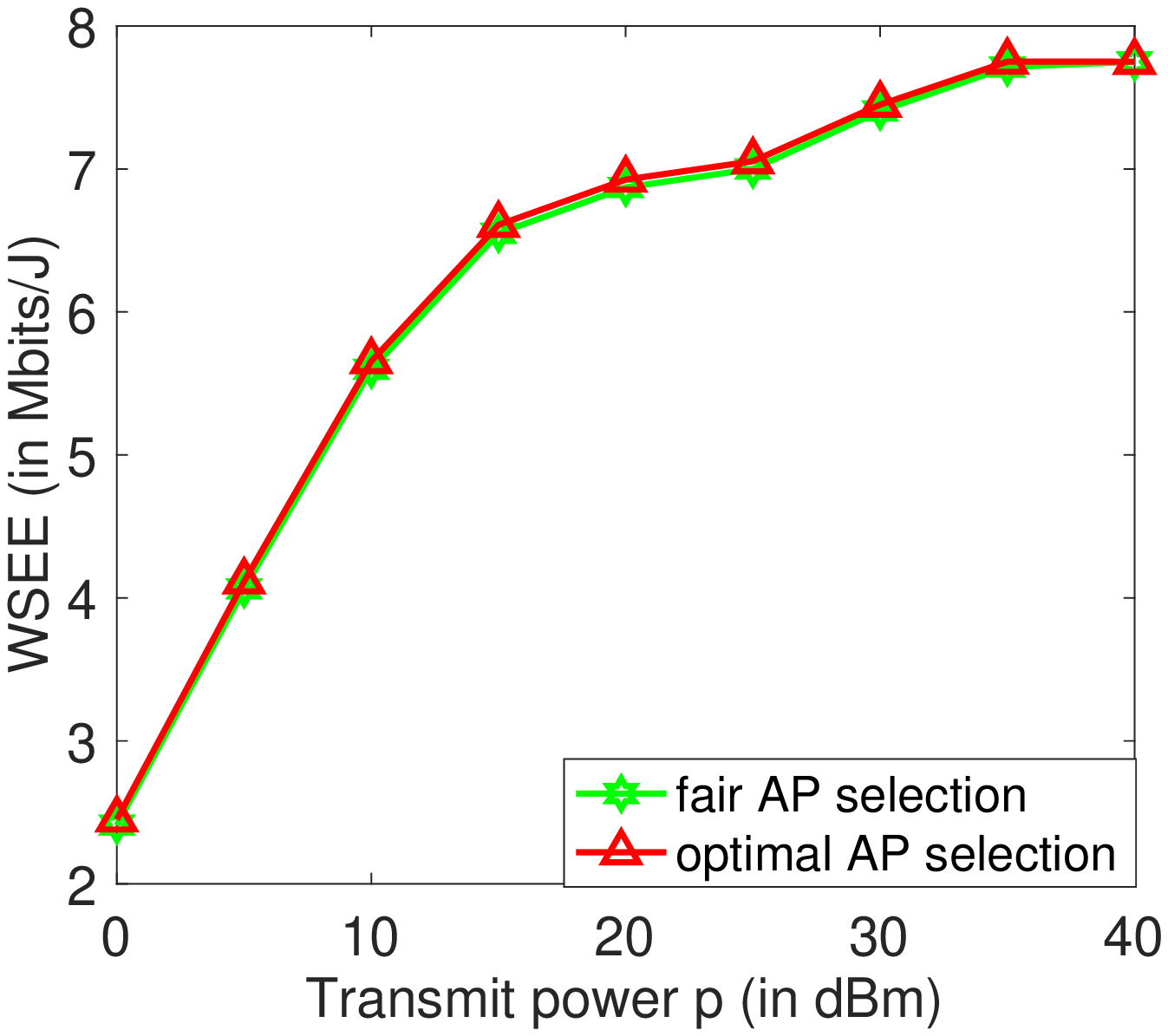}
	    \caption{} 
	    \label{fig:optimal_fair_AP} 
	\end{subfigure}
	\begin{subfigure}{0.3\textwidth}
	    \centering
	    \includegraphics[height = \linewidth, width = \textwidth]{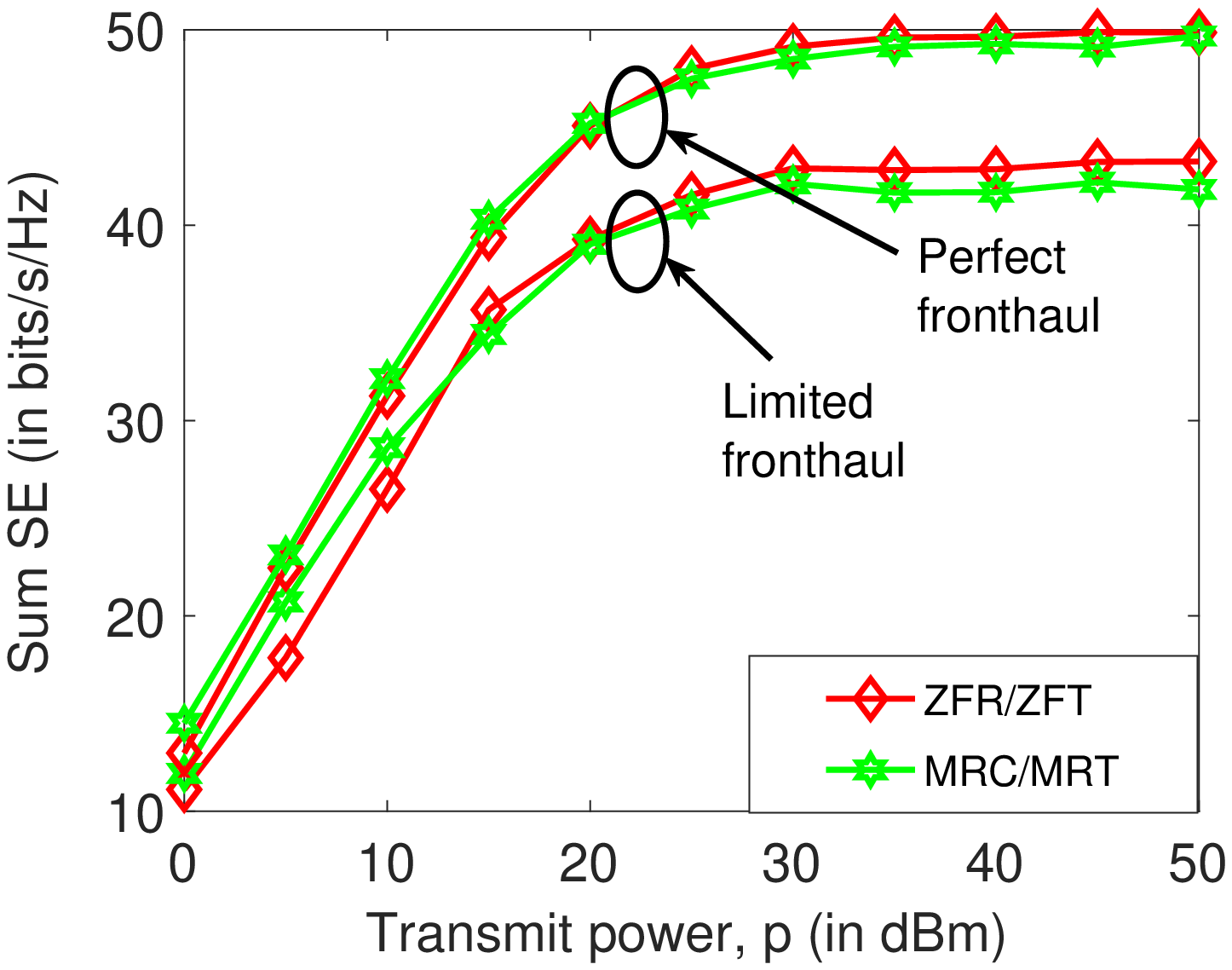}
	    \caption{}
        \label{fig:MR_ZF}
	\end{subfigure}
	\caption{{(a) WSEE comparison between fair and optimal AP selection algorithms, (b) SE comparison between MRC/MRT and ZFR/ZFT transceivers, vs maximum transmit power $p$ with $M = 32$ APs, $N_t = N_r = 8$ transmit and receive antennas, $K_d = 12$ downlink users and $K_u = 8$ uplink users.}} 
	\label{fig:perf-comp}
\end{figure} 

{We observe from Fig.~\ref{fig:optimal_fair_AP} that the proposed fair AP selection approach has almost as well as the optimal one in~\cite{tvchen4}. The proposed procedure efficiently eliminates the AP-UE links that do not have sufficient channel gain and thus contribute little to the system throughput while consuming a significant amount of power. Turning off APs according to the optimal AP selection procedure in~\cite{tvchen4}, thus only provides marginally better WSEE.}

{\textit{MRC/MRT and ZFR/ZFT comparison:} For this study, we considered a FD CF mMIMO system with $M = 32$ multi-antenna APs having $N_t = N_r = 8$ transmit and receive antennas antennas each, $K_d = 12$ downlink UEs and $K_u = 8$ uplink UEs. We consider two fronthaul cases: i) perfect high-capacity  with  $\Tilde{a} = \Tilde{b} = 1$, and ii) limited $C_{\text{fh}} = 10$ Mbps capacity with $\nu = 2$ quantization bits. We see {from Fig.~\ref{fig:MR_ZF}} that for both fronthaul capacities, the MRC/MRT transceiver for the scenario considered herein, although slightly inferior at high transmit power, performs reasonably well when compared with computationally-intensive ZFR/ZFT transceiver.}

\textbf{WSEE variation with parameters:} 
We now vary WSEE with important system  parameters and obtain crucial insights into energy-efficient FD CF mMIMO system designing. We consider $M = 32$ APs, {$N_t = N_r = N = 8$} AP transmit and receive antennas, {$K_d = 12$ downlink UEs, $K_u = 8$ uplink UEs} and QoS constraints $S_{ok} = S_{ol} = 0.1$ bits/s/Hz, unless mentioned otherwise.

\begin{figure}[htbp]
	\centering
	\begin{subfigure}{.3\textwidth}
		\centering
		\includegraphics[height = \linewidth, width=\linewidth]{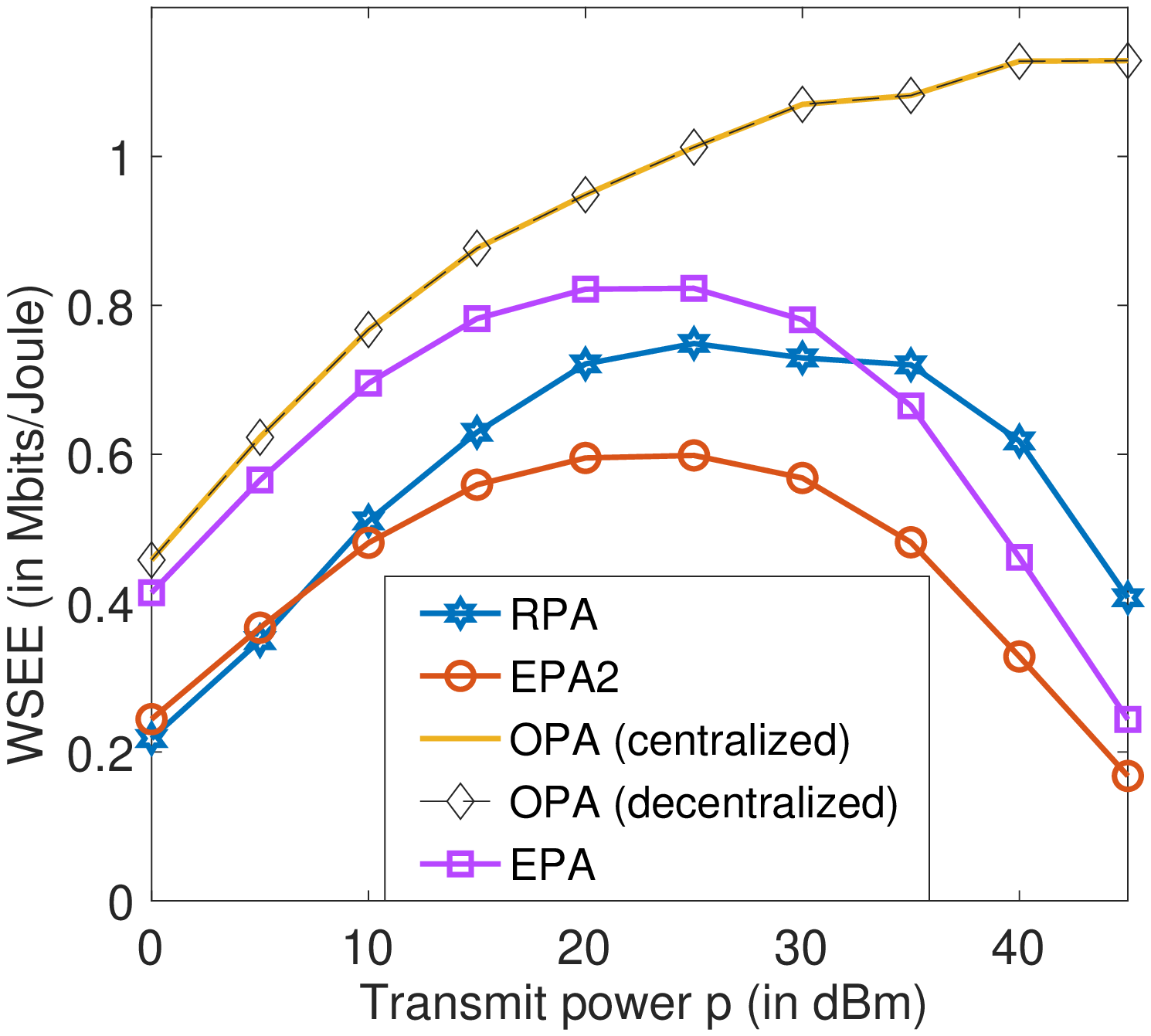}
		\caption{} 
		\label{fig:6a}
	\end{subfigure}
	\begin{subfigure}{.3\textwidth}
		\centering
		\includegraphics[height = \linewidth, width=\linewidth]{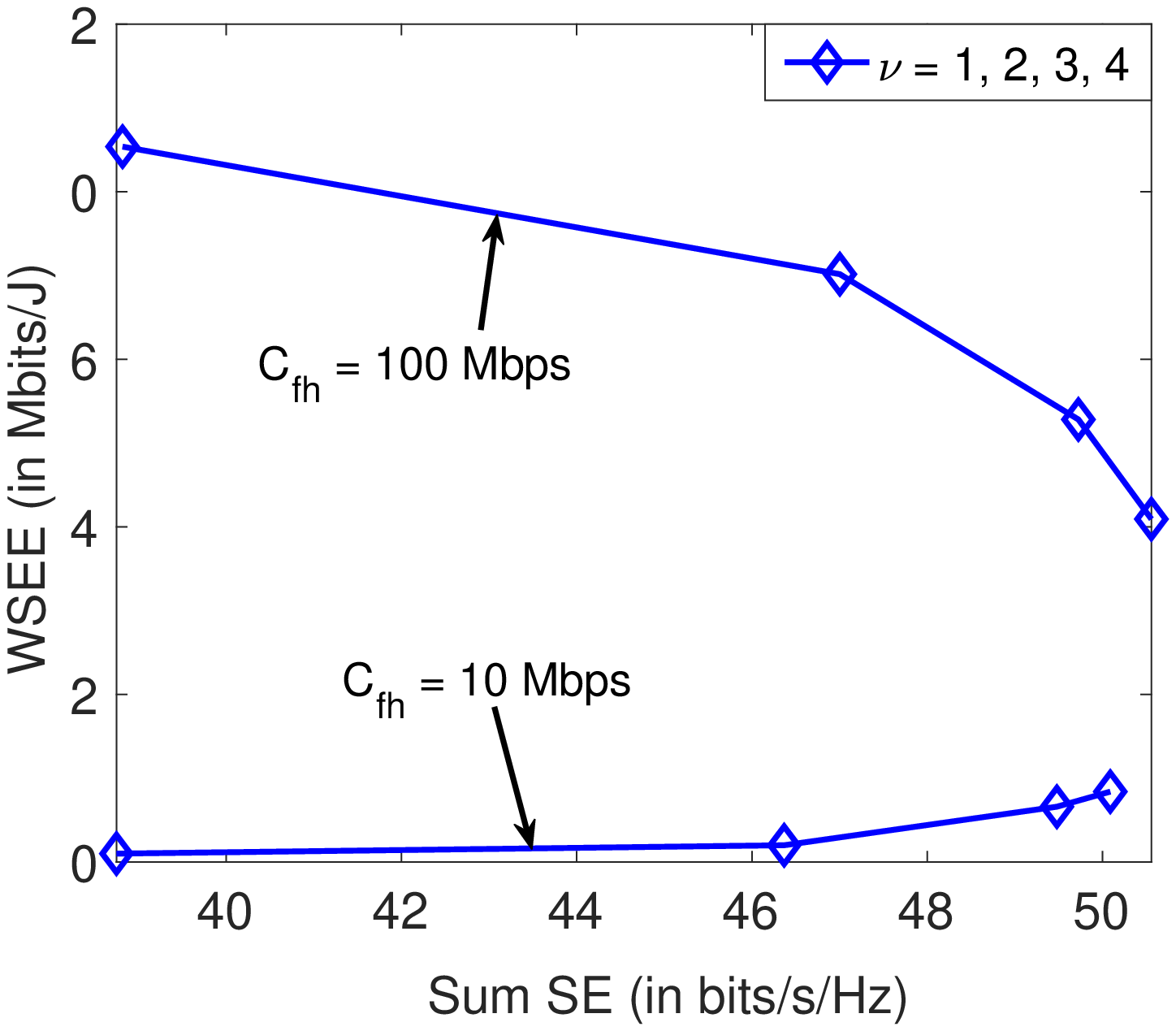}
		\caption{}
		\label{fig:7b}
	\end{subfigure}
	\begin{subfigure}{0.3\textwidth}
		\centering
		\includegraphics[height = \linewidth, width=\linewidth]{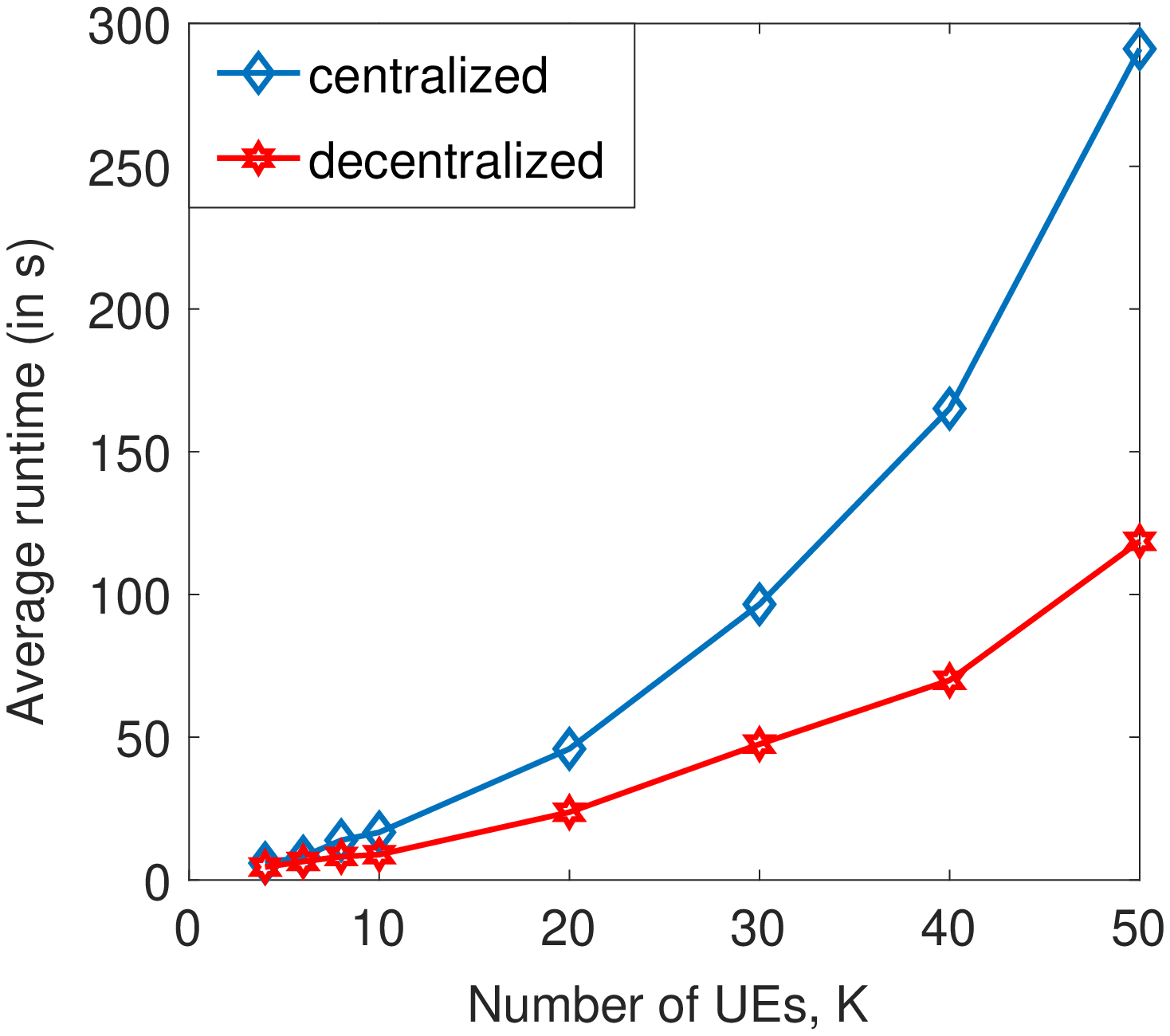}
		\caption{}
		\label{fig:runtime}
	\end{subfigure}
	\caption{WSEE vs (a) Maximum transmit power and (b) Sum SE by varying $\nu = 1 \text{ to } 4$, with $M = 32, K_d = K_u = 10, N_t = N_r = 2$ and $S_{ok} = S_{ol} = 0.1$ bits/s/Hz; c) Comparison of per-iteration runtime for decentralized and centralized algorithms.}
	\label{fig:5}
\end{figure}
We plot in Fig. \ref{fig:6a} the WSEE by simultaneously varying downlink and uplink transmit power as {$p_d = 2p_u = p$}. We consider centralized and decentralized optimal power allocation (OPA) approaches from Algorithm~\ref{algo1} and Algorithm~\ref{algo2}, respectively. We compare them with three sub-optimal power allocation schemes: i) equal power allocation of type 1, labeled as ``EPA 1", where $\eta_{mk} = \left({b}N_t\left(\sum_{k \in \kappa_{dm}} \gamma^{d}_{mk}\right)\right)^{-1}, \forall k \in \kappa_{dm}$ and $\theta_l = 1$~\cite{CellFreeEE, CellFreeEEUQ}, ii) equal power allocation of type 2, labeled as ``EPA 2", where $\eta_{mk} = \left({b} N_t K_{dm} \gamma^{d}_{mk}\right)^{-1}, \forall k \in \kappa_{dm}$ and $\theta_l = 1$ \cite{CellFreeEE}, and iii) random power allocation, labeled as ``RPA", where power control coefficients are chosen randomly from a uniform distribution between $0$ and the ``EPA 1" value. We note that the existing literature has not yet optimized the WSEE metric for CF mMIMO systems, and hence we can only compare with above sub-optimal schemes. Further, the decentralized ADMM approach, with lower computational complexity, has the same WSEE as that of the centralized one. Also, both decentralized and centralized approaches far outperform the baseline schemes.  

We next characterize in Fig. \ref{fig:7b} the joint variation of WSEE and sum SE with the number of quantization bits $\nu$ in the fronthaul links. The WSEE is obtained using decentralized Algorithm~\ref{algo2}. We consider transmit power $p_d = 2p_u = p = 30$ dBm 
and take two different cases: i) high fronthaul capacity, $C_{\text{fh}} = 100$ Mbps, which is sufficiently high to support all the UEs, and ii) limited fronthaul capacity, $C_{\text{fh}} = 10$ Mbps, which limits the number of UEs a single AP can serve. We observe that for $C_{fh} = 100$ Mbps, the WSEE falls with increase in $\nu$, even though the corresponding sum SE increases. For $C_{fh} = 10$ Mbps, both sum SE and WSEE simultaneously increase with increase in $\nu$. To explain this behavior, we note from Fig.~\ref{fig:2c} that increasing $\nu$ improves the sum SE for  $C_{\text{fh}} = 100$ Mbps and $C_{\text{fh}} = 10$ Mbps. For $C_{fh} = 100$ Mbps, the APs serve all the UEs, i.e., $K_{dm} = K_d$ and $K_{um} = K_u$, so increasing $\nu$ linearly increases the fronthaul data rate, $R_{fh}$ (see~\eqref{fhrate}). This, as seen from~\eqref{fixpower}, increases the traffic-dependent fronthaul power consumption. Using lower number (1-2) of quantization bits is therefore more energy-efficient, as it provides sufficiently good SE with a low energy consumption. However, for $C_{fh} = 10$ Mbps, $K_{um}$ and $K_{dm}$ have an upper limit, given by~\eqref{maxUEs}, which is inversely related to $\nu$. The product, $\nu (K_{um} + K_{dm})$, remains nearly constant for all values of $\nu$. Thus, $R_{fh}$ (see~\eqref{fhrate}) doesn't increase with increase in $\nu$ and remains close to the capacity, $C_{fh}$. The traffic-dependent fronthaul power consumption, given in~\eqref{fixpower}, hence, remains close to $P_{\text{ft}}$. A higher number  of quantization bits ($3-4$) therefore provides a higher sum SE and hence, also maximizes the WSEE.

\textbf{Latency:} The \textit{per-iteration} complexity of the \textit{decentralized} Algorithm~\ref{algo2}, as observed earlier in Section~\ref{complexitycompare},  is lower than the \textit{centralized} Algorithm~\ref{algo1}. We now demonstrate the same by comparing their \textit{per-iteration} runtime. For this simulation, as shown in Fig.~\ref{fig:runtime}, we consider an FD CF mMIMO system with $M = 32$ APs, each having $N_t = N_r = 8$ transmit and receive antennas, and plot the average runtime of each iteration by varying the total number of UEs, $K$, with $K_d = K_u = K/2$. We note that the decentralized algorithm has significantly lower  per-iteration runtime, particularly for large {$K$.} Both these algorithms require only large-scale channel coefficients and hence need to be executed only once in \textit{hundreds} of coherence intervals. 
\section{Conclusion}
We derived a SE lower bound for a FD CF mMIMO wireless system with optimal uniform  fronthaul quantization. Using a \textit{two-layered} approach, we optimized WSEE using SCA framework which in each iteration solves a GCP either centrally or decentrally using ADMM. We showed how WSEE incorporates EE requirements of different UEs. We analytically and numerically demonstrated the convergence of  decentralized algorithm. We showed that it achieves the same WSEE as the centralized approach with a much reduced computational complexity. 
\appendices 
\section{}\label{UQModel}
We use the optimal uniform quantization model from~\cite{CellFreeMaxMinUQ,CellFreeEEUQ}. Using Bussgang decomposition~\cite{Bussgang}, the quantization function $\mathcal{Q}(x) \! = \! \Tilde{a}x + \sqrt{p_x}\Tilde{\varsigma}_{d}$, where $p_x = \mathbb{E}\{|x|^{2}\}$ is the power of the unquantized signal $x$, $\Tilde{a} \! = \! \frac{1}{p_x}\int_{\mathcal{X}} \! xh(x)f_{X}(x) dx$, $\Tilde{b} = \frac{1}{p_x}\int_{\mathcal{X}}h^{2}(x)f_{X}(x) dx$ and $\Tilde{\varsigma}_{d}$ is the normalized distortion whose power is given as $\mathbb{E}\{\Tilde{\varsigma}^{2}_{d}\}=\Tilde{b}-\Tilde{a}^2$. Here $h(x)$ is the mid-rise uniform quantizer with  $L=2^{\nu}$ quantization levels rising in steps of size $\Tilde{\Delta}$, and $\nu$ being the number of quantization bits. The signal-to-distortion ratio $\text{SDR} = \frac{\mathbb{E}\{(\Tilde{a}x)^{2}\}}{p_x\mathbb{E}\{\Tilde{\varsigma}^2_{d})\}} = \frac{\Tilde{a}^2}{\Tilde{b}-\Tilde{a}^2}.$ The optimal step-size $\Tilde{\Delta}_{\text{opt}}$ maximizes the SDR {for a given $\nu$}. The optimal $\Tilde{a}$ and $\Tilde{b}$ values are calculated using {the optimal $\Tilde{\Delta}_{\text{opt}}$ for each value of $\nu$}, and are given in Table~\ref{table:A}~\cite{CellFreeMaxMinUQ}.
\begin{table}[h!]
\centering
 \begin{tabular}{|| c  c  c  c ||} 
 \hline
 $\nu$ & $\Tilde{\Delta}_{\text{opt}}$ & $\mathbb{E}\{\Tilde{\varsigma}^{2}_{d}\} =  \Tilde{b}-\Tilde{a}^2$ & $\Tilde{a}$ \\ [0.5ex] 
 \hline\hline
 1 & 1.596 & 0.2313 & 0.6366 \\  
 \hline
 2 & 0.9957 & 0.10472 & 0.88115 \\
 \hline
 3 & 0.586 & 0.036037 & 0.96256 \\
 \hline   
 4 & 0.3352 & 0.011409 & 0.98845 \\
 \hline 
 5 & 0.1881 & 0.003482 & 0.996505 \\ 
 \hline
 6 & 0.1041 & 0.0010389 & 0.99896 \\
 \hline
\end{tabular}
\caption{Optimal Uniform Quantization Parameters}
\label{table:A}
\end{table}
\section{} \label{SINRterms}
We now derive the achievable SE expression for the $k$th downlink UE in \eqref{dlrate}. From Section \ref{ul_ch_est}, we know that $\bm{g}^{d}_{mk} = \hat{\bm{g}}^{d}_{mk} + \bm{e}^{d}_{mk}$, where $\hat{\bm{g}}^{d}_{mk}$ and $\bm{e}^{d}_{mk}$ are independent and $\mathbb{E}\{\|\hat{\bm{g}}^{d}_{mk}\|^{2}\} = N_t\gamma^{d}_{mk}$. We can express the desired signal for the $k$th downlink UE as 
\begin{align}
    &\mathbb{E}\{|\text{DS}^d_k|^{2}\} \notag\\
    &=  \Tilde{a}^{2} \rho_d \mathbb{E}\{|\sum\nolimits_{m \in \mathcal{M}^{d}_{k}} \!\! \sqrt{\eta_{mk}} \mathbb{E}\{(\hat{\bm{g}}^{d}_{mk})^{T} (\hat{\bm{g}}^{d}_{mk})^{*}\} s^d_k|^{2}\}  \notag\\
    &= \Tilde{a}^{2} N^2_t \rho_d (\sum\nolimits_{m \in \mathcal{M}^{d}_{k}} \!\! \sqrt{\eta_{mk}} \gamma^d_{mk})^{2}. \label{dl_DS} 
\end{align}
We now calculate the beamforming uncertainty for the $k$th downlink UE as follows
\begin{align}
	&\notag\!\mathbb{E}\{|\text{BU}^d_k|^{2}\} \\
		&= \Tilde{a}^{2} \rho_d \!\!\!\!\sum_{m \in \mathcal{M}^{d}_{k}} \!\!\!\!\! \eta_{mk} \mathbb{E}\{| (\bm{g}^{d}_{mk})^{T}\! (\hat{\bm{g}}^{d}_{mk})^{*} \!-\! \mathbb{E}\{(\bm{g}^{d}_{mk})^{T}\! (\hat{\bm{g}}^{d}_{mk})^{*})\}|^2\} \notag\\
	\notag &\!\stackrel{(a)}{=} \Tilde{a}^{2} \rho_d  \sum\nolimits_{m \in \mathcal{M}^{d}_{k}} \eta_{mk} (N_t(N_t +1) (\gamma^d_{mk})^2 \\
	 &\quad +  N_t \gamma^d_{mk}(\beta^d_{mk} - \gamma^d_{mk}) -  N^2_t (\gamma^d_{mk})^2) \notag\\ 
	&= \Tilde{a}^{2} N_t \rho_d \sum\nolimits_{m \in \mathcal{M}^{d}_{k}} \eta_{mk} \beta^{d}_{mk} \gamma^{d}_{mk}. \label{dl_BU} 
\end{align}

Equality $(a)$ is because i) $\hat{\bm{g}}^{d}_{mk}$ are zero-mean and uncorrelated; and ii) $\mathbb{E}\{\|\hat{\bm{g}}^{d}_{mk}\|^{4}\} = N_t(N_t+1) (\gamma^{d}_{mk})^{2}$~\cite{FDCellFree} and $\mathbb{E}\{\|\bm{e}^{d}_{mk}\|^{2}\} = (\beta^d_{mk} - \gamma^d_{mk})$. 

We now simplify MUI for the $k$th downlink UE:  
\begin{align}
	&\mathbb{E}\{|\text{MUI}^d_k|^{2}\} \notag \\
	&= \! \Tilde{a}^2 \rho_d \!\! \sum\nolimits_{m=1}^{M} \sum\nolimits_{q \in \kappa_{dm} \setminus k} \!\! \eta_{mq} \mathbb{E}\{|(\bm{g}^{d}_{mk})^{T} (\hat{\bm{g}}^{d}_{mq})^{*}|^2\} \!\nonumber \\
	&\stackrel{(a)}{=} \Tilde{a}^{2} N_t \rho_d  \sum\nolimits_{m=1}^{M} \sum\nolimits_{q \in \kappa_{dm} \setminus k} \beta^{d}_{mk} \eta_{mq}  \gamma^{d}_{mq}. \label{dl_MUI}  
\end{align}
Equality (a) is because: i) $\hat{\bm{g}}^{d}_{mq}$ and $\bm{g}^{d}_{mk}$ are mutually independent; and $\text{ii) }\mathbb{E}\{|(\bm{g}^{d}_{mk})^{T} (\hat{\bm{g}}^{d}_{mq})^{*}|^2\}\!\! =\!\! \mathbb{E}\{ (\hat{\bm{g}}^{d}_{mq})^{T} \mathbb{E}\{(\bm{g}^{d}_{mk})^{*} (\bm{g}^{d}_{mk})^{T}\} (\hat{\bm{g}}^{d}_{mq})^{*}\} = N_t  \beta^d_{mk} \gamma^d_{mq}.$ 

We next calculate UDI for the $k$th downlink UE:
\begin{align}
	\mathbb{E}\{|\text{UDI}^d_k|^{2}\} = \rho_u  \sum_{l=1}^{K_u} \mathbb{E}\{|h_{kl}|^2\} \theta_l 
	= \rho_u  \sum_{l=1}^{K_u} \Tilde{\beta}_{kl} \theta_l. \label{dl_UDI} 
\end{align}

We express the total quantization distortion {(TQD)} for the $k$th downlink UE as follows
\begin{align}
	\mathbb{E}\{|\text{TQD}^d_k|^2\}\!&\approx \! \rho_d \!\! \sum\nolimits_{m=1}^{M} \sum\nolimits_{q \in \kappa_{dm}} \mathbb{E}\{|(\bm{g}^{d}_{mk})^{T} (\hat{\bm{g}}^{d}_{mq})^{*} \varsigma^{d}_{mq}|^2\} \! \nonumber \\
	&\stackrel{(a)}{=} \! (\Tilde{b} \! - \! \Tilde{a}^{2}) N_t \rho_d \!\! \sum\nolimits_{m=1}^{M} \sum\nolimits_{q \in \kappa_{dm}} \!\!\!\!\!\!\!\!  \beta^{d}_{mk} \eta_{mq} \gamma^{d}_{mq}.\!\!\!\! \label{dl_TQD} 
\end{align}
Equality $(a)$ is because: i) $\mathbb{E}\{|\varsigma^d_{mk}|^2\} = (\Tilde{b} - \Tilde{a}^2) \eta_{mk}$; ii) distortion $\varsigma^{d}_{mq}$ is independent of channels $\bm{g}^{d}_{mk}$ and $\hat{\bm{g}}^{d}_{mq}$; and iii) $\mathbb{E}\{|(\bm{g}^{d}_{mk})^{T} (\hat{\bm{g}}^{d}_{mq})^{*} \varsigma^{d}_{mq}|^2\} = (\Tilde{b} - \Tilde{a}^{2}) \eta_{mq} \beta^d_{mk} \mathbb{E}\{(\hat{\bm{g}}^{d}_{mq})^{T}(\hat{\bm{g}}^{d}_{mq})^{*}\} = (\Tilde{b} - \Tilde{a}^{2}) N_t \beta^d_{mk} \eta_{mq} \gamma^d_{mq}$.  The result in~\eqref{dlrate} follows from the expression for the achievable SE lower bound 
\begin{align*}
	S^d_k = \tau_f \log_2 \left(1 + \frac{\mathbb{E}\{|\text{DS}^d_k|^2\}}{\left\{\splitfrac{\mathbb{E}\{|\text{BU}^d_k|^2\} + \mathbb{E}\{|\text{MUI}^d_k|^2+ \mathbb{E}\{|\text{UDI}^d_k|^2\} \}}{+  \mathbb{E}\{|\text{TQD}^d_k|^2\} + \mathbb{E}\{|w^d_k|^2\}}\right\}}\right).
\end{align*}

We now derive the achievable SE expression for the $l$th uplink UE in \eqref{ulrate}. We know from Section \ref{ul_ch_est} that $\bm{g}^{u}_{ml} = \hat{\bm{g}}^{u}_{ml} + \bm{e}^{u}_{ml}$, where $\hat{\bm{g}}^{u}_{ml}$ and $\bm{e}^{u}_{ml}$ are independent and $\mathbb{E}\{\|\hat{\bm{g}}^{u}_{ml}\|^{2}\} = N_r\gamma^{u}_{ml}$. We can express the desired signal for the $l$th uplink UE as given next
\begin{align}
	\mathbb{E}\{|\text{DS}^{u}_{l}|^{2}\} &= \mathbb{E}\{|\Tilde{a} \!\! \sum_{m \in \mathcal{M}^{u}_{l}}\!\!\!\! \sqrt{\rho_u} \mathbb{E}\{\sqrt{\theta_l}(\hat{\bm{g}}^{u}_{ml})^{H} (\hat{\bm{g}}^{u}_{ml} + \bm{e}^{u}_{ml}) s^{u}_{l}\}|^{2}\} \nonumber \\
	&=  \Tilde{a}^{2} N^2_r \rho_u \theta_l (\sum\nolimits_{m \in \mathcal{M}^{u}_{l}} \gamma^{u}_{ml})^{2}. \label{ul_DS}
\end{align}

The beamforming uncertainty for the $l$th uplink UE is 
\begin{align}
	&\mathbb{E}\{|\text{BU}^u_l|^{2}\} \notag\\
	&= \notag \Tilde{a}^{2} \rho_u \theta_l \sum_{m \in \mathcal{M}^{u}_{l}} \mathbb{E}\{\| ((\hat{\bm{g}}^{u}_{ml})^{H} \bm{g}^{u}_{ml} - \mathbb{E}\{(\hat{\bm{g}}^{u}_{ml})^{H} \bm{g}^{u}_{ml} )\}\|^{2}\} \\
	&\stackrel{(a)}{=} \Tilde{a}^{2}\rho_u \theta_l \!\sum_{m \in \mathcal{M}^{u}_{l}}\!\! (\mathbb{E}\{\|\hat{\bm{g}}^{u}_{ml}\|^{4}\} + \mathbb{E}\{| (\hat{\bm{g}}^{u}_{ml})^{H}\bm{e}^{u}_{ml}|^{2}\} - N^2_r(\gamma^{u}_{ml})^{2}\}) \notag\\
	&\stackrel{(b)}{=}  \Tilde{a}^{2}\rho_u N_r \theta_l \! \sum\nolimits_{m \in \mathcal{M}^{u}_{l}}\!\! \gamma^{u}_{ml}\beta^{u}_{ml}. \label{ul_BU}
	\end{align}
Equality $(a)$ is because:  i)  $\bm{e}^u_{ml}$ and $\hat{\bm{g}}^{u}_{ml}$ are zero-mean and uncorrelated; ii) $ \mathbb{E}\{ |\hat{\bm{g}}^{u}_{ml}|^{2}\} = N_r \gamma^{u}_{ml}$. Equality $(b)$ is because $\mathbb{E}\{\|\hat{\bm{g}}^{u}_{ml}\|^{4}\} = N_r(N_r+1) (\gamma^{u}_{ml})^{2}$~\cite{FDCellFree} and $\mathbb{E}\{\|\bm{e}^{u}_{ml}\|^{2}\} = (\beta^u_{ml} - \gamma^u_{ml})$. 

We simplify the MUI for the $l$th uplink UE as 
\begin{align}
	\mathbb{E}\{|\text{MUI}^u_l|^2\} \!&=\! \Tilde{a}^2 \rho_u \!\! \sum\nolimits_{m \in \mathcal{M}^u_l} \sum\nolimits_{q=1, q \neq l}^{K_u} \!\! \theta_q  \mathbb{E}\{|(\hat{\bm{g}}^{u}_{ml})^{H} \bm{g}^{u}_{mq}|^2\} \! \nonumber \\
	&\stackrel{(a)}{=} \! \Tilde{a}^{2}\rho_u N_r \!\! \sum\nolimits_{m \in \mathcal{M}^{u}_{l}} \sum\nolimits_{q=1, q \neq l}^{K_u} \!\! \gamma^{u}_{ml}\beta^{u}_{mq}\theta_q. \!\!\!\! \label{ul_MUI}
\end{align}
Equality $(a)$ is obtained by using these facts: i) $\hat{\bm{g}}^{u}_{ml}$, $\bm{g}^{u}_{mq}$ are mutually independent; and \text{ii)} 
\begin{align}
	 \mathbb{E}\{|(\hat{\bm{g}}^{u}_{ml})^{H} \bm{g}^{u}_{mq}|^{2} \} &=  \mathbb{E}\{(\bm{g}^{u}_{mq})^{H} \mathbb{E}\{(\hat{\bm{g}}^{u}_{ml})(\hat{\bm{g}}^{u}_{ml})^{H}\} \bm{g}^{u}_{mq}\}\!\! 
	 \nonumber \\ 
	 &= \gamma^{u}_{ml} \mathbb{E}\{||\bm{g}^{u}_{mq}||^{2}\}\!\! =\!\! N_r \gamma^{u}_{ml} \beta^{u}_{mq}. \label{MUI_term_ul}
\end{align}
We next obtain the noise power for the $l$th uplink UE as
\begin{align}
    &\mathbb{E}\{|\text{N}^u_l|^{2}\} \!=\! \Tilde{a}^{2}\!\!\!\! \sum_{m \in \mathcal{M}^{u}_{l}} \!\!\!\! \mathbb{E}\{|(\hat{\bm{g}}^{u}_{ml})^{H} \bm{w}^{u}_{m}|^{2}\} = \Tilde{a}^{2} N_r \!\!\!\! \sum_{m \in \mathcal{M}^{u}_{l}} \!\!\!\! \gamma^{u}_{ml}, \text{ where }\!\!\!\!\! \notag \\
    &\mathbb{E}\{|(\hat{\bm{g}}^{u}_{ml})^{H} \bm{w}^{u}_{m}|^{2}\} \! =\! \mathbb{E}\{(\bm{w}^{u}_{m})^{H} \mathbb{E}\{ \hat{\bm{g}}^{u}_{ml} (\hat{\bm{g}}^{u}_{ml})^{H}\} \bm{w}^{u}_{m}\} \! =\! N_r \gamma^{u}_{ml}.\!\!\!\! \label{N_term}
\end{align}
The undistorted MR-combined uplink signal at the $m$th AP is expressed as
\begin{align} 
    &(\hat{\bm{g}}^{u}_{ml})^{H}\bm{y}^{u}_{m} \notag\\
    &= \sum_{q=1}^{K_u} (\hat{\bm{g}}^{u}_{ml})^{H}\bm{g}^{u}_{mq} x^{u}_{q} + \sum_{i=1}^{M} (\hat{\bm{g}}^{u}_{ml})^{H} \bm{H}_{mi} \bm{x}^{d}_{i} 
    + (\hat{\bm{g}}^{u}_{ml})^{H} \bm{w}^{u}_{m} \notag\\
   &= \underbrace{\sqrt{\rho_u}  (\hat{\bm{g}}^{u}_{ml})^{H} \bm{g}^{u}_{ml} \sqrt{\theta_l} s^{u}_{l}}_{\text{message signal}} 
    + \underbrace{\sqrt{\rho_u} \sum\nolimits_{q=1, q \neq l}^{K_u} (\hat{\bm{g}}^{u}_{ml})^{H} \bm{g}^{u}_{mq} \sqrt{\theta_q} s^{u}_{q}}_{\text{multi-user interference, MUI}^u_l} \nonumber \\
\end{align} 
\begin{align}
    &\quad + \underbrace{\sqrt{\rho_d}  \sum_{i=1}^{M} \! \sum_{k \in \kappa_{di}} \!\!  (\hat{\bm{g}}^{u}_{ml})^{H} \bm{H}_{mi} (\hat{\bm{g}}^{d}_{ik})^{*} (\Tilde{a}\sqrt{\eta_{ik}}s^{d}_{k} + \varsigma^{d}_{ik})}_{\text{intra-/inter-AP residual interference, RI}^u_l}\notag\\ 
    &\quad +  \underbrace{(\hat{\bm{g}}^{u}_{ml})^{H} \bm{w}^{u}_{m}}_{\text{additive noise at APs, N}^u_l}\!\!\!\!\!\!. \nonumber 
\end{align}

We assume, similar to~\cite{CellFreeEEUQ},  that the quantization distortion is uncorrelated across the fronthaul links. The TQD power for the $l$th uplink UE is accordingly expressed as
\begin{align*}
\mathbb{E}\{|\text{TQD}^u_l|^2\}  &\approx  \sum_{m \in \mathcal{M}^u_l} \!\! \mathbb{E}\{|\zeta^u_{ml}|^2\}\\
	 &\approx  (\Tilde{b} \!-\! \Tilde{a}^{2}) \!\! \sum_{m \in \mathcal{M}^{u}_{l}} \! \mathbb{E}\{| (\hat{\bm{g}}^{u}_{ml})^{H}\bm{y}_m|^2\}.
\end{align*}
Using arguments similar to~\eqref{ul_DS}-\eqref{N_term}, the contributions of the message signal (DS + BU), MUI and noise (N) to the 
{TQD} for the $l$th uplink UE are
\begin{align*}
    &\mathbb{E}\{|\text{TQD}^u_l|^{2}\}_{\text{DS+BU}} \notag\\ &\approx (\Tilde{b} - \Tilde{a}^{2}) N_r \rho_u \theta_l (N_r \sum_{m \in \mathcal{M}^u_l}  (\gamma^{u}_{ml})^{2} + \sum_{m \in \mathcal{M}^u_l} \gamma^u_{ml} \beta^{u}_{ml}). \\
    &\mathbb{E}\{|\text{TQD}^u_l|^{2}\}_{\text{MUI}} \approx (\Tilde{b} - \Tilde{a}^{2}) N_r \rho_u \!\! \sum_{m \in \mathcal{M}^u_l} \! \sum_{q=1, q \neq l}^{K_u} \! \gamma^u_{ml} \beta^{u}_{mq} \theta_q, \\
    &\mathbb{E}\{|\text{TQD}^u_l|^{2}\}_{\text{N}} \approx (\Tilde{b} - \Tilde{a}^{2}) N_r \!\! \sum_{m \in \mathcal{M}^u_l} \!\! \gamma^u_{ml}. 
\end{align*}
To accurately model the RI with limited fronthaul capacity and compute the corresponding power, as well as its contribution to the quantization distortion, we propose a lemma. 
\begin{lemma}
	The intra-/inter-AP RI power and the RI contribution to the TQD power for the $l$th uplink UE in a FD CF mMIMO system with MRT/MRC transceiver are expressed as
	\begin{align}
	    &\mathbb{E}\{|\text{RI}^u_l|^{2}\} \notag \\
	   &= \Tilde{a}^2 \Tilde{b} N_r N_t \rho_d  \!\sum_{i=1}^{M} \!\sum_{k \in \kappa_{di}}  \gamma^{u}_{ml} \gamma^{d}_{ik} \beta_{\text{RI},mi} \gamma_{\text{RI}} \eta_{ik} N_r \gamma^{u}_{ml}.\!\!\!\! \label{ul_RI} \\
        &\mathbb{E}\{|\text{TQD}^u_l|^{2}\}_{\text{RI}}\notag \\
         &\!\approx\! (\Tilde{b} \!-\! \Tilde{a}^{2}) \Tilde{b} N_r N_t \rho_d \!\!\!\!\sum_{m \in \mathcal{M}^{u}_{l}} \! \sum_{i=1}^{M}\! \sum_{k \in \kappa_{di}} \!\!\!\! \gamma^{u}_{ml}  \beta_{\text{RI},mi} \gamma_{\text{RI}} \eta_{ik} \gamma^{d}_{ik}.
    \end{align}
\end{lemma}
\begin{proof} 
	We express the RI power of the undistorted, MR combined received signal for the $l$th uplink UE as 
   	\begin{align*}
   	    &\notag \mathbb{E}\{|\widetilde{\text{RI}}^u_l|^{2}\} \\
   	    &= \rho_d \sum_{i=1}^{M} \sum_{k \in \kappa_{di}} \mathbb{E}\{|(\hat{\bm{g}}^{u}_{ml})^{H} \bm{H}_{mi} (\hat{\bm{g}}^{d}_{ik})^{*} (\Tilde{a} \sqrt{\eta_{ik}} s^d_k + \zeta^d_{ik}) |^{2}\}  \\
	     \notag &\stackrel{(a)}{=}  \rho_d  \sum_{i=1}^{M} \sum_{k \in \kappa_{di}} \mathbb{E}\{|(\hat{\bm{g}}^{u}_{ml})^{H} \bm{H}_{mi} (\hat{\bm{g}}^{d}_{ik})^{*}|^{2} \Tilde{b} \eta_{ik}\} \\
	       &\stackrel{(b)}{=} \Tilde{b} N_r N_t \rho_d  \!\sum_{i=1}^{M} \!\sum_{k \in \kappa_{di}}  \gamma^{u}_{ml} \gamma^{d}_{ik} \beta_{\text{RI},mi} \gamma_{\text{RI}} \eta_{ik} N_r \gamma^{u}_{ml}.
    \end{align*} 

	Equality $(a)$ is because signal $\Tilde{a} \sqrt{\eta_{ik}} s^d_k$ and  quantization noise $\varsigma^{d}_{ik}$, are uncorrelated, and $\mathbb{E}\{|\varsigma^{d}_{ik}|^2\}\!\! =\!\! (\Tilde{b}\! -\! \Tilde{a}^2) \eta_{ik}$. Equality $(b)$ is because: i) $\hat{\bm{g}}^{u}_{ml}$, $\bm{H}_{mi}$ and $\hat{\bm{g}}^{d}_{mk}$ are mutually independent, 
	\begin{align}
    	\notag \text{ii) } &\mathbb{E}\{|(\hat{\bm{g}}^{u}_{ml})^{H} \bm{H}_{mi} (\hat{\bm{g}}^{d}_{ik})^{*}|^{2}\} \nonumber \\
    	&= \mathbb{E}\{(\hat{\bm{g}}^{d}_{ik})^{T} \mathbb{E}\{\bm{H}^{H}_{mi} \mathbb{E}\{ (\hat{\bm{g}}^{u}_{ml}) (\hat{\bm{g}}^{u}_{ml})^{H}\} \bm{H}_{mi}\} (\hat{\bm{g}}^{d}_{ik})^{*}\} \nonumber \\
    	&=  \gamma^{u}_{ml} \mathbb{E}\{(\hat{\bm{g}}^{d}_{ik})^{T} \mathbb{E}\{\bm{H}^{H}_{mi}  \bm{H}_{mi}\} (\hat{\bm{g}}^{d}_{ik})^{*}\}  \nonumber \\
    	&=\!\! N_r \gamma^{u}_{ml}\beta_{\text{RI},mi} \gamma_{\text{RI}} \mathbb{E}\!\{(\hat{\bm{g}}^{d}_{ik})^{T}(\hat{\bm{g}}^{d}_{ik})^{*}\}\!\! = \!\! N_r N_t \gamma^{u}_{ml} \gamma^d_k \beta_{\text{RI},mi} \gamma_{\text{RI}}.\!\!\!\! \label{RI_term}
	\end{align}
    We obtain the i) attenuated intra-/inter-AP RI power as $\mathbb{E}\{|\text{RI}^u_l|^{2}\} = \Tilde{a}^{2}\mathbb{E}\{|\widetilde{\text{RI}}^u_l|^{2}\}$; and intra-/inter-AP RI contribution to the TQD power as $\mathbb{E}\{|\text{TQD}^u_l|^{2}\}_{\text{RI}} \approx (\Tilde{b} - \Tilde{a}^{2}) \sum_{m \in \mathcal{M}^{u}_{l}} \mathbb{E}\{|\widetilde{\text{RI}}^u_l|^{2}\}.$ 
\end{proof}
The total quantization distortion for the $l$th uplink UE is given as $\mathbb{E}\{|\text{TQD}^u_l|^{2}\} \!=\! \mathbb{E}\{|\text{TQD}^u_l|^{2}\}_{\text{DS+BU}} + \mathbb{E}\{|\text{TQD}^u_l|^{2}\}_{\text{MUI}} + \mathbb{E}\{|\text{TQD}^u_l|^{2}\}_{\text{RI}} + \mathbb{E}\{|\text{TQD}^u_l|^{2}\}_{\text{N}}.$

The result in \eqref{ulrate} follows from the expression 
\begin{align*}
	S^u_l = \tau_f \log_2 \left(1 + \frac{\mathbb{E}\{|\text{DS}^u_l|^2\}}{\left\{\splitfrac{\mathbb{E}\{|\text{BU}^u_l|^2\} + \mathbb{E}\{|\text{MUI}^u_l|^2\} + \mathbb{E}\{|\text{RI}^u_l|^2\}}{+  \mathbb{E}\{|\text{TQD}^u_l|^2\} +  \mathbb{E}\{|\text{N}^u_l|^2\}}\right\}}\right).
\end{align*}
\bibliographystyle{IEEEtran}
\bibliography{IEEEabrv,paper_template_ref}
\begin{IEEEbiography}
[{\includegraphics[width=1in,height=1.25in,clip]{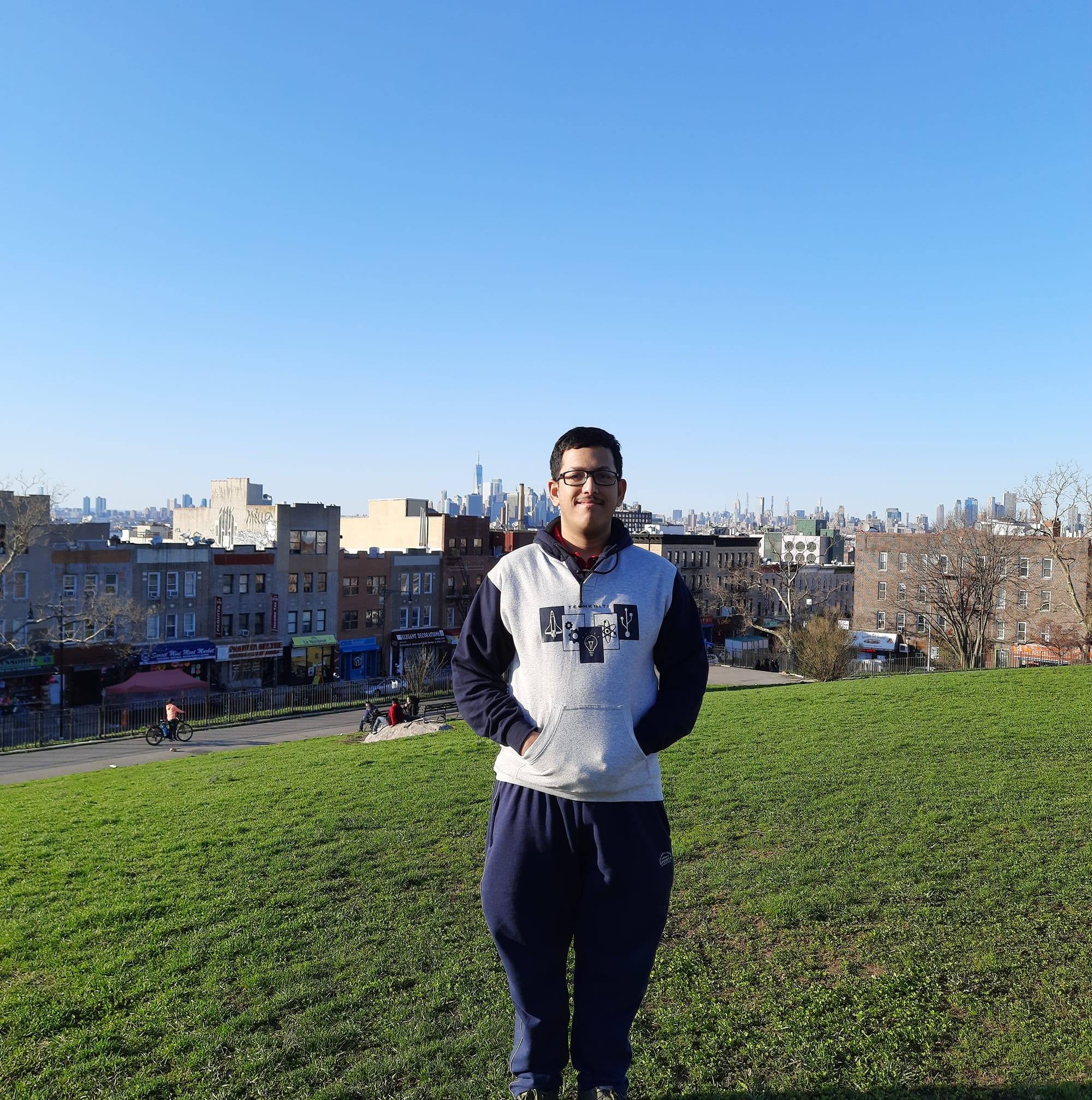}}]
{Soumyadeep Datta} received his B.Tech.-M.Tech. Dual Degree in Electrical Engineering from the Indian Institute of Technology Kanpur, India, in the year 2020. Since September 2020, he is pursuing a Dual Ph.D. degree with the Department of Electrical Engineering, Indian Institute of Technology Kanpur, India, and the Department of Electrical and Computer Engineering, NYU Tandon School of Engineering, USA. His research interests include beyond-5G wireless systems, wireless networks and cross-layer optimization.  
\end{IEEEbiography}
\begin{IEEEbiography}
[{\includegraphics[width=1in,height=1.25in,clip]{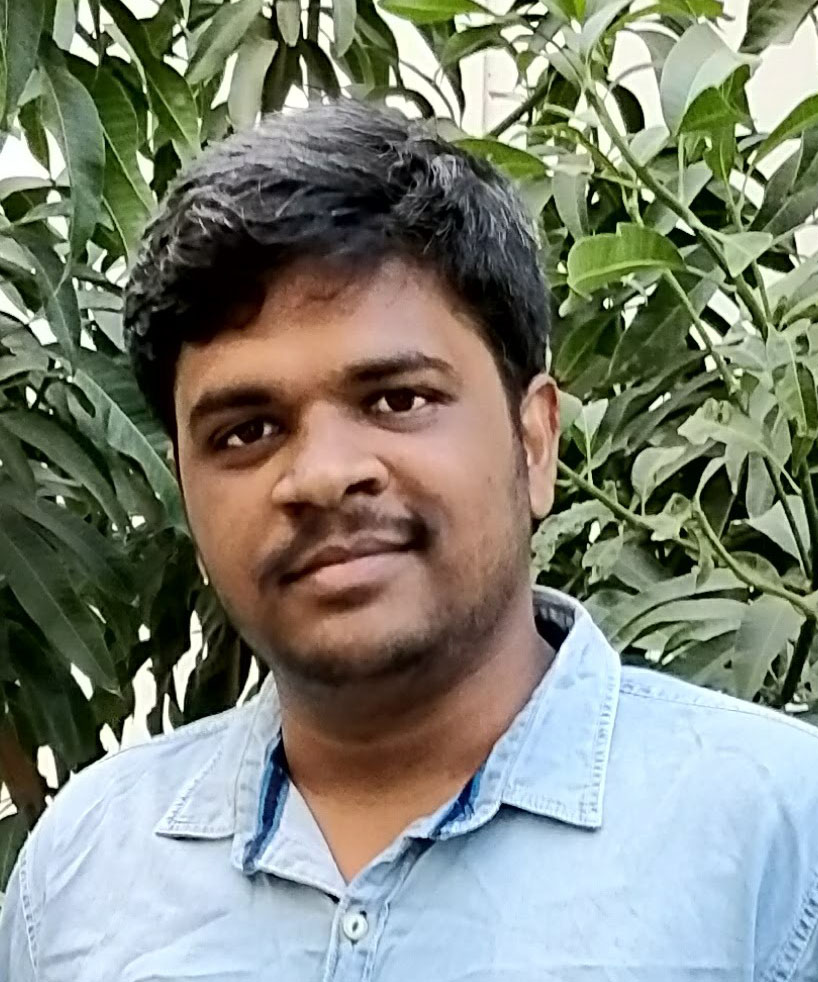}}]
{Dheeraj Naidu Amudala}
received his B.Tech degree in Electronics and Communication Engineering from JNTUACEA, Ananthapuramu,
India, in the year 2016. Since July 2016, he is working towards his M.Tech and Ph.D. degrees in the Department of Electrical Engineering from Indian Institute of Technology, Kanpur, India.

His research interests include massive MIMO, full-duplex, wireless relaying systems and optimization theory. 
\end{IEEEbiography}
\begin{IEEEbiography}
[{\includegraphics[width=1in,height=1.25in,clip]{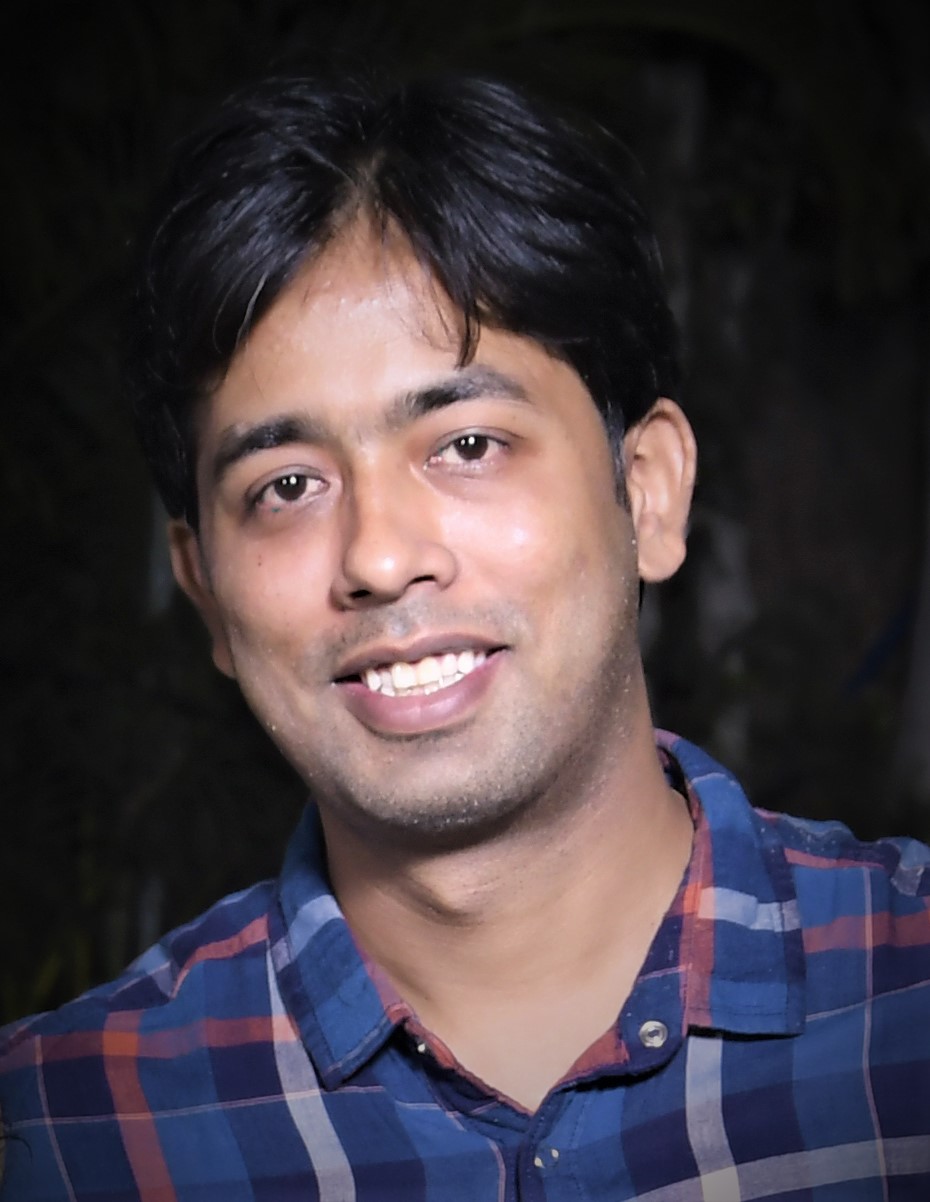}}]
{Ekant Sharma} received the M.Tech. and Ph.D. degrees in electrical engineering from the Signal Processing, Communication and Networks Group, Department of Electrical Engineering, Indian Institute of Technology Kanpur, India, in May 2011 and May 2020, respectively. From 2011 to 2012, he was with the IBM-India Software Lab and worked as an Associate Software Engineer. From August 2019 to January 2021, he worked at 5G Testbed Lab, Indian Institute of Technology Kanpur where he designed base station hardware and software algorithms for 5G NR. He is currently working as an Assistant Professor at Indian Institute of Technology Roorkee. His Ph.D. thesis received outstanding thesis award and also it was chosen for category: SPCOM Best Doctoral Dissertation—Honourable Mention at IEEE SPCOM conference. His research interests are within the areas of wireless communications systems, with special focus on practical massive MIMO, full-duplex, relays, energy efficiency and optimization.
\end{IEEEbiography}
\begin{IEEEbiography}[{\includegraphics[width=1.1in,height=1.1in,clip]{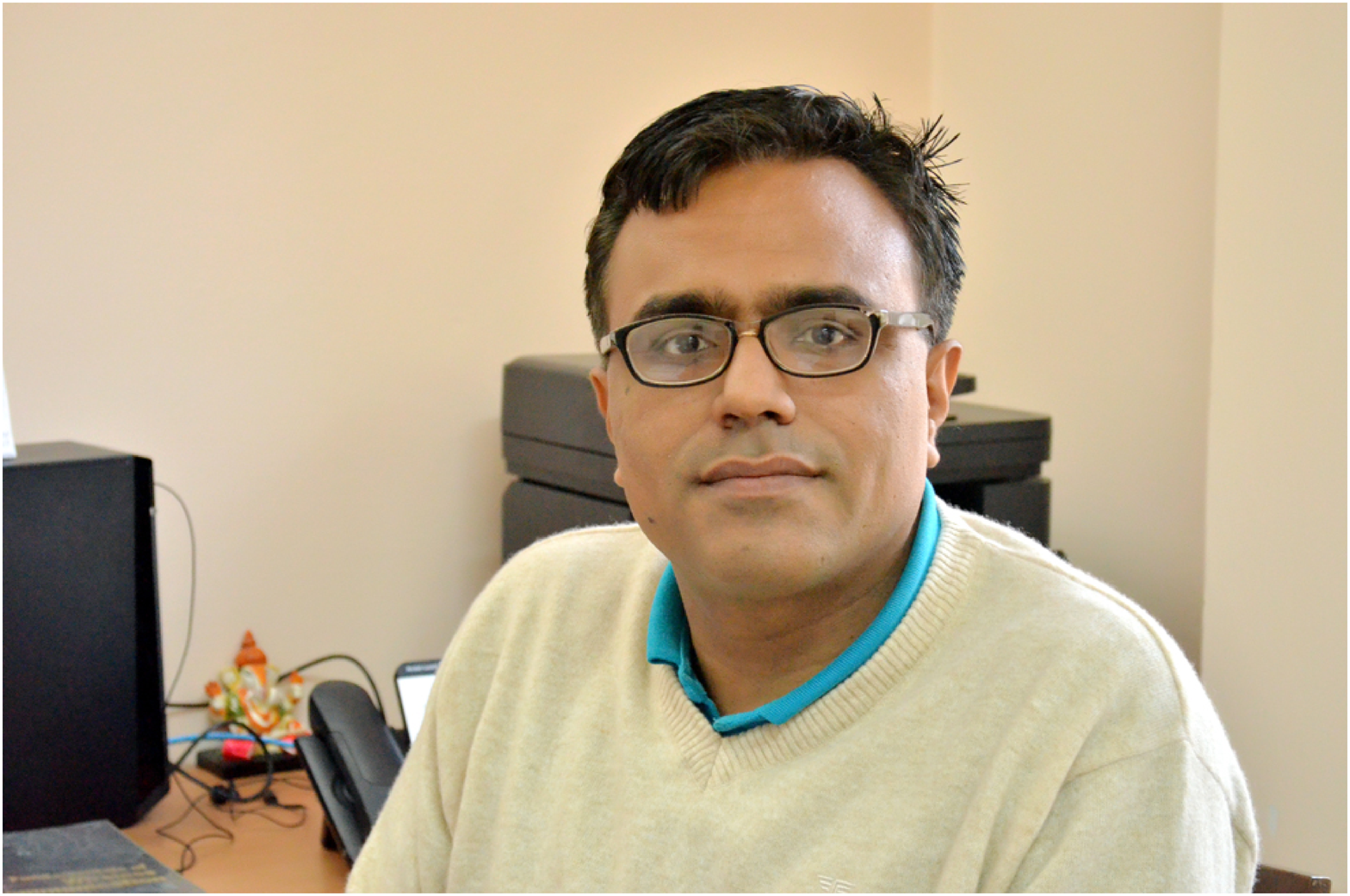}}]
{Rohit Budhiraja}
received the M.S. degree in electrical engineering and the Ph.D. degree from IIT Madras in 2004 and 2015, respectively. From 2004 to 2011, he worked for two start-ups where he designed both hardware and software algorithms, from scratch, for physical layer processing of WiMAX- and LTE-based cellular systems. He is currently an Assistant Professor with IIT Kanpur, where he is also leading an effort to design a 5G research testbed. His current research interests include design of energy-efficient transceiver algorithms for 5G massive MIMO and full-duplex systems, robust precoder design for wireless relaying, machine learning methods for channel estimation in mm-wave systems, and spatial modulation system design. His paper was shortlisted as one of the finalists for the Best Student Paper Awards at the IEEE International Conference on Signal Processing and Communications, Bangalore, India, in 2014. He also received IIT Madras Research Award for the quality and quantity of research work done in the Ph.D., Early Career Research Award, and Teaching Excellence Certificate at IIT Kanpur.
\end{IEEEbiography}
\begin{IEEEbiography}[{\includegraphics[width=1in,height=1.25in,clip,keepaspectratio]{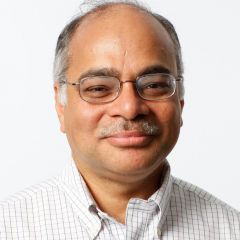}}]{Shivendra S. Panwar}
(S’82–M’85–SM’00–F’11) received the Ph.D. degree in electrical and computer engineering from the University of Massachusetts, Amherst, MA, USA, in 1986. He is currently a Professor with the Electrical and Computer Engineering Department, NYU Tandon School of Engineering. He is also the Director of the New York State Center for Advanced Technology in Telecommunications (CATT), the co-founder of the New York City Media Lab, and a member of NYU Wireless. His research interests include the performance analysis and design of networks. His current research focuses on cross-layer research issues in wireless networks, and multimedia transport over networks. He has coauthored a textbook titled \emph{TCP/IP Essentials: A Lab based Approach} (Cambridge University Press). He was a winner of the IEEE Communication Society’s Leonard Abraham Prize for 2004, the ICC Best Paper Award in 2016, and the Sony Research Award. He was also co-awarded the Best Paper in 2011 Multimedia Communications Award. He has served as the Secretary for the Technical Affairs Council of the IEEE Communications Society.
\end{IEEEbiography}
\end{document}